\documentclass[11pt]{article}
\usepackage{amsmath,amssymb,amsfonts,latexsym,graphicx,amsthm}
\usepackage{fullpage,color}
\usepackage{url,hyperref}
\usepackage{comment}
\usepackage[linesnumbered,boxed,ruled,vlined]{algorithm2e}
\usepackage{framed}
\usepackage[shortlabels]{enumitem}
\usepackage{cleveref}
\usepackage{todonotes}
\usepackage[normalem]{ulem}
\usepackage{mathabx}
\usepackage{tikz}
\usetikzlibrary{calc}
\usetikzlibrary{arrows.meta}
\usetikzlibrary{backgrounds}
\usetikzlibrary{arrows.meta}

\pgfdeclarelayer{foreground}
\pgfsetlayers{main,foreground}

\usepackage{subcaption}
\usepackage[margin=1in]{geometry}
\usepackage{ifthen}
\usepackage{thm-restate}
\usepackage{xcolor}

\usepackage{stmaryrd}

\newtheorem{theorem}{Theorem}[section]
\newtheorem{lemma}{Lemma}[section]

\newtheorem{corollary}{Corollary}[section]
\newtheorem{claim}{Claim}[lemma]
\newtheorem{question}{Question}[section]

\newtheorem{observation}{Observation}[section]

        {\medskip}

\newcommand{\eps}{\epsilon}
\newcommand{\hide}[1]{}

\newcommand{\ceil}[1]{\left\lceil #1 \right\rceil}

\newcommand{\brac}[1]{\left(#1\right)}

\newcommand{\mst}{\mathsf{MST}}

\newcommand{\wts}{\mathbf{w}}

\newcommand{\dist}{\mathsf{dist}}

\newcommand{\anc}{\mathsf{anc}}
\newcommand{\net}{\mathsf{net}}

\newcommand{\p}{\mathbf{p}}

\definecolor{BrickRed}{rgb}{0.8, 0.25, 0.33}
\def\EMPH#1{\emph{\textcolor{BrickRed}{#1}}}

\begin{document}

	\title{Improved Euclidean Shallow Light Trees}

	\author{
		Hung Le\thanks{University of Massachusetts Amherst, \href{mailto:hungle@cs.umass.edu}{hungle@cs.umass.edu}}\and
		Shay Solomon\thanks{Tel Aviv University, \href{mailto:solo.shay@gmail.com}{solo.shay@gmail.com}}\and	
		Cuong Than\thanks{University of Massachusetts Amherst, \href{mailto:cthan@umass.edu}{cthan@umass.edu}}\and
		Csaba D. T\'oth\thanks{California State University Northridge and Tufts University, 	\href{mailto:csaba.toth@csun.edu}{csaba.toth@csun.edu}}\and
		Tianyi Zhang\thanks{State Key Laboratory for Novel Software Technology, Nanjing University, Email: \href{mailto:tianyiz25@nju.edu.cn}{\nolinkurl{tianyiz25@nju.edu.cn}}}
	}

\date{}

\maketitle

\thispagestyle{empty}

\begin{abstract}
	A {\em minimum spanning tree (MST)} and a {\em shortest paths tree (SPT)} are two of the most fundamental graph structures, capturing the two conflicting objectives of minimizing total weight and preserving distances from a designated root vertex, respectively. A natural goal is to obtain a single spanning tree that simultaneously approximates both objectives. For parameters $\alpha,\beta \ge 1$, a spanning tree $T$ of a weighted graph $G$ rooted at a designated vertex $r$ is called an $(\alpha,\beta)$-shallow-light tree (SLT) if (i) for every vertex $v$, $\dist_T(r,v) \le \alpha \cdot \dist_G(r,v)$ ({\em root-stretch} $\alpha$), and (ii) $\wts(T) \le \beta \cdot \wts(\mst)$ ({\em lightness} $\beta$). The pioneering work of Khuller, Raghavachari, and Young (SODA~1993) constructed  $(1+\epsilon, \tfrac{2}{\epsilon}+1)$-SLTs for general weighted graphs, and proved that this tradeoff between root-stretch and lightness is tight even for {\em series-parallel} graphs. {\bf They further asked whether even a slight improvement, namely reducing the lightness to $\tfrac{2-c}{\epsilon}$ for any constant $c>0$, is possible in the Euclidean plane}. Despite considerable attention, this question has remained open, and the tight tradeoff for general graphs has persisted as the state of the art across all nontrivial graph families.

	We resolve this longstanding question in the affirmative. Our result applies not only to the Euclidean plane, but to Euclidean spaces of \EMPH{arbitrary dimension}. Specifically, we show that every Euclidean instance admits an SLT with root-stretch $1+\eps$ and lightness at most $\left(\frac{5}{3} + o_\eps(1)\right) \cdot \frac{1}{\eps}$, thereby significantly improving upon the longstanding $2/\eps$ barrier. We complement this result with two lower bounds: (1) lightness  $(4/3-o_\epsilon(1))/\epsilon$ is needed in Euclidean high-dimensional instances; (2) lightness $(2-o_\epsilon(1))/\epsilon$ is needed in \EMPH{$\ell_1$-normed} high-dimensional spaces.
    
    As our second main result, we provide a construction of SLTs in \EMPH{the Euclidean plane}, with root stretch $1+\epsilon$ and lightness at most $\left(\frac{2\pi}{\sqrt{4\pi^2+1}}+o_\epsilon(1)\right)\frac{1}{\epsilon} \approx (0.987+o_\epsilon(1))\frac{1}{\epsilon}$. Notably, this reduces the leading $2/\eps$ term in the lightness bound by more than a factor of two, and comes quite close to the lower bound of $\left(\frac{2\pi}{2\pi+1} +o_\epsilon(1))\right) \cdot \frac{1}{\eps} \approx (0.862 +o_\epsilon(1))\frac{1}{\epsilon}$ by Elkin and Solomon (FOCS 2011).
We complement this result with a lower bound of $(1-o_\epsilon(1))/\epsilon$ on the lightness in $3$ dimensions,
thus establishing a separation between the $2$ and $3$ dimensional cases.    
 
\end{abstract}

\clearpage

\thispagestyle{empty}
\tableofcontents
\clearpage

\setcounter{page}{1}

\section{Introduction}

Two of the most fundamental tree structures in graph algorithms are minimum spanning trees (MST) and single-source shortest-path trees. A minimum spanning tree captures the global geometry of a graph: it connects all relevant vertices using the least possible total weight, but it may distort distances from any particular source by an arbitrarily large factor. In contrast, a single-source shortest-path tree captures the metric geometry from a fixed source: it preserves the shortest-path distance from the source to every vertex, but its total weight can be much larger than that of a minimum spanning tree. A natural and widely studied goal is therefore to design a tree that simultaneously resembles both objects: it should have total weight comparable to the minimum spanning tree, while approximately preserving distances from the source.

Formally speaking, let $G=(V,E,\wts)$ be a connected undirected graph with nonnegative edge weights, let $s\in V$ be a distinguished source vertex. We write $\dist_G(u,v)$ for the shortest-path distance between $u$ and $v$ in $G$. For parameters $\alpha,\beta\ge 1$, an $(\alpha,\beta)$ shallow-light tree (SLT) for $G$ is a spanning tree $T$ in $G$ such that
\begin{align*}
    \dist_T(s,t)&\le \alpha\cdot \dist_G(s,t), \,\,\, \forall t\in V\\
    \wts(T)&\le \beta\cdot \wts(\mst).
\end{align*}
The first condition is the \emph{root-stretch} condition, which requires the distance from the source to every terminal in $T$ to approximate its shortest-path distance in $G$. The second condition is the \emph{lightness} condition, which requires the total weight of $T$ to be within a factor $\beta$ of the MST.

The study of shallow-light trees originated in the distributed-computing works of Awerbuch, Baratz, and Peleg \cite{ABP90,ABP92}.
Motivated by a natural cost-sensitive communication model \cite{ABP90}, 
they first studied the diameter version of the problem and proved that 
every weighted graph admits a spanning tree of weight at most $2(1+1/\eps)\wts(\mst)$ whose diameter is at most $2/\eps+1$ times the graph diameter. In their subsequent work on efficient broadcast~\cite{ABP92}, they introduced the rooted version of SLTs, showing that every root admits a spanning tree with root-stretch at most $1+\eps$ and lightness at most $1+4/\eps$. Thus, these papers already obtained the same qualitative $1+\eps$ versus $O(1/\eps)$ shallow-light 
trade-off, though with weaker constants.

Independently of \cite{ABP92}, Khuller, Raghavachari, and Young~\cite{KRY95} proved that for every edge-weighted undirected graph and every parameter $\eps>0$, there exists a $(1+\eps,1+2/\eps)$-SLT. Motivated by applications in performance-driven routing and network design where the root-stretch is required to be arbitrarily close to $1$, 
and following much of the literature on shallow-light and low-light trees, we shall focus on the regime of $\eps=o(1)$~\cite{CKRSW91,CKRSW92,ChenY20,LiQCM021}. In this regime, the leading constant in front of the $1/\eps$ lightness term becomes the central quantity.

Khuller et al.\ also showed that, for general graphs, their trade-off is tight: even on \emph{series-parallel} graphs, namely graphs of treewidth $2$, root-stretch $1+\eps$ may force lightness at least $1+2/\eps$. Thus, even in a very restricted family of sparse graphs, the leading constant $2$ in the $2/\eps$ term of the lightness bound cannot be improved. In their Conclusions section, Khuller et al.~\cite{KRY95} asked whether a better trade-off is possible in Euclidean graphs, in particular in the Euclidean plane.

\begin{question}[\cite{KRY95}]\label{qst:main}
	Is there a constant $c>0$, such that for any parameter $\epsilon >0$, every 2-dimensional Euclidean instance admits a
	$\brac{1+\eps,\frac{2-c}{\eps}+O(1)}$-SLT? 
\end{question}

In the Conclusions section of \cite{KRY95}, the authors suggested a concrete route for improving the coefficient $2$ in 2-dimensional Euclidean instances: replacing the factor-$2$ MST traversal in their analysis by a shorter Traveling Salesman path. They wrote: ``\emph{In Euclidean graphs (i.e., points in the plane), the proof only requires that the algorithm walks around the graph from the root visiting every vertex once, i.e., that the algorithm traverse a Traveling Salesman path starting at the root. In Euclidean graphs, perhaps such a path of weight at most $(2-c)$ times the minimum spanning tree weight always exists and can be found in polynomial time}''.

Subsequently, \cite{fekete1997network} presented a Euclidean plane instance for which the ratio between the minimum-weight Traveling Salesman path and the MST weight is arbitrarily close to $2$. Consequently, this route cannot improve the lightness analysis of \cite{KRY95}: it is already tight, even in the Euclidean plane. In particular, any improvement over the constant $2$ in the $2/\eps$ term of the lightness bound must depart from the MST-traversal paradigm used in \cite{KRY95} and in \cite{ABP90,ABP91}.

Elkin and Solomon \cite{ElkinS15} showed that there exists a 2-dimensional Euclidean instance for which any $(1+\eps,\beta)$-SLT has lightness $\beta \ge \left(\frac{2\pi}{2\pi+1} +o_\epsilon(1))\right) \cdot \frac{1}{\eps} \approx (0.862 +o_\epsilon(1))\frac{1}{\epsilon}$. This instance, discussed in detail in \Cref{sec:overview_plane}, is simply a set of evenly spaced points around a circle boundary, with the circle center as the source.
Thus, prior to this work, the best known upper bound in the Euclidean plane
remained $2/\epsilon$, as in general graphs, whereas the best known lower bound was
$\approx0.862/\epsilon$.

\subsection{Our Results}
\subsubsection{Upper Bounds}

More than three decades after Khuller et al.~posed \Cref{qst:main}, we resolve it in the affirmative. 
Our result applies not only to the Euclidean plane, but to Euclidean spaces of \EMPH{arbitrary dimension},
which is perhaps surprising, especially given that the bound $2/\eps$ is tight even in the very structured family of series-parallel graphs. Specifically, 
we show that every Euclidean instance 
admits a $\brac{1+\eps,\frac{2-c}{\eps}+O(1)}$-SLT with $c = \frac{1}{3}$, thereby significantly improving upon the longstanding $2/\epsilon$ barrier of \cite{KRY95} in any Euclidean space:

\begin{theorem}\label{thm:high-dim-slt}
	For every 
    set $S$ of $n$ terminal points in Euclidean space of arbitrary dimension, every $s \in S$ and every sufficiently small $\eps > 0$, there exists a $\brac{1 + \eps, \frac{5}{3\eps} + O(1)}$-SLT rooted at $s$.
\end{theorem}

Furthermore, in the Euclidean plane---the setting explicitly asked about by Khuller et al.---we obtain a substantially stronger bound, reducing the leading $2/\eps$ term by more than a factor of two:
\begin{theorem}\label{thm:plane-slt}
	For every finite set $S\subset \mathbb{R}^2$ of $n$ terminal points, every $s\in S$, and every sufficiently small $\eps>0$, there is a
	$\bigl(1+\epsilon,(\frac{2\pi}{\sqrt{4\pi^2+1}}+o_\epsilon(1))\cdot \frac{1}{\eps}\bigr)$-SLT rooted at $s$, where the $o_{\eps}(1)$ term tends to $0$ as $\eps \rightarrow 0$.
\end{theorem}
The leading constant $\frac{2\pi}{\sqrt{4\pi^2+1}}\approx0.987$ in our upper bound comes quite close to the lower bound  $\frac{2\pi}{2\pi+1}\approx0.862$ of \cite{ElkinS15}. Bridging this remaining gap is left as an intriguing open problem.

Our geometric constructions admit polynomial time algorithms, 
but we do not focus on optimizing their running time.

\subsubsection{Lower Bounds}
Our upper bounds are complemented by the following lower bounds.

We  prove that a $\brac{1+\epsilon, (1 - c)/\epsilon + O(1)}$-SLT cannot be achieved in 3-dimensional Euclidean spaces for any constant $c>0$, providing a separation between the  $2$ and $3$ dimensional cases.

\begin{restatable}{theorem}{LowerBoundThree}\label{thm:LB-3D}
	For every sufficiently small $\eps>0$, there exists a finite set
	$P_\eps\subset \mathbb{R}^3$ and a point $s\in P_\eps$ such that every
	$(1+\eps,\beta)$-SLT rooted at $s$ for $P_\eps$ has lightness	$\beta\geq (1-o_{\eps}(1))\frac{1}{\eps}$.
\end{restatable}

We next show that even an $\brac{1+\epsilon, \frac{1}{\epsilon} + O(1)}$-SLT is not achievable in 4-dimensional spaces. 

\begin{restatable}{theorem}{LowerBoundFour}\label{thm:four-dim-lower-bound}
	For every sufficiently small \(\eps>0\), there exists a finite point set \(P_\eps\subset \mathbb{R}^4\) and a point \(s\in P_\eps\) such that every \((1+\eps,\beta)\)-SLT   rooted at $s$ for \(P_\eps\) has lightness
	$\beta >\left(12/11-o_\epsilon(1)\right)\frac{1}{\eps}$.
\end{restatable}

We then generalize our 4-dimensional lower bound for high dimensional Euclidean spaces, showing that SLTs must have lightness at least $\left(4/3-o_\epsilon(1)\right)\frac{1}{\eps}$ in high dimensions.

\begin{restatable}{theorem}{LowerBoundHigh}\label{thm:high-dim-lower-bound}
	For every sufficiently small \(\eps>0\) and sufficiently large integer $d$, there exists a finite point set \(P_\eps\subset \mathbb{R}^d\) and a point \(s\in P_\eps\) such that every \((1+\eps,\beta)\)-SLT  rooted at $s$ for \(P_\eps\) has lightness
	$\beta >\left(4/3-o_\epsilon(1)\right)\frac{1}{\eps}$.
\end{restatable}

Finally, we contrast our construction of
\Cref{thm:high-dim-slt}, which achieves lightness
$\approx\frac{5}{3\eps}$ in arbitrary Euclidean dimension,
with high-dimensional \EMPH{$\ell_1$-normed spaces}, where we show that every
SLT must have lightness at least
$\left(2-o_\eps(1)\right)\frac1\eps$.
This demonstrates that the improvement over the longstanding $2/\eps$ barrier crucially relies on Euclidean geometry and does not extend to
general normed spaces.

\begin{restatable}{theorem}{LowerBoundHighellone}\label{thm:high-dim-lower-bound-ell1}
For every sufficiently small $\eps>0$ and sufficiently large integer $d$ such that $\eps d<1$, there exists a finite point set
$P_\eps\subset (\mathbb{R}^d,\|\cdot\|_1)$
and a point $s\in P_\eps$ such that every
$(1+\eps,\beta)$-SLT rooted at $s$ for $P_\eps$ has lightness
$\beta>\left(2-o_\epsilon(1)\right)\frac{1}{\eps}$.
\end{restatable}

\subsection{Related Work}
Research on shallow-light trees has developed in several directions. We briefly review the work most closely related to ours.

Our work concerns non-Steiner shallow-light trees. However, Steiner variants of SLTs have also been studied extensively. Elkin and Solomon~\cite{ElkinS15} showed that allowing Steiner points can make SLTs substantially lighter than non-Steiner ones. In the Euclidean plane, Solomon~\cite{Solomon15} constructed Euclidean Steiner $(1+\eps,O(\sqrt{1/\eps}))$-SLTs and proved that this dependence is optimal up to constant factors. Very recently, an extension by Frost, Kokado, and Tóth~\cite{Frost2026EuclideanSS} proves the same $O(\sqrt{1/\eps})$ lightness bound for Euclidean Steiner SLTs in arbitrary dimension, independent of the dimension.

Another line of research concerns optimization versions of shallow-light and delay-constrained Steiner trees, often with separate cost and length functions. Kortsarz and Peleg~\cite{kortsarz1997approximating} studied approximation algorithms for diameter-bounded shallow-light Steiner trees, and Naor and Schieber~\cite{NaorS97} gave improved approximations for shallow-light spanning trees. Hajiaghayi, Kortsarz, and Salavatipour~\cite{Hajiaghayi2006ApproximatingBA} studied shallow-light $k$-Steiner trees and their connection to buy-at-bulk network design. Chimani and Spoerhase~\cite{ChimaniS15} used directed shallow-light Steiner trees as a tool for bounded-distance network design and directed light spanners. 
These works study approximation algorithms and computational complexity for optimization versions of shallow-light and delay-constrained Steiner trees; they are complementary to our focus on the \EMPH{existential} trade-off between root stretch and lightness.

\section{Technical Overview}

\subsection{Shallow Light Trees in Euclidean Spaces of Arbitrary Dimension}

\paragraph{The pruning framework.}
The arbitrary-dimensional construction is based on a very simple pruning intuition.
Root an MST \(T\) at the source \(s\).  For a rooted subtree \(T_v\), let
\(\p(v)\) be the parent of \(v\), let
\[
\widehat{\wts}(T_v)=\wts(T_v)+\wts(v,\p(v)),
\]
and let
\[
\anc(T_v)\in V(T_v)
\]
be a vertex of \(T_v\) closest to \(s\) in Euclidean metric.  Think of
\(\widehat{\wts}(T_v)\) as the amount of MST weight that disappears when
\(T_v\) is pruned from the current tree.

Let \(\lambda=\frac35=0.6\).  Suppose first that we could decompose the MST into rooted
subtrees \(P\) such that, for each piece,
\[
\widehat{\wts}(P)
\approx
\lambda\eps\cdot \wts(s,\anc(P)).
\]
Then the construction would be immediate: For every piece \(P\), add one
shortcut edge from \(s\) to the closest vertex \(\anc(P)\) in $P$, and keep all MST
edges inside \(P\).  For any terminal \(u\in V(P)\), the path that first uses
the shortcut \((s,\anc(P))\) and then follows the MST inside \(P\) has length at
most
\[
\wts(s,\anc(P))+\dist_P(\anc(P),u).
\]
Since \(\anc(P)\) is the closest vertex of \(P\) to \(s\), we have
\[
\wts(s,\anc(P))\le \wts(s,u).
\]
If the piece has weight at most \(\lambda\eps\cdot \wts(s,\anc(P))\), then
\[
\dist_P(\anc(P),u)
\le
\widehat{\wts}(P)
\le
\lambda\eps\cdot \wts(s,\anc(P))
\le
\lambda\eps\cdot \wts(s,u).
\]
Therefore
\[
\dist(s,u)
\le
\wts(s,\anc(P))+\dist_P(\anc(P),u)
\le
(1+\lambda\eps)\wts(s,u)
\le
(1+\eps)\wts(s,u).
\]
Thus a single shortcut to the closest vertex of each piece would be enough for the
root stretch.

The same condition also explains the lightness.  The shortcut added for \(P\)
has weight \(\wts(s,\anc(P))\), and the lower bound
\[
\widehat{\wts}(P)
\gtrsim
\lambda\eps\cdot \wts(s,\anc(P))
\]
gives
\[
\wts(s,\anc(P))
\lesssim
\frac{1}{\lambda\eps}\widehat{\wts}(P)
=
\frac{5}{3\eps}\widehat{\wts}(P).
\]
Since the pruned quantities \(\widehat{\wts}(P)\) telescope over the pruning
process, their total is at most the original MST weight.  Hence the total
shortcut weight would be at most
\[
\frac{1}{\lambda\eps}\wts(\mst)
=
\frac{5}{3\eps}\wts(\mst).
\]
This ideal calculation is the source of the constant \(5/3\).

\paragraph{Why the ideal decomposition is too optimistic.}
The difficulty is that such a perfectly balanced decomposition need not exist.
A natural attempt is to process the rooted MST bottom-up, following a post-order
traversal, and greedily prune the first subtree \(T_v\) for which
\[
\widehat{\wts}(T_v)
\ge
\lambda\eps\cdot \wts(s,\anc(T_v)).
\]
This rule guarantees the lower bound needed to pay for the shortcut
\((s, \anc(T_v))\).  However, it gives no comparable upper bound on
\(\widehat{\wts}(T_v)\).

The reason is that the transition from a child to its parent can be abrupt.
A child subtree \(T_x\) may still be small, 
\(
\widehat{\wts}(T_x)
<
\lambda\eps\cdot \wts(s,\anc(T_x)),
\)
but when we move to its parent \(v\), the subtree \(T_v\) suddenly contains the
parent edge and all child components of \(v\).  Thus it may happen that
\(
\widehat{\wts}(T_v)
\gg
\lambda\eps\cdot \wts(s,\anc(T_v)).
\)
In this case, \(T_v\) is good for charging, because it has already crossed the
threshold, but it may be bad for stretch.  Indeed, after taking the shortcut
from \(s\) to the closest vertex \(\anc(T_v)\), reaching a vertex in another
child component may require going through \(v\).  Intuitively, as we will explain shortly, this costs about
two local traversals of length roughly \(\lambda\eps\rho\), where
\(\rho=\wts(s,\anc(T_v))\).  Since \(\lambda>1/2\), such a
\(2\lambda\eps\rho\)-length detour is too large to fit within the
\(\eps\rho\) stretch budget.

The actual construction keeps the useful part of the post-order rule and repairs
precisely this failure.  We choose the first vertex \(v^\star\) satisfying the
threshold condition
\[
\widehat{\wts}(T_{v^\star})
\ge
\lambda\eps\cdot \wts(s,\anc(T_{v^\star})).
\]
Then every child subtree of \(v^\star\) fails the threshold.  Define the variables
\[
a=\anc(T_{v^\star})
\qquad\text{and}\qquad
\rho=\wts(s,a).
\]
If \(a\neq v^\star\), let \(T_x\) be the child subtree of \(v^\star\) that
contains \(a\).  Since \(T_x\) fails the threshold and \(a=\anc(T_x)\), the path
from \(v^\star\) to \(a\) is contained in \(T_x\) together with the parent edge
of \(x\), and hence
\[
\dist_T(v^\star,a)
\le
\widehat{\wts}(T_x)
<
\lambda\eps\cdot \wts(s,a)
=
\lambda\eps\rho .
\]
The same inequality is trivial if \(a=v^\star\).

Now consider any vertex \(u\in V(T_{v^\star})\).  If \(u\neq v^\star\), let
\(T_y\) be the child subtree of \(v^\star\) that contains \(u\).  Since \(T_y\)
also fails the threshold, we have
\[
\dist_T(v^\star,u)
\le
\widehat{\wts}(T_y)
<
\lambda\eps\cdot \wts(s,\anc(T_y))
\le
\lambda\eps\cdot \wts(s,u),
\]
where the last inequality uses the definition of \(\anc(T_y)\).  Thus, using
the bound on \(\dist_T(v^\star,a)\),
\[
\wts(s,u)
\le
\rho+\dist_T(a,v^\star)+\dist_T(v^\star,u)
\le
(1+\lambda\eps)\rho+\lambda\eps\cdot \wts(s,u),
\]
and therefore
\(
\wts(s,u)
\le
\frac{1+\lambda\eps}{1-\lambda\eps}\rho .
\)
Similarly, we have
\[
\wts(s,u)
\ge
\rho-\dist_T(a,v^\star)-\dist_T(v^\star,u)
\ge
(1-\lambda\eps)\rho-\lambda\eps\cdot \wts(s,u),
\]
so we have
\(
\wts(s,u)
\ge
\frac{1-\lambda\eps}{1+\lambda\eps}\rho .
\)
Consequently, all vertices of \(T_{v^\star}\) lie in a thin radial shell around
radius \(\rho\).  Moreover, every child component of \(v^\star\) can be
traversed from \(v^\star\) with local cost
\[
\dist_T(v^\star,u)
\le
\lambda\eps\cdot \wts(s,u)
\le
\frac{(1+\lambda\eps)\lambda\eps}{1-\lambda\eps}\rho .
\]
Thus the informal \(\lambda\eps\rho\) bound should be understood only up to the
factor \((1+\lambda\eps)/(1-\lambda\eps)=1+O(\eps)\).  In particular, the
generic estimate for going from \(a\) to another child component through
\(v^\star\) is bounded by
\[
\dist_T(a,v^\star)+\dist_T(v^\star,u)
\le
\lambda\eps\rho+
\frac{(1+\lambda\eps)\lambda\eps}{1-\lambda\eps}\rho
=
\frac{2\lambda\eps}{1-\lambda\eps}\rho
=
(2\lambda+O(\eps))\eps\rho .
\]
For any choice $\lambda>1/2$, this exceeds the available additive stretch budget. Indeed, all vertices of $T_{v^\star}$ lie in the same thin shell, so 
$\wts(s,u)=(1+O(\eps))\rho$. Hence the available additive stretch is only $(1+O(\eps))\eps\rho$, whereas the detour through $v^\star$ is $(2\lambda+O(\eps))\eps\rho$.
Since $2\lambda>1$, we conclude that a single shortcut to $a$ together with the MST edges of
$T_{v^\star}$ cannot guarantee $(1+\eps)$-stretch.

\paragraph{Exploiting collinearity.}

The main difficulty is how to reach vertices on other child subtrees of
\(v^\star\) which exclude \(\anc(T_{v^\star})\).
To begin with, we only need to worry about the net points in these subtrees.
Namely, choose a sparse
\(10^{-3}\eps\rho\)-net of \(T_{v^\star}\), and let \(N_{v^\star}\) be the union of these net points.  Every terminal in
those child components can first move to a nearby net point at cost only
\(10^{-3}\eps\rho\), so it remains to provide shallow paths from \(s\) to the
net points.

There are two cases.  If routing through the closest vertex
\(a=\anc(T_{v^\star})\) has enough slack for every net point
\(w\in N_{v^\star}\), namely
\[
\wts(s,a)+\wts(a,w)
\le
(1+\eps)\wts(s,w)-0.005\eps\rho,
\]
then we add the shortcut \((s,a)\) and a star from \(a\) to the net points.  The
extra \(0.005\eps\rho\) slack pays for the final move from a net point to the
terminal it represents.

Otherwise, some net point \(w\) witnesses that the route through \(a\) is almost
tight.  In that case, our key observation is that the three points
\(a,v^\star,w\) should be \emph{nearly collinear}.

Let us explain the meaning of this collinearity.  Since \(a\) is the closest
vertex of \(T_{v^\star}\) to \(s\), we have
\[
\wts(s,w)\ge \wts(s,a)=\rho .
\]
On the other hand, the failure of the previous case gives
\[
\rho+\wts(a,w)
>
(1+\eps)\wts(s,w)-0.005\eps\rho .
\]
Writing \(x=\wts(a,w)/\rho\), this implies
\[
\wts(s,w)
<
\left(1+\frac{x-0.995\eps}{1+\eps}\right)\rho .
\]
In particular \(x\ge 0.995\eps\).  Thus \(a\) and \(w\) are separated by
distance at least \(0.995\eps\rho\), but the route \(s\to a\to w\) is still
almost as short as the direct distance \(\wts(s,w)\).  Geometrically, this means
that the triangle \(saw\) is nearly degenerate: the points \(s,a,w\) are almost
on one line, with point \(a\) lying approximately on the line segment from \(s\) to \(w\). \EMPH{This is where we use (dimension-independent) Euclidean geometry}.

At the same time, the choice of \(v^\star\) implies that every child subtree of
\(v^\star\) is still below the pruning threshold.  Hence both \(a\) and \(w\)
are close to \(v^\star\):
\[
\wts(a,v^\star)\le \lambda\eps\rho,
\qquad
\wts(w,v^\star)
\le
\frac{(1+\lambda\eps)\lambda\eps}{1-\lambda\eps}\rho .
\]
So \(a\) and \(w\) form a base of length \(\Omega(\eps\rho)\), while
\(v^\star\) is within \(O(\eps\rho)\) of both endpoints.  This is the precise
sense in which \(a,v^\star,w\) are nearly collinear: the triangle
\(av^\star w\) is thin, and \(v^\star\) lies close to the line through \(a\) and
\(w\).

\EMPH{This collinearity gives a better estimate for \(\wts(s,v^\star)\) than the
trivial triangle inequality.}  The trivial bound is
\[
\wts(s,v^\star)
\le
\wts(s,a)+\wts(a,v^\star)
\le
(1+\lambda\eps)\rho .
\]
This is not strong enough, because after shortcutting to \(v^\star\) we still
have to pay for the edge from \(v^\star\) to a net point.  The near-collinearity
saves a constant fraction of the \(\eps\rho\) slack.

Indeed, let \(m\) be the midpoint of the segment \(aw\).  The near-tightness of
the route through \(a\), together with Apollonius' theorem in the triangle
\(saw\), shows that \(m\) is closer to \(s\) than the triangle inequality would
suggest:
\[
\wts(s,m)
\le
\left(
1+\frac{x-0.995\eps}{2(1+\eps)}
\right)\rho .
\]
Similarly, applying Apollonius' theorem in the triangle \(av^\star w\) gives an
upper bound on the distance from \(v^\star\) to the same midpoint:
\[
\wts(v^\star,m)
\le
\sqrt{
	\frac{(1+\lambda^2\eps^2)\lambda^2\eps^2}{(1-\lambda\eps)^2}
	-\frac{x^2}{4}
}\cdot \rho .
\]
Combining the two bounds through \(m\), and then optimizing over the possible
values of \(x\), gives
\[
\wts(s,v^\star)
\le
\wts(s,m)+\wts(m,v^\star)
\le
\left(1+\eps(1.424\lambda-0.4925)\right)\rho .
\]
For our choice \(\lambda=0.6\), the coefficient of \(\eps\) is
\[
1.424\lambda-0.4925=0.3619,
\]
which is substantially better than the trivial coefficient \(\lambda=0.6\).

We therefore use \(v^\star\) as the shortcut center in this second case: add the
edge \((s,v^\star)\) and a star from \(v^\star\) to the net points in
\(N_{v^\star}\).  For any terminal \(u\) represented by a closest net point
\(w\in N_{v^\star}\), the resulting route satisfies
\[
\begin{aligned}
	\dist_{G(T_{v^\star})}(s,u)
	&\le
	\wts(s,v^\star)+\wts(v^\star,w)+10^{-3}\eps\rho  \\
	&\le
	\left(1+\eps(1.424\lambda-0.4925)\right)\rho
	+
	\frac{(1+\lambda\eps)\lambda\eps}{1-\lambda\eps}\rho
	+
	10^{-3}\eps\rho  \\
	&\le
	\left(1+\eps(2.427\lambda-0.4915)\right)\rho .
\end{aligned}
\]
Since \(\lambda=0.6\), we have
\[
2.427\lambda-0.4915 < 1.
\]
Thus
\[
\dist_{G(T_{v^\star})}(s,u)
\le
(1+\eps)\rho
\le
(1+\eps)\wts(s,u),
\]
where the last inequality uses the fact that \(a=\anc(T_{v^\star})\) is the
closest vertex of \(T_{v^\star}\) to \(s\).

Finally, the additional star edges only contribute lower-order weight.  The net
spacing is \(10^{-3}\eps\rho\), so
\[
|N_{v^\star}|
\le
\left\lceil
\frac{\widehat{\wts}(T_{v^\star})}{10^{-3}\eps\rho}
\right\rceil .
\]
Each star edge has length \(O(\eps\rho)\), because every child component of
\(v^\star\) is below the pruning threshold.  Therefore the total weight of the
star is \(O(1)\cdot\widehat{\wts}(T_{v^\star})\).  The single root shortcut is
charged using the cutoff condition
\[
\rho=\wts(s,\anc(T_{v^\star}))
\le
\frac{1}{\lambda\eps}\widehat{\wts}(T_{v^\star}) .
\]
Hence each pruning step adds weight at most
\[
\left(\frac{1}{\lambda\eps}+O(1)\right)
\widehat{\wts}(T_{v^\star}) .
\]
Summing over all pruning steps and using the telescoping of
\(\widehat{\wts}(T_{v^\star})\), the total weight is
\[
\left(\frac{1}{\lambda\eps}+O(1)\right)\wts(\mst)
=
\left(\frac{5}{3\eps}+O(1)\right)\wts(\mst),
\]
where the last equality uses \(\lambda=0.6\).

\subsection{Shallow Light Trees in the Euclidean Plane} \label{sec:overview_plane}
\paragraph{A hard example.}
We begin with the simplest obstruction one should keep in mind; this lower bound is due to \cite{Elkin2011SteinerST}.  Put the source
\(s\) at the center of a unit circle, and place
\[
N=\left\lfloor \frac{2\pi}{(1+\zeta)\eps}\right\rfloor
\]
terminals \(p_1,\ldots,p_N\) evenly around the circle, where \(\zeta>0\) is an
arbitrarily small constant.  Thus every terminal is at Euclidean distance \(1\)
from the source, and the distance between consecutive terminals on the circle is
\[
2\sin\!\left(\frac{\pi}{N}\right)
=
(1+\zeta+o_\eps(1))\eps
>
\eps
\]
for sufficiently small \(\eps\).

Now consider any tree \(T\) rooted at \(s\) with root-stretch at most $1+\eps$.  For every
terminal \(p_i\), the root-to-terminal distance in \(T\) must be at most
\[
(1+\eps)\wts(s,p_i)=1+\eps .
\]
We claim that the edge \((s,p_i)\) is forced.  Indeed, if the path from \(s\) to
\(p_i\) first goes from \(s\) to some other boundary terminal \(p_j\), then that
path has length at least
\[
\wts(s,p_j)+\wts(p_j,p_i)
>
1+\eps ,
\]
because \(\wts(s,p_j)=1\) and every two distinct boundary terminals are separated
by more than \(\eps\).  This violates the required stretch.  Hence every  spoke (edge)
\((s,p_i)\) must belong to \(T\), and therefore
\[
\wts(T)\ge N
=
\left(\frac{2\pi}{1+\zeta}-o_\eps(1)\right)\frac{1}{\eps}.
\]

On the other hand, the MST is essentially the boundary cycle, with one boundary
edge deleted, together with one spoke from \(s\) to the circle.  Its weight is
\[
\wts(\mst)
\le
1+(N-1)\cdot 2\sin\!\left(\frac{\pi}{N}\right)
< 1+2\pi.
\]
Consequently every \((1+\eps)\)-shallow tree on this instance has lightness at
least
\[
\frac{\wts(T)}{\wts(\mst)}
\ge
\left(
\frac{2\pi}{(1+\zeta)(2\pi+1)}
-o_\eps(1)
\right)\frac{1}{\eps}.
\]
Since \(\zeta>0\) is an arbitrarily small constant, this gives the asymptotic lower bound
$\frac{2\pi}{2\pi+1}\cdot \frac{1}{\eps}$.

\paragraph{Cone partition and radial paths.}
This example was the main benchmark for us.  It says that the right obstruction
is not the depth-first traversal of the MST.  The MST pays once for going around
the circle and once for reaching the circle, while a shallow tree must buy
roughly one radial spoke per \(\eps\) radians.  Thus the natural target is a
charging scheme in which the shortcut cost is paid for by angular progress of
the MST around the root, rather than by a doubled MST traversal.

The circle example suggests the basic shape of the construction.  In that
example, the hard part is angular: at radius \(1\), a terminal that is
\(\eps\) radians away from the nearest available spoke already cannot be reached
within stretch \(1+\eps\).  Thus a shallow tree has to buy roughly one radial
shortcut for every \(\eps\) radians.  This led us to the following guiding
principle.

\medskip
\noindent
\emph{Instead of shortcutting a traversal of the MST, partition the plane into
	cones of angle 
    about \(\eps\), and inside each cone build a nearly radial
	shortcut path.}

\medskip

If all terminals in one cone lay on a single ray from the source, then the right
thing to do would be obvious: connect them in increasing order of distance from
\(s\).  The path from \(s\) to any terminal would then be exactly shortest. For a general instance, the right analogue is to choose one representative at
each relevant distance scale, and connect all the representatives as a radial path. When the cones have small angles $\theta \approx \epsilon$, we hope that the radial path is close to a straight ray. The key point is that we can pay a sideway deviation for each scale, and if the scales are geometrically increasing, it will sum to about $\theta r$.
We therefore partition the space into annuli whose
radii grow geometrically, say by a factor \(1+\alpha\), and in each nonempty
intersection of a cone and an annulus we choose an \emph{anchor}: the closest
terminal to the source in that region.  The anchors of a fixed cone are then
connected in increasing order of scale, forming a zig--zag path that stays inside
the cone and moves steadily away from \(s\).

This is the first main idea of the algorithm.  The annuli make the construction
scale-free: an error of size \(\eps r\) is affordable for terminals at radius
\(r\), and the geometric levels let us measure all errors relative to the correct
radius.  The cones make the construction directional: \EMPH{the shortcut path is
allowed to ignore the detailed shape of the MST and instead follows the radial
geometry suggested by the lower bound}.

A small but important technical issue appears at this point.  Consecutive
anchors on the zig--zag path cannot be allowed to come from neighboring annular
regions. To see why, suppose two anchors \(x\) and \(y\) lie in the same cone,
with \(\wts(s,x)=\rho\le R=\wts(s,y)\).  If their radial separation
\(R-\rho\) is a constant fraction of the current scale, then
\[
\wts(x,y)
=
(R-\rho)
+ \text{small angular error}.
\]
Thus the edge \(xy\) behaves almost like pure radial progress.  But if \(x\) and
\(y\) come from neighboring annular regions, their radial
separation may be much smaller than the scale \(R\).  Then the same edge may
spend length moving sideways across the cone while making almost no radial
progress.  Repeating this many times would consume the entire stretch budget. Check \Cref{fig:overview-nonadjacent-anchors} for an illustration.
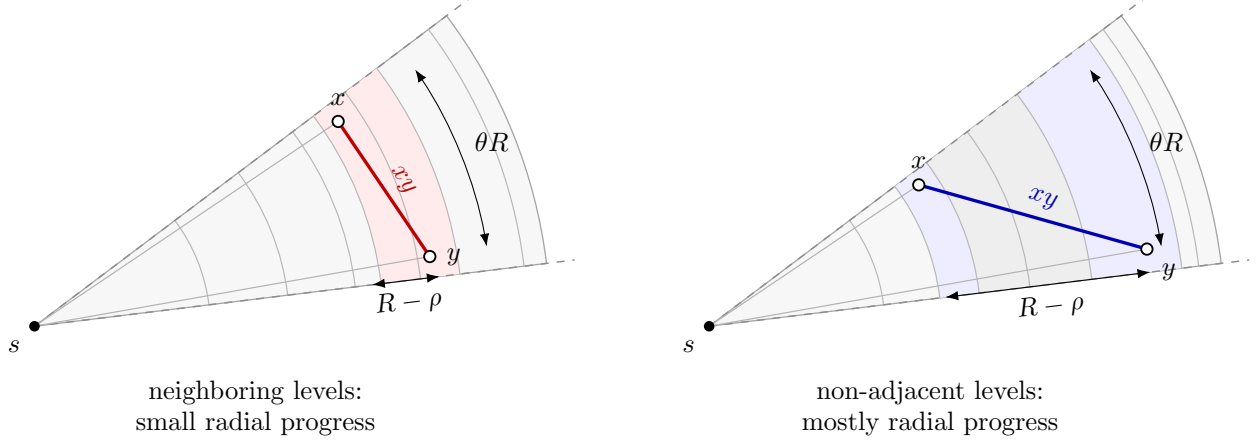
\begin{figure}
	\centering
	\begin{tikzpicture}[
    scale=1.5,
    >=Latex,
    every node/.style={font=\small},
    cone/.style={fill=gray!6, draw=black!35, line width=.5pt},
    ray/.style={draw=black!45, dashed, line width=.55pt},
    annulus boundary/.style={draw=black!28, line width=.45pt},
    bad band/.style={fill=red!8, draw=none},
    good band/.style={fill=blue!7, draw=none},
    skipped band/.style={fill=gray!13, draw=none},
    bad edge/.style={line width=1.25pt, red!75!black},
    good edge/.style={line width=1.25pt, blue!70!black},
    guide/.style={draw=black!30, line width=.45pt},
    measure/.style={<->, thin},
    anchornode/.style={circle, draw=black, fill=white, line width=.7pt,
        minimum size=1.5mm, inner sep=0pt},
    source/.style={circle, fill=black, inner sep=1.4pt}
]

\def\angA{7}
\def\angB{37}
\def\Rmax{4.55}

\begin{scope}
    \coordinate (sL) at (0,0);

    \path[cone]
        (sL) -- (\angA:\Rmax)
        arc[start angle=\angA, end angle=\angB, radius=\Rmax]
        -- cycle;

    \path[bad band]
        (\angA:3.08) arc[start angle=\angA, end angle=\angB, radius=3.08]
        -- (\angB:3.43) arc[start angle=\angB, end angle=\angA, radius=3.43]
        -- cycle;
    \path[bad band]
        (\angA:3.43) arc[start angle=\angA, end angle=\angB, radius=3.43]
        -- (\angB:3.78) arc[start angle=\angB, end angle=\angA, radius=3.78]
        -- cycle;

    \foreach \r in {1.55,2.25,2.85,3.08,3.43,3.78,4.35} {
        \draw[annulus boundary]
        (\angA:\r) arc[start angle=\angA, end angle=\angB, radius=\r];
    }
    \draw[ray] (sL) -- (\angA:\Rmax+.25);
    \draw[ray] (sL) -- (\angB:\Rmax+.25);

    \coordinate (xL) at (34:3.23);
    \coordinate (yL) at (10:3.54);
    \draw[guide] (sL) -- (xL);
    \draw[guide] (sL) -- (yL);
    \draw[bad edge] (xL) -- node[pos=.52, above, sloped] {$xy$} (yL);

    \node[source, label=below left:{$s$}] at (sL) {};
    \node[anchornode, label=above:{$x$}] at (xL) {};
    \node[anchornode, label=right:{$y$}] at (yL) {};

    \draw[measure] (\angA:3) -- node[below, sloped] {$R-\rho$} (\angA:3.6);
    \draw[measure] (10:4.05) arc[start angle=10, end angle=34, radius=4.05];
    \node at (22:4.35) {$\theta R$};

    \node[align=center] at (1.95,-.72)
        {neighboring levels:\\small radial progress};
\end{scope}

\begin{scope}[shift={(5.95,0)}]
    \coordinate (sR) at (0,0);

    \path[cone]
        (sR) -- (\angA:\Rmax)
        arc[start angle=\angA, end angle=\angB, radius=\Rmax]
        -- cycle;

    \path[good band]
        (\angA:2.05) arc[start angle=\angA, end angle=\angB, radius=2.05]
        -- (\angB:2.40) arc[start angle=\angB, end angle=\angA, radius=2.40]
        -- cycle;
    \path[skipped band]
        (\angA:2.40) arc[start angle=\angA, end angle=\angB, radius=2.40]
        -- (\angB:3.40) arc[start angle=\angB, end angle=\angA, radius=3.40]
        -- cycle;
    \path[good band]
        (\angA:3.40) arc[start angle=\angA, end angle=\angB, radius=3.40]
        -- (\angB:4.20) arc[start angle=\angB, end angle=\angA, radius=4.20]
        -- cycle;

    \foreach \r in {1.55,2.05,2.40,2.85,3.40,4.20,4.35} {
        \draw[annulus boundary]
        (\angA:\r) arc[start angle=\angA, end angle=\angB, radius=\r];
    }
    \draw[ray] (sR) -- (\angA:\Rmax+.25);
    \draw[ray] (sR) -- (\angB:\Rmax+.25);

    \coordinate (xR) at (34:2.23);
    \coordinate (yR) at (10:3.92);
    \draw[guide] (sR) -- (xR);
    \draw[guide] (sR) -- (yR);
    \draw[good edge] (xR) -- node[pos=.52, above, sloped] {$xy$} (yR);

    \node[source, label=below left:{$s$}] at (sR) {};
    \node[anchornode, label=above:{$x$}] at (xR) {};
    \node[anchornode, label=below right:{$y$}] at (yR) {};

    \draw[measure] (\angA:2.1) -- node[below, sloped] {$R-\rho$} (\angA:3.92);
    \draw[measure] (10:4.05) arc[start angle=10, end angle=34, radius=4.05];
    \node at (22:4.35) {$\theta R$};

    \node[align=center] at (1.95,-.72)
        {non-adjacent levels:\\mostly radial progress};
\end{scope}

\end{tikzpicture}
	\caption{If we connect two anchor points in two neighboring annulli, then the radial path may zig-zag and accumulate too much error.}
	\label{fig:overview-nonadjacent-anchors}
\end{figure}

For this reason the radial path uses only a maximal subsequence of
\emph{non-adjacent} anchor levels.  In each cone, we list all annular levels that
contain an anchor and greedily keep levels with a gap of at least one level
between consecutive kept anchors.  The kept anchors form the radial zig--zag
path.  Every skipped level is adjacent to a kept level, so it is still close
enough in scale to be served by a local shortcut to the corresponding kept
anchor.

This is the conceptual role of the annular decomposition.  It is not merely a
bookkeeping device for the proof; it is what allows the shortcut path in each
cone to look radial at every scale.  The construction buys one outward-moving
path per cone, as the lower bound suggests, but it samples that path only at
well-separated distance scales so that each shortcut edge pays mostly for radial
progress rather than for uncontrolled angular motion.

\paragraph{Charging radial paths to MST.}
The next question is how to pay for the radial paths.  The lower-bound example
already hints at the right comparison: a radial shortcut inside a cone should be
charged to the amount by which the MST travels angularly through that cone.  In
the circle instance this is exactly what happens.  The MST runs around the
circle, crossing every cone once, while the shallow tree buys one spoke per
cone.  Thus the shortcut cost is naturally compared against the MST length
inside the same angular sector.

Let us first ignore all boundary effects and look at the ideal calculation.
Fix a cone \(C\) of angle
\[
\theta \approx \eps
\]
and suppose that the construction buys one radial shortcut in \(C\) out to
distance \(R\) from the source.  This shortcut has length about \(R\).  If the
MST crosses the same cone at radius about \(R\), then it must travel from one
side of the cone to the other.  The Euclidean distance between the two boundary
rays at radius \(R\) is at least 
\[
2R\sin(\theta/2)
=
(1-o_\eps(1))\theta R .
\]
Thus a single MST crossing through this cone has length at least roughly
\(\theta R\), while the radial shortcut we want to buy has length roughly \(R\).
Equivalently,
\[
\text{radial shortcut cost}
\;\approx\;
R
\;\approx\;
\frac{1}{\theta}\cdot \theta R
\;\lesssim\;
\frac{1}{\eps}\cdot
\text{MST crossing cost}.
\]
This is the basic reason one should expect lightness on the order of
\(1/\eps\).  If every radial shortcut could be matched to a nearly disjoint MST
crossing in the same cone and at the same scale, then summation of the above estimate
would give
\[
\sum_{\text{radial shortcuts }P}\wts(P)
\;\lesssim\;
\frac{1}{\theta}
\sum_{\text{matched crossings }Q}\wts(Q)
\;\lesssim\;
\frac{1}{\eps}\wts(\mst).
\]
The rest of the construction is designed to make this naive charging picture
true up to lower-order losses. Afterwards, we will explain further how to bypass the $1/\eps$ bound. 

For a general MST, however, the dream statement behind this calculation is not
automatic:  We would like to say that whenever a cone \(C\) has a terminal anchor (an anchor is called \emph{terminal} if it is the last one on a radial path)
in an annular region \(A\), the MST component containing that anchor crosses
\(C\cap A\) from one side of the cone to the other.  If this were always true,
then the above \(R\) versus \(\theta R\) calculation would immediately give the
right \(1/\eps\) scale.

This dream statement is false for disjoint annuli.  The reason is that cutting
the MST by annular boundaries can shatter it into many small components.  A
terminal anchor (which is the last anchor on a radial path) may lie in a tiny piece of \(\mst\cap A\) that enters and exits
through the radial boundaries of the annulus, or that barely touches the annulus
near one of its boundaries.  Such a piece need not cross the cone at all.  It may
contain the anchor, but it carries no useful angular motion against which the
radial shortcut can be charged.

This is the point at which the construction departs from a literal annular
partition.  Instead of using disjoint annuli, we use overlapping annular regions.
Each region has a central part and two thin buffer zones.  A component of
\(\mst\) restricted to such a region is processed only if it reaches the central
part; components that live entirely inside the buffers are ignored at that scale.
The buffers are chosen sparsely, so the MST length counted in overlaps is only a
small extra fraction of \(\wts(\mst)\).

The overlap has two conceptual roles.  First, it prevents terminals near an
annular boundary from being trapped in insignificant boundary fragments.  If a
terminal lies in a buffer, then it also lies in one of the neighboring
overlapping regions, and the path in the MST from that terminal toward the root
must leave the buffer on one side or the other.  Thus the terminal is assigned to
a region in which its component reaches the non-buffer part.  In this sense every
terminal is seen at a scale where the relevant MST component is not just a
boundary artifact.

Second, the overlap gives every processed component a definite amount of radial
thickness.  A component that touches the central part of an overlapping annulus
and also connects to the rest of the MST must cross one of the buffer zones.  So
it contains length proportional to the buffer width. This lower bound is what allows us to treat
small exceptional shortcut costs as lower-order: even if a component owns only a
few terminal anchors, the total cost of the corresponding radial paths can be
charged to the fact that the component itself has nontrivial radial size.

After this modification, the charging picture becomes much closer to the circle
example.  A terminal anchor is owned by a connected component of the MST inside
one overlapping annular region.  If this component owns terminal anchors in many
cones, then every non-extreme cone between them is crossed by the component:
there is an anchor on one angular side, another anchor on the other angular
side, and the tree path between the two must pass through the cone. Those
non-extreme cones are precisely the ones whose radial paths we keep as
\emph{active}. Refer to \Cref{fig:overview-active-cone-crossing} for an illustration. 

\begin{figure}
	\centering
	\begin{tikzpicture}[
    scale=1.5,
    >=Latex,
    every node/.style={font=\small},
    annulus/.style={fill=gray!5, draw=none},
    buffer/.style={fill=gray!14, draw=none},
    active cone/.style={fill=blue!8, draw=none},
    boundary/.style={draw=black!35, line width=.5pt},
    cone boundary/.style={draw=black!38, dashed, line width=.55pt},
    mst/.style={line width=1.05pt, black!70},
    crossing/.style={line width=1.45pt, blue!70!black},
    radial/.style={line width=1.25pt, red!75!black},
    anchornode/.style={circle, draw=black, fill=white, line width=.7pt,
        minimum size=1.5mm, inner sep=0pt},
    activeanchornode/.style={circle, draw=blue!70!black, fill=white, line width=.9pt,
        minimum size=1.5mm, inner sep=0pt},
    inactiveanchornode/.style={circle, draw=black!45, fill=gray!14, line width=.6pt,
        minimum size=1.5mm, inner sep=0pt},
    crossingpoint/.style={circle, fill=blue!70!black, inner sep=1.35pt},
    source/.style={circle, fill=black, inner sep=1.4pt}
]

\def\angMin{15}
\def\angMax{85}
\def\inner{2.25}
\def\bin{2.75}
\def\bout{4.65}
\def\outer{5.15}
\def\actA{43}
\def\actB{57}

\path[annulus]
    (\angMin:\inner) arc[start angle=\angMin, end angle=\angMax, radius=\inner]
    -- (\angMax:\outer) arc[start angle=\angMax, end angle=\angMin, radius=\outer]
    -- cycle;

\path[buffer]
    (\angMin:\inner) arc[start angle=\angMin, end angle=\angMax, radius=\inner]
    -- (\angMax:\bin) arc[start angle=\angMax, end angle=\angMin, radius=\bin]
    -- cycle;
\path[buffer]
    (\angMin:\bout) arc[start angle=\angMin, end angle=\angMax, radius=\bout]
    -- (\angMax:\outer) arc[start angle=\angMax, end angle=\angMin, radius=\outer]
    -- cycle;

\path[active cone]
    (\actA:\inner) arc[start angle=\actA, end angle=\actB, radius=\inner]
    -- (\actB:\outer) arc[start angle=\actB, end angle=\actA, radius=\outer]
    -- cycle;

\foreach \r in {\inner,\bin,\bout,\outer} {
    \draw[boundary]
    (\angMin:\r) arc[start angle=\angMin, end angle=\angMax, radius=\r];
}
\foreach \a in {15,29,43,57,71,85} {
    \draw[cone boundary] (\a:\inner-.15) -- (\a:\outer+.20);
}

\node at (88:3.80) {$A_h$};
\node[gray!65] at (22:2.48) {buffer};
\node[gray!65] at (78:4.92) {buffer};
\node[blue!70!black] at (50:5.48) {$C_r$};

\coordinate (aone) at (20:3.75);
\coordinate (aL)   at (34:4.05);
\coordinate (u)    at (43:3.55);
\coordinate (z1)   at (47:3.15);
\coordinate (z2)   at (53:4.25);
\coordinate (v)    at (57:3.80);
\coordinate (aR)   at (66:4.35);
\coordinate (am)   at (80:3.55);
\coordinate (br)   at (52:3.79);
\coordinate (aMid) at (50:3.3);
\coordinate (s)    at (0,0);

\draw[radial]  (s) -- (aMid);
\node[red!75!black] at (65:1.75) {$P(C_r)$};

\draw[mst]
    (aone) .. controls (23:4.35) and (29:3.55) .. (aL)
    .. controls (38:4.45) and (39:3.15) .. (u);
\draw[crossing]
    (u) .. controls (z1) and (z2) .. (v);
\draw[mst]
    (v) .. controls (61:3.15) and (62:4.95) .. (aR)
    .. controls (72:4.45) and (77:3.15) .. (am);
\draw[mst] (br) -- (aMid);

\node[black!70] at (82:3.10) {$T$};

\node[source, label=below left:{$s$}] at (s) {};
\node[inactiveanchornode] at (aone) {};
\node[anchornode, label=right:{$a_{r-1}$}] at (aL) {};
\node[activeanchornode, label=left:{$a_r$}] at (aMid) {};
\node[anchornode, label=above right:{$a_{r+1}$}] at (aR) {};
\node[inactiveanchornode] at (am) {};

\node[crossingpoint, label=below:{$u$}] at (u) {};
\node[crossingpoint, label=above: {$v$}] at (v) {};

\end{tikzpicture}
	\caption{In this picture, the gray curves are $\mst$ edges, and $a_{r-1}, a_r, a_{r+1}$ are terminal anchors on a connected component of $\mst$ in region $A_h$, and $P(C_r)$ is an active radial path from $s$ to $a_r$. As the tree path $\mst[u, v]$ crosses the cone $C_r$, we can charge the cost of path $P(C_r)$ in the SLT to tree path $\mst[u, v]$.}
	\label{fig:overview-active-cone-crossing}
\end{figure}
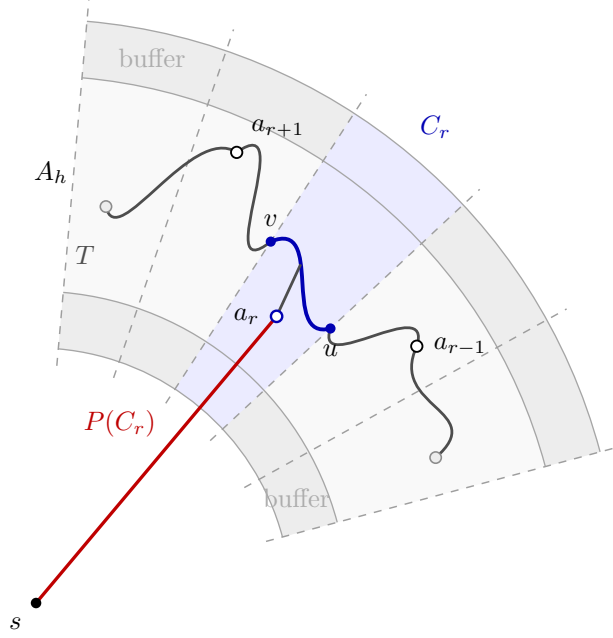

For an active cone, the earlier calculation can now be applied.  The active
radial path reaches some radius \(R\), so its cost is essentially \(R\).  The
owning MST component crosses the cone at the same scale, and this crossing has
angular width about \(\theta R\).  Hence
\[
\wts(\text{active radial path})
\;\lesssim\;
\frac{1}{\theta}\,
\wts(\text{corresponding MST crossing})
\;\approx\;
\frac{1}{\eps}\,
\wts(\text{corresponding MST crossing}).
\]
Because the processed annular regions overlap only sparsely, the total MST
length available for these crossings is still only \((1+o_\eps(1))\wts(\mst)\).
Therefore the active part of the shortcut system has total weight
\[
\sum_{\text{active paths }P}\wts(P)
\;\lesssim\;
\frac{1+o_\eps(1)}{\eps}\wts(\mst).
\]

The remaining cases are deliberately thrown away from the main charging
argument.  If a component owns only a few terminal anchors, all of its radial
paths are declared inactive.  If it owns many terminal anchors, only the two
extreme radial paths are declared inactive, since the component need not cross
the cones at its angular ends.  The inactive paths are affordable because each
processed component has the radial-thickness lower bound supplied by the
overlapping annuli.  They contribute only lower-order weight compared to
\(\wts(\mst)/\eps\).

Thus the overlapping annuli are not a technical nuisance but the mechanism that
makes the lower-bound intuition usable.  They let us ignore boundary fragments,
assign every terminal to a meaningful scale, and isolate the radial paths whose
cost can truly be charged to MST crossings through cones.  Once this is achieved,
the simple ratio
\[
\frac{\text{radial scale }R}{\text{angular crossing scale }\eps R}
\approx
\frac{1}{\eps}
\]
explains why the construction has lightness at most $\frac{1+o_\eps(1)}{\eps}$. 

\paragraph{Bypassing $1/\epsilon$ lightness.}
The charging scheme above explains why the lightness should be bounded by around $1/\epsilon$, 
but it only gives a coefficient close to \(1\).  Indeed, in the ideal
situation, a radial path of length \(R\) is charged to an MST crossing of length
about \(\theta R\), where \(\theta\approx\eps\).  This gives
\[
R
\approx
\frac{1}{\theta}\cdot \theta R
\approx
\frac{1}{\eps}\cdot \text{MST crossing length}.
\]
So the naive crossing charge has no room to beat \(1/\eps\).

The lower-bound example tells us what this calculation is missing.  On the unit
circle, the MST does not only pay for going around the circle.  It also pays one
extra unit of length to connect the source to the circle.  Thus the shallow tree
pays about \(2\pi/\eps\), but the MST pays about
\(
2\pi+1,
\)
not merely \(2\pi\).  This is why the lower-bound coefficient is
\(
\frac{2\pi}{2\pi+1}
\)
rather than \(1\).  To get a coefficient below \(1\) in the upper bound, we need
to recover an analogue of this extra source-to-farthest-point cost for an
arbitrary MST.

Let
\(
L=\max_{t\in S}\wts(s,t)
\)
be the distance from the source to the farthest terminal.  Any MST contains a
path from \(s\) to such a farthest terminal, and along this path the distance
from \(s\) increases from \(0\) to \(L\).  Thus the MST has one unit of radial
progress at scale \(L\), in the same way that the circle example has the single
spoke from the source to the circle.  The difficulty is that, in a general
instance, this radial progress need not be carried by edges disjoint from the
cone crossings.  A single slanted piece of the MST can simultaneously move
across a cone and move outward.  Therefore we cannot simply add an extra \(L\)
to the crossing cost.

The way to extract the missing contribution is to use the {\em Pythagorean theorem}.
For each active cone crossing, look at coordinates adapted to the cone: one
coordinate is the transverse (angular) direction across the cone, and the other is the
radial direction away from the source.  The naive charge uses only the
transverse width of the crossing.  If the same MST path also makes radial
progress \(A\), while its transverse fluctuation is \(B\), then the Euclidean
length of the relevant subpath is at least
\[
\sqrt{A^2+B^2}.
\]
After subtracting the transverse part already used by the naive crossing charge,
there remains an additional contribution of size
\[
\sqrt{A^2+B^2}-B .
\]
This is the quantitative substitute for the extra source-to-cycle edge in the
lower-bound example.

Summing this estimate over all active crossings gives the following heuristic
calculation.  Let \(W\) be the total transverse width of the active cone
crossings.  This is the quantity that pays for the radial shortcut paths:
\[
\sum_{\text{active radial paths }P}\wts(P)
\approx
\frac{W}{\eps}.
\]
The total transverse fluctuation $W$ over all cones is at most the circumference of
the circle of radius \(L\), namely \(2\pi L\).  On the other hand, the
active crossings, together with the unused part of the MST, must account for
radial progress from \(0\) to \(L\). 
Hence the Pythagorean estimate yields an extra MST contribution of
$\sqrt{W^2+L^2}-W$.
Thus the MST is not merely paying $W$, it is paying at least $\sqrt{W^2+L^2}$.
Since $W\le 2\pi L$, we obtain $$\frac{W}{\wts(\mst)}
\le
\frac{W}{\sqrt{W^2+L^2}}
\le
\frac{2\pi L}{\sqrt{(2\pi L)^2+L^2}}
=
\frac{2\pi}{\sqrt{4\pi^2+1}}
=
0.987570492\ldots.$$

This is the source of the constant in \Cref{thm:plane-slt}.  The active radial paths have
total weight approximately
\(
\frac{W}{\eps},
\)
but the MST weight is larger than \(W\) by the Pythagorean radial-progress
excess.  Consequently the active part contributes at most
\[
\left(
\frac{2\pi}{\sqrt{4\pi^2+1}}
+o_\epsilon(1)
\right)
\frac{\wts(\mst)}{\eps}.
\]
The inactive shortcuts and the overlap losses are lower-order with the
$\epsilon$-dependent parameter choice, so the same asymptotic constant survives
in the final bound.

\subsection{Lower Bounds}
We next highlight the key ideas behind our lower bounds.
The starting point is the aforementioned lower bound of Elkin and Solomon~\cite{ElkinS15},
which already captures the fundamental obstruction in the Euclidean plane.
Place the source $s$ at the center of a unit circle, and distribute terminals evenly
along the circle boundary, with consecutive terminals at distance slightly larger than
$\eps$.
Since every boundary point is exactly distance $1$ from the source, routing through
any other terminal already exceeds the $(1+\eps)$ stretch budget.
Consequently, every $(1+\eps,\beta)$-SLT is forced to connect the source directly to
every boundary point.
The MST consists essentially of the boundary cycle together with a
single radial spoke, yielding the lower bound
$\beta \ge \left(\frac{2\pi}{2\pi+1}+o_\eps(1)\right)\frac1\eps$.

This basic construction also suggests a natural guiding principle for proving lower
bounds for SLTs: \EMPH{force as many terminals as possible to connect directly to the root, while
keeping the MST as light as possible}.
Realizing this principle, however, turns out to be rather challenging.
Each of our subsequent constructions overcomes a different obstacle, introducing a new
geometric idea that eventually leads to the following leading constants in the lightness lower bounds: $1$ (3D), $12/11$ (4D), $4/3$ (high dimensions), and finally
back to $2$ in high-dimensional $\ell_1$-normed spaces.

Our first new lower bound shows that moving from 2 to 3 dimensions already requires crossing the natural threshold of 1 in the leading constant.
Instead of a one-dimensional circle, we construct a nearly square
$N\times N$ lattice on a thin patch of the unit sphere, with
$N=\Theta(\eps^{-3/4})$ and with adjacent lattice points at distance
$\eps+O(\eps^{3/2})$.
All distinct lattice points are separated by more than $\eps$, while each lies at distance exactly $1$ from the root $s$. Thus, as in the planar circle example, every $(1+\eps,\beta)$-SLT is forced to contain the edge from $s$ to each lattice point. Hence its
weight is at least $(N+1)^2$.
On the other hand, the lattice admits a spanning tree of weight at
most $1+N(N+2)(\eps+O(\eps^{3/2}))$, hence every such SLT has lightness $\beta$ at least
$$\frac{(N+1)^2}{1+N(N+2)(\eps+O(\eps^{3/2}))} =
\left(1-o_\eps(1)\right)\frac1\eps,$$
since $N=\Theta(\eps^{-3/4})$.
The 3-dimensional construction, however, has an inherent limitation.
Every additional lattice point forced to connect directly to the root also contributes essentially one more lattice edge to the MST. \EMPH{Thus, this construction cannot produce a leading constant exceeding 1}.

This naturally raises the question of whether the constant 1 can already be surpassed in 3 dimensions. Our intuition is that any such improvement would require fundamentally different geometric ideas. \EMPH{Motivated by this, we move to 4 dimensions, where we overcome the constant-1 barrier}. Instead of arranging the terminals as isolated lattice points, our 4-dimensional construction groups them into a chain of regular tetrahedra on the unit sphere.
Each tetrahedron contributes four vertices that are all forced to connect directly to the root, yet the MST pays only a constant amount to connect one tetrahedron to the next and to span the tetrahedron itself using an additional auxiliary point near the centroid of each tetrahedron.
Thus, a single local geometric gadget now generates several forced root edges while
incurring only a comparatively small increase in the MST weight.
This breaks the one-for-one correspondence inherent in the 3-dimensional lattice construction, producing the improved lower bound $(12/11-o_\eps(1))/\eps$.

The 4-dimensional construction naturally extends to arbitrary dimensions by
replacing tetrahedra with regular simplices. As before, every simplex vertex is
forced to connect directly to the root, while the MST connects all simplex
vertices through a single auxiliary point. \EMPH{Higher-dimensional simplices yield increasingly efficient local gadgets, reducing the MST cost per forced root
edge}. Consequently, the achievable ratio converges to $4/3$, yielding the lower bound $(4/3-o_\eps(1))/\eps$ in sufficiently high dimensions.

Finally, we show that \EMPH{this entire Euclidean approach breaks down in general
normed spaces}. The $\ell_1$ lower bound abandons the simplex construction
altogether and instead uses a completely different family of local gadgets,
exploiting geometric properties unique to the $\ell_1$ metric. Each gadget
consists of a star centered at one auxiliary point together with $M$ attached terminal leafs. The terminals all lie at distance $1$ from the root, while each lies at distance $(1+o_\eps(1))\eps/2$ from their auxiliary point. At the same time, the
geometry guarantees that every attached terminal satisfies the forced-spoke
criterion: routing to it through any other point already violates the
$1+\eps$ stretch constraint. Hence every $(1+\eps,\beta)$-SLT is forced to
pay one unit for each of the $M$ terminals.
The MST, however, can connect the same $M$ terminals through their common
auxiliary point, paying only $(1+o_\eps(1))M\eps/2$ for the gadget, plus one
additional $\eps$-scale connection to link the auxiliary points together. Thus,
each gadget contributes $M$ forced root edges but only
$(1+o_\eps(1))\eps(M/2+1)$ to the MST. As the dimension grows, the additive
overhead becomes negligible, the ratio tends to $2/\eps$, and we obtain the
lower bound $(2-o_\eps(1))/\eps$. This shows that our improvement below the longstanding $2/\eps$ barrier is inherently Euclidean and does not extend to
high-dimensional $\ell_1$-normed spaces.
\section{Improved SLT in Euclidean Spaces of Arbitrary Dimension}
\label{sec:highdim}
In this section, we prove \Cref{thm:high-dim-slt}; namely, we are given a set $S$ of $n$ terminal points in Euclidean space of \EMPH{arbitrary dimension}, a root $s \in S$, and a parameter $\eps > 0$, and the goal is to build an SLT on  $S$ with the required bounds on the root-stretch and lightness. Let $\mst=\mst(S)$ denote the minimum spanning tree of point set $S$, and write $\wts(x,y)=\|x-y\|_2$ for Euclidean distance.

Our algorithm proceeds via the following iterative pruning framework.

\paragraph{Iterative pruning framework.} Initialize $T$ to be an $\mst$ of $S$, rooted at $s$. For every vertex $v\in S\setminus \{s\}$, let $T_v$ be the subtree of $T$ rooted at $v$, let $\p(v)$ denote the parent of $v$ in $T$, and let $\anc(T_v)$ be a vertex $u\in V(T_v)$ that minimizes the Euclidean distance $\wts(s,u)$.  Define notation $\widehat{\wts}(T_v)=\wts(T_v)+\wts(v,\p(v))$. Note that if we delete $T_v$ from the tree $T$, then the weight of the tree decreases by $\widehat{\wts}(T_v)$, that is, $\widehat{\wts}(T_v)=\wts(T)-\wts(T-T_v)$. For any weighted graph $H$, we use $\wts(H)$ to denote the weight of all edges in $H$.

We design a local algorithm to construct graph $G$, which iteratively chooses a subtree $T_{v}$ of $T$, rooted at a vertex $v\in V(T)$, adds an edge $e_v$ between the root $s$ and some vertex $u\in V(T_v)$ and some additional (shortcut) edges induced by $V(T_v)$, such that for any vertex $t\in V(T_v)$, there is an $st$-path in graph $G$ of weight at most $(1+\eps)\wts(s,t)$ using the edge $e_v$ and the subgraph of $G$ induced by $V(T_v)$. Let \(G(T_v)\) denote the graph consisting of the added edges for \(T_v\). After that, we will prune $T_v$ from $T$ and argue that the total edge weight we add to the graph $G$, which is $\wts(G(T_v))$, is at most $\frac{1}{\lambda\epsilon}\widehat{\wts}(T_v)$ for a fixed parameter $\lambda = 0.6$. Finally, the SLT of $S$ is a shortest-path tree of $G$ rooted at $s$. This tree preserves all root-distance bounds and has weight at most \(\wts(G)\). This is sufficient for establishing \Cref{thm:high-dim-slt}. Next, we describe how to choose a subtree $T_v$ and design the graph $G(T_v)$ in each iteration.

\paragraph{Choosing a subtree $T_{v^\star}$.} 

We choose a vertex $v$ based on the following criteria; see \Cref{fig:high-dim-pruning} for an illustration.

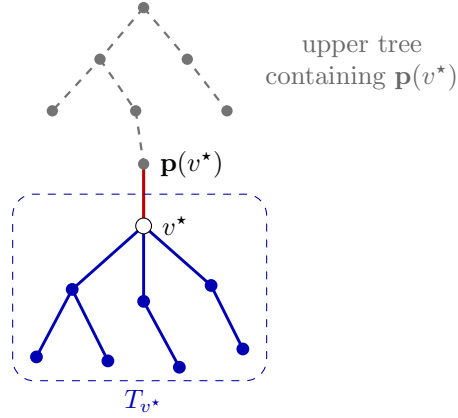
\begin{figure}[t]
	\centering
		\begin{tikzpicture}[
		scale=1.05,
		every node/.style={font=\small},
		Hnode/.style={circle, fill=blue!70!black, inner sep=1.7pt},
		Unode/.style={circle, fill=black!55, inner sep=1.5pt},
		special/.style={circle, draw=black, fill=white, inner sep=2.1pt},
		Hedge/.style={line width=1.1pt, blue!70!black},
		Uedge/.style={line width=.85pt, black!55, dashed},
		pedge/.style={line width=1.15pt, red!75!black},
		]
		
		\coordinate (r)   at (0,2.75);
		\coordinate (a)   at (-0.55,2.10);
		\coordinate (b)   at (0.55,2.10);
		\coordinate (c)   at (-1.15,1.45);
		\coordinate (d)   at (-0.10,1.45);
		\coordinate (e)   at (1.05,1.45);
		\coordinate (pu)  at (0.00,0.78);
		
		\draw[Uedge] (r) -- (a) -- (c);
		\draw[Uedge] (a) -- (d);
		\draw[Uedge] (r) -- (b) -- (e);
		\draw[Uedge] (d) -- (pu);
		
		\node[Unode] at (r) {};
		\node[Unode] at (a) {};
		\node[Unode] at (b) {};
		\node[Unode] at (c) {};
		\node[Unode] at (d) {};
		\node[Unode] at (e) {};
		
		\node[black!60, align=center] at (2.75,2.05)
		{upper tree\\containing $\p(v^\star)$};
		
		\coordinate (u) at (0,0);
		
		\draw[pedge] (pu) -- (u);
		
		\coordinate (v1) at (-0.90,-0.80);
		\coordinate (v2) at (0.00,-0.95);
		\coordinate (v3) at (0.85,-0.75);
		\coordinate (v4) at (-1.35,-1.65);
		\coordinate (v5) at (-0.45,-1.70);
		\coordinate (v6) at (0.45,-1.78);
		\coordinate (v7) at (1.25,-1.55);
		
		\draw[Hedge] (u) -- (v1);
		\draw[Hedge] (u) -- (v2);
		\draw[Hedge] (u) -- (v3);
		\draw[Hedge] (v1) -- (v4);
		\draw[Hedge] (v1) -- (v5);
		\draw[Hedge] (v2) -- (v6);
		\draw[Hedge] (v3) -- (v7);
		
		\node[Hnode] at (v1) {};
		\node[Hnode] at (v2) {};
		\node[Hnode] at (v3) {};
		\node[Hnode] at (v4) {};
		\node[Hnode] at (v5) {};
		\node[Hnode] at (v6) {};
		\node[Hnode] at (v7) {};
		
		\node[Unode, label=right:{$\p(v^\star)$}] at (pu) {};		
		\node[special, label=right:{$v^\star$}] at (u) {};

		\draw[blue!65!black, dashed, rounded corners=8pt]
		(-1.65,-1.95) rectangle (1.55,0.4);
		
		\node[blue!70!black] at (0,-2.20) {$T_{v^\star}$};
		
	\end{tikzpicture}
	\caption{One pruning step in the high-dimensional construction.  The selected subtree \(T_{v^\star}\) is separated from the remaining upper part of the current MST by the parent edge \((v^\star,\p(v^\star))\).}
	\label{fig:high-dim-pruning}
\end{figure}

\begin{lemma}\label{lem:subtree}
    If $|S|\geq 2$ and $0<\eps\leq \frac{1}{\lambda}$, then there exists a vertex $v\in V(T)$, $v\neq s$, satisfying 
\begin{equation}\label{eq:cutoff}
   \widehat{\wts}(T_{v}) \geq \lambda\eps\cdot  \wts(s,\anc(T_{v})) .
\end{equation}
\end{lemma}
\begin{proof}
We show that every child of $s$ satisfies \Cref{eq:cutoff}. Let $v$ be a child of $s$. 
then $\wts(s,\anc(T_{v})) \leq \wts(s,v)$. Since $\p(v)=s$, then $\widehat{\wts}(T_v)\geq \wts(v,\p(v))=\wts(s,v)$. Consequently, $\wts(s,\anc(T_{v})) \leq \wts(s,v)\leq \frac{1}{\lambda\eps} \cdot \widehat{\wts}(T_{v})$ for all $0<\eps\leq \frac{1}{\lambda}$.
\end{proof}
In a post-order traversal of $T$, we choose a first vertex $v^\star$ that satisfies \Cref{eq:cutoff}.
In particular, the children of $v^\star$ (if any) do not satisfy \Cref{eq:cutoff}.

\paragraph{Constructing a low root-stretch and light subgraph for $V(T_{v^\star})\cup \{s\}$.}
We construct an SLT for the vertices in $V(T_{v^\star})\cup \{s\}$, using $T_{v^\star}$, one edge between $s$ and some vertex in $V(T_{v^\star})$, and possibly some additional edges induced by $V(T_{v^\star})$. We distinguish between two cases depending on the relative position of the vertices $s$, $v^\star$ and $\anc(T_{v^\star})$ in the Euclidean space. In all cases, we need to construc a graph $G(T_{v^\star})$ such that $G(T_{v^\star})$ has root-stretch at most $1+\eps$, and its weight is bounded by $\wts(G(T_{v^\star}))\leq \brac{\frac{1}{\lambda\eps} + O(1)} \cdot \widehat{\wts} (T_{v^\star})$. We distinguish between several cases for which we will employ different strategies.

Construct a $\brac{10^{-3}\eps \rho}$-net of the subtree $T_{v^\star}$,
called $N_{v^\star}$. Consider two cases below.
\begin{enumerate}
	\item If, for every vertex $w\in N_{v^\star}$, we have
	$$\wts(s, \anc(T_{v^\star})) + \wts(\anc(T_{v^\star}), w) \leq (1+\epsilon)\wts(s, w) - 0.005\epsilon\rho ,$$
	then define the graph $G(T_{v^\star})$ to be the union of $T_{v^\star}$, the edge $(s,\anc(T_{v^\star}))$, and the star centered at $\anc(T_{v^\star})$ with leaves in $N_{v^\star}$.
	
	\item Otherwise, there exists a vertex $w\in N_{v^\star}$ such that 
	$$\wts(s, \anc(T_{v^\star})) + \wts(\anc(T_{v^\star}), w) > (1+\epsilon)\wts(s, w) - 0.005\epsilon\rho .$$
	Then define the graph $G(T_{v^\star})$ to be the union of $T_{v^\star}$, the edge $(s, v^\star)$, and the star centered at $v^\star$ with leaves in $N_{v^\star}$.
\end{enumerate} 

This completes the construction of $G(T_{v^\star})$. It remains to analyze the root stretch and the lightness of $G(T_{v^\star})$. Define $\rho = \wts(s, \anc(T_{v^\star}))$. Before going to the analysis, we need the following upper bound on distances $\wts(s, u)$ for all $u\in V(T_{v^\star})$.
\begin{lemma}\label{lem:dist-bnd}
	For every $u\in V(T_{v^\star})$, we have
	$$\frac{1 - \lambda\epsilon}{1+\lambda\epsilon}\cdot\rho\leq \wts(s, u)\leq \frac{1 + \lambda\epsilon}{1-\lambda\epsilon}\cdot\rho .$$
\end{lemma}
\begin{proof}

    If $u = v^\star$, consider the path from $\anc(T_{v^*})$ to $v^*$. If $v^*$ and $\anc(T_{v^*})$ are the same, both inequalities trivially hold. Otherwise, let $w$ be the child of $v^\star$ such that $T_w$ containing $\anc(T_{v^\star})$. Since $v^\star$ is the first vertex in post-order traversal of $T$ satisfying \Cref{eq:cutoff}, $w$ does not satisfies \Cref{eq:cutoff}, meaning $\widehat{\wts}(T_w) < \lambda\eps \cdot \wts(s, \anc(T_w)) = \lambda\eps\rho$. Since the tree $T_w$ together with the edge $(w, v^\star)$ contains a path from $v^\star$ to $\anc(T_w)$ (or $\anc(T_{v^\star})$), 
    \begin{equation}
        \label{eq:dist-to-anc}
        \wts(v^\star, \anc(T_{v^\star})) \leq \widehat{\wts}(T_w) < \lambda\eps\rho.
    \end{equation}
    Using triangle inequality, we get: 
    \begin{equation*}
        \wts(s, v^\star) \leq \wts(s, \anc(T_{v^\star})) + \wts(\anc(T_{v^\star}), v^\star) \leq \rho + \lambda\eps\rho = (1 + \lambda\eps)\rho \leq \frac{1 + \lambda\eps}{1 - \lambda\eps}\rho, 
    \end{equation*}
    as claimed. For the left hand side inequality, we have
    \begin{equation*}
        \wts(s, v^\star) \geq \wts(s, \anc(T_{v^\star})) - \wts(\anc(T_{v^\star}), v^\star) \geq \rho - \lambda\eps\rho = (1 - \lambda\eps)\rho \geq \frac{1 - \lambda\eps}{1 + \lambda\eps}\rho.
    \end{equation*}
    
	We focus on the case when $u\neq v^\star$. Assume $u$ lies in a subtree $T_v$ where $v$ is a child of $v^\star$. Then, similar to the proof for $w$, we have
	\begin{equation}\label{eq:vstaru}
	    \wts(v^\star, u)\leq \widehat{\wts}(T_v)\leq \lambda \epsilon \cdot \wts(s, \anc(T_v))\leq \lambda\epsilon\cdot \wts(s, u) .
	\end{equation}
	Then, the triangle inequality yields
	$$\begin{aligned}
		\wts(s, u) &\leq \wts(s, \anc(T_{v^\star})) + \wts(\anc(T_{v^\star}), v^\star) + \wts(v^\star, u)\\
		&\leq \rho + \lambda \epsilon\rho + \lambda\epsilon\cdot \wts(s, u).
	\end{aligned}$$
    The last inequality follows from \Cref{eq:dist-to-anc} and \Cref{eq:vstaru}. 
	This is equivalent to $\wts(s, u)\leq \frac{1 + \lambda\epsilon}{1-\lambda\epsilon}\cdot\rho$. Similarly, we have
	$$\begin{aligned}
		\wts(s, u) &\geq \wts(s, \anc(T_{v^\star})) - \wts(\anc(T_{v^\star}), v^\star) - \wts(v^\star, u)\\
		&\geq \rho - \lambda \epsilon\rho - \lambda\epsilon\cdot \wts(s, u) .
	\end{aligned}$$
	This is equivalent to $\wts(s, u)\geq \frac{1 - \lambda\epsilon}{1+\lambda\epsilon}\cdot\rho$.
\end{proof}

\begin{corollary}
    \label{cor:v*-point-dist}
    For every $u \in V(T_{v^\star})$, $\wts(v^\star, u) \leq \frac{1 + \lambda\eps}{1 - \lambda\eps} \cdot \lambda\eps\rho$.
\end{corollary}

\begin{proof}
    If $u=v^\star$, the inequality is trivial. Otherwise, let $w$ be the child of $v^\star$ such that $u$ is in $T_w$. Since the path from $u$ to $v^\star$ is in $T_{w} \cup \{(w, v^\star)\}$,
    \begin{equation}
        \label{eq:wv*1}
        \wts(v^\star, u) \leq \widehat{\wts}(T_{w}) \leq \wts(s, \anc(T_{w})) \cdot \lambda\eps.
    \end{equation}
    The last inequality holds because the tree $T_{w}$ does not satisfy \Cref{eq:cutoff}. By \Cref{lem:dist-bnd}, $\wts(s, \anc(T_{w})) \leq \frac{1 + \lambda\eps}{1 - \lambda\eps}\cdot \rho$. Together with \Cref{eq:wv*1}, we obtain
    \begin{equation}
        \label{eq:wv*2}
        \wts(v^\star, u) \leq  \wts(s, \anc(T_{w}))  \cdot \lambda\eps \leq \frac{1 + \lambda\eps}{1 - \lambda\eps}\cdot \lambda\eps\rho \qquad, 
    \end{equation}
    as desired.
\end{proof}

\paragraph{Root-stretch analysis in Case 1.} 
We analyze the root stretch in $G(T_{v^\star})$. 
Assume that $\anc(T_{v^\star})$ is in the subtree $T_v$ with $\p(v)=v^\star$ ($\anc(T_v) = \anc(T_{v^\star})$). Since the subtree $T_v$ does not satisfy \Cref{eq:cutoff}, then $\widehat{\wts}(T_{v}) < \lambda \eps \cdot \wts(s,\anc(T_v))$. 
For every vertex $u\in V(T_v)\cup \{v^\star\}$, we have 
\begin{align*}
 \dist_{G(T_{v^\star})}(s,u) & \leq \wts(s,\anc(T_{v^\star}))+\dist_T(\anc(T_{v^\star}),u) \\
    &\leq \wts(s,\anc(T_{v^\star})) + \widehat{\wts}(T_{v})\quad\text{(since $\anc(T_{v^\star})$ and $u$ are in $T_v \cup \{(v, v^\star)\}$)}\\
    &\leq \left(1+\lambda\eps\right) \cdot \wts(s,\anc(T_{v^\star})),
\end{align*}
so the root stretch of every vertex $u\in V(T_v)\cup \{v^\star\}$ in $G(T_{v^\star})$ is bounded by 
\[ 
    \frac{\dist_{G(T_{v^\star})}(s,u)}{\wts(s,u)} 
    \leq \frac{(1+\lambda\eps) \cdot \wts(s,\anc(T_v))}{ \wts(s,\anc(T_{v^\star}))}
    = 1+\lambda\eps <1+\eps.
\]

For a vertex $u\in V(T_{v^\star})\setminus V(T_v)$, let $w\in N_{v^\star}$ be the net vertex closest to $u$ in $G(T_{v^\star})$. Then, we have
$$\begin{aligned}
	\dist_{G(T_{v^\star})}(s, u) &\leq \underbrace{\wts(s, \anc(T_{v^\star})) + \wts(\anc(T_{v^\star}), w)}_{\leq (1 + \eps)\wts(s, w) -0.005\eps\rho} + \underbrace{\dist_{G(T_{v^\star})}(w, u)}_{\leq 10^{-3}\eps\rho}\\
	&\leq (1+\epsilon)\wts(s, w) - 4\times 10^{-3}\epsilon\rho\\
    &\leq (1 + \eps)(\wts(s, u) + \underbrace{\wts(u, w)}_{\leq 10^{-3}\eps\rho}) - 4\times 10^{-3}\epsilon\rho \qquad \text{(by triangle inequality)}\\
	&\leq (1+\epsilon)\brac{\wts(s, u) + 10^{-3}\epsilon\rho} - 4\cdot 10^{-3}\epsilon\rho\\
	&\leq (1+\epsilon)\wts(s, u) .
\end{aligned}$$

\paragraph{Root-stretch analysis in Case 2.}
In this case, the key idea is to upper bound the distance between $s$ and $v^\star$ using elementary Euclidean geometry, rather than relying only on the triangle inequality.

For notational convenience, let $a=\anc(T_{v^\star})$. Intuitively, if the quantity $\wts(s,a)+\wts(a,w)$ is close to $(1+\eps)\wts(s,w)$, and both $\wts(v^\star, a), \wts(v^\star, w)$ not much larger than $\eps \rho / 2$, then $a, v^\star, w$ would be nearly collinear.
The point $v^\star$ is close to both $a$ and $w$, and its distance to the line $aw$ can be controlled by the excess in $\wts(v^\star,a)+\wts(v^\star,w)$ over the edge length $\wts(a, w)$. This allows us to obtain a sharper upper bound on $\wts(s,v^\star)$, leaving enough slack to connect the remaining vertices. See \Cref{fig:high-dim-collinear} for an illustration.

\begin{figure}[t]
	\centering
	\begin{tikzpicture}[
	scale=1.15,
	every node/.style={font=\small},
	vertex/.style={
		circle,
		draw=black,
		fill=white,
		line width=.75pt,
		minimum size=2.2mm,
		inner sep=0pt
	},
	mainedge/.style={
		line width=1.05pt,
		black
	},
	dashededge/.style={
		line width=.9pt,
		black!55,
		dashed
	},
	collinear/.style={
		line width=1.05pt,
		blue!70!black
	}
	]
	
	\coordinate (s) at (0,3.0);
	\coordinate (a) at (-2.4,0);
	\coordinate (w) at (2.4,0);
	\coordinate (m) at (0,0);
	
	\coordinate (vstar) at (0.55,-0.4);
	
	\draw[mainedge] (s) -- (a) -- (w) -- cycle;
	
	\draw[collinear] (a) -- (vstar) -- (w);
	
	\draw[dashededge] (s) -- (m);
	\draw[dashededge] (m) -- (vstar);
	
	
	\node[vertex, label=above:{$s$}] at (s) {};
	\node[vertex, label=below left:{$a$}] at (a) {};
	\node[vertex, label=below right:{$w$}] at (w) {};
	\node[vertex, label=above right:{$m$}] at (m) {};
	\node[vertex, label=below:{$v^\star$}] at (vstar) {};
	
	\node[black!55, right=2pt] at ($(s)!0.6!(m)$) {$\wts(s,m)$};
	\node[black!55, above=2pt] at ($(m)!0.55!(w)$) {$\wts(m,w)$};
	
\end{tikzpicture}
	\caption{The geometric configuration in the proof of Lemma~\ref{lem:colinear}. 
		The point \(m\) is the midpoint of \(aw\), and the points
		\(a,v^\star,w\) are approximately collinear.}
	\label{fig:high-dim-collinear}
\end{figure}
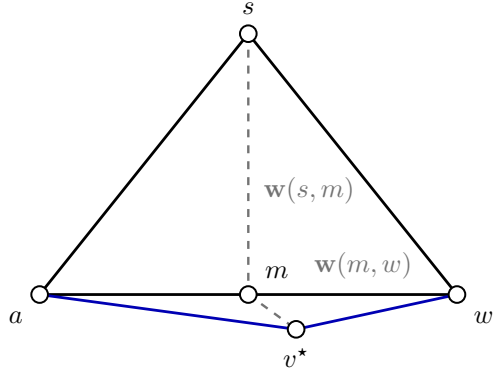

\begin{lemma}\label{lem:colinear}
	$\wts(s, v^\star)\leq \brac{1 + \epsilon (1.424\lambda - 0.4925) }\rho$.
\end{lemma}

To bound $\wts(s, v^\star)$, we use the triangle inequality. Let $m$ be the midpoint of $aw$, $\wts(s, v^\star) \leq \wts(s, m) + \wts(m, v^\star)$. We then bound $\wts(s, m)$ and $\wts(m, v^\star)$ by applying Apollonius' theorem. Let $x$ be the ratio $\frac{\wts(a, w)}{\rho}$. We show that: 

\begin{claim}
    \label{clm:sm}
    $\wts(s, m) \leq \left(1 + \frac{x - 0.995\eps}{2(1 + \eps)}\right)\cdot \rho$.
\end{claim}

\begin{proof}
    By the condition of Case~2, 
    \begin{equation*}
    \wts(s, a) + \wts(a, w) > (1+\epsilon)\wts(s, w) - 0.005\epsilon\rho.
    \end{equation*}
    Since $\wts(s, a) = \rho$, we have
    \begin{equation}\label{eq:case2+}
        \wts(s, w) \leq \frac{x + (1 + 0.005\eps)}{1 + \eps}\cdot \rho = \left(1 + \frac{x - 0.995\eps}{1 + \eps}\right)\rho.
    \end{equation}

    Applying Apollonius' theorem, we get
    \begin{equation}
    \label{eq:sm-bound}
        \wts(s, m)^2
	=
	\frac{\wts(s, a)^2+\wts(s, w)^2}{2}-\frac{\wts(a, w)^2}{4}
	=
	\frac{\rho^2+\wts(s, w)^2}{2}-\frac{(x\rho)^2}{4}.
    \end{equation}

    The combination of \Cref{eq:case2+} and \Cref{eq:sm-bound} yields
    \begin{equation*}
        \wts(s, m)^2 \leq  \left(\frac{1}{2} + \frac{1}{2} \cdot \left(1 + \frac{x - 0.995\eps}{1 + \eps}\right)^2 - \frac{x^2}{4}\right)\rho^2 .
    \end{equation*}
    Setting $q = \frac{x - 0.995\eps}{1 + \eps}$, we have
    \begin{equation*}
        \begin{split}
            \wts(s, m)^2 &\leq \left(\frac{1}{2} + \frac{(1 + q)^2}{2} - \frac{x^2}{4}\right)\rho^2 \\
            &= \left(1 + q + \frac{q^2}{2} - \frac{q^2}{4}\right)\rho^2 \quad\text{(since $q < x$)}\\
            &= \left(1 + \frac{q}{2}\right)^2\rho^2 ,
        \end{split}
    \end{equation*}
    implying $\wts(s, m) \leq \left(1 + \frac{q}{2}\right)\rho = \left(1 + \frac{x - 0.995\eps}{2(1 + \eps)}\right)\rho$. 
\end{proof}

We then bound $\wts(v^{\star}, m)$.
\begin{claim}
    \label{clm:v*m}
    $\wts(v^\star, m) \leq
	\sqrt{\frac{\brac{1+\lambda^2\epsilon^2}\epsilon^2\lambda^2}{(1 - \lambda\eps)^2}-\frac{x^2}{4}}\cdot \rho$.
\end{claim}

\begin{proof}
    Applying Apollonius' theorem for the triangle $\Delta(a, w, v^\star)$, we get
    \begin{equation}
        \label{eq:awv*}
        \wts(v^\star, m)^2 = \frac{\wts(a, v^\star)^2+\wts(w, v^\star)^2}{2}-\frac{\wts(a, w)^2}{4} = \frac{\wts(a, v^\star)^2+\wts(w, v^\star)^2}{2} - \frac{x^2\rho^2}{4}.
    \end{equation}

    By \Cref{cor:v*-point-dist}, $\wts(v^\star, w) \leq \frac{1 + \lambda\eps}{1 - \lambda\eps} \cdot \lambda\eps\rho$. On the other hand, since $a = \anc(T_{v^\star})$, by \Cref{eq:dist-to-anc}, $\wts(a, v^\star) \leq \lambda\eps\rho$. Combined with \Cref{eq:awv*}, we obtain
    \begin{equation*}
        \wts(v^\star, m)^2 \leq \frac{(\lambda\eps\rho)^2 + \left(\frac{1 + \lambda\eps}{1 - \lambda\eps}\right)^2(\lambda\eps\rho)^2}{2} - \frac{x^2\rho^2}{4} = \rho^2 \cdot \left(\frac{(1 + \lambda^2\eps^2)\lambda^2\eps^2}{(1 - \lambda\eps)^2} - \frac{x^2}{4}\right).
    \end{equation*}
    Taking square root of both sides, our claim follows. 
\end{proof}

We now prove \Cref{lem:colinear}.

\begin{proof}[Proof of \Cref{lem:colinear}]
	From \Cref{clm:sm}, \Cref{clm:v*m} and the triangle inequality,
	\begin{equation}\label{eq:svstar}
	    \wts(s, v^\star)
	\le
	\wts(s, m) + \wts(m, v^\star)
	\leq
	\brac{1+\frac{x-0.995\eps}{2(1+\eps)}
	+\sqrt{\frac{\brac{1+\lambda^2\epsilon^2}\epsilon^2\lambda^2}{(1 - \lambda\eps)^2}-\frac{x^2}{4}}}\cdot \rho.
	\end{equation}
	
	It remains to upper bound the expression independently of $x$.
	Define
    \begin{equation}
        \label{eq:phi}
        \begin{split}
            \phi(x) &=
	\frac{x-0.995\eps}{2(1+\eps)}
	+\sqrt{\frac{\brac{1+\lambda^2\epsilon^2}\epsilon^2\lambda^2}{(1 - \lambda\eps)^2}-\frac{x^2}{4}}\\
        &= -\frac{0.995\eps}{2(1+\eps)}
		+
		\frac{1}{1+\eps}\cdot \frac{x}{2}
		+
		\sqrt{
			\frac{(1+\lambda^2\eps^2)\eps^2\lambda^2}
			{(1 - \lambda\eps)^2}
			-\frac{x^2}{4}
		}.
        \end{split}
    \end{equation}
	By the Cauchy-Schwarz inequality,
	\begin{equation}
        \label{eq:cauchy-schwarz}
	    \begin{split}
	        &\frac{1}{1+\eps}\cdot \frac{x}{2}
		+
		\sqrt{
			\frac{(1+\lambda^2\eps^2)\eps^2\lambda^2}
			{(1 - \lambda\eps)^2}
			-\frac{x^2}{4}
		}
		\\
		&\qquad\le
		\sqrt{
			\frac{1}{(1+\eps)^2}+1
		}
		\cdot
		\sqrt{
			\frac{x^2}{4}
			+
			\frac{(1+\lambda^2\eps^2)\eps^2\lambda^2}
			{(1 - \lambda\eps)^2}
			-\frac{x^2}{4}
		}
		\\
		&\qquad=
		\sqrt{
			1+\frac{1}{(1+\eps)^2}
		}
		\cdot
		\sqrt{
			\frac{(1+\lambda^2\eps^2)\eps^2\lambda^2}
			{(1 - \lambda\eps)^2}
		}\\
        &\qquad=\sqrt{
		1+\frac{1}{(1+\eps)^2}
	}
	\cdot
	\frac{\eps\lambda\sqrt{1+\lambda^2\eps^2}}{1-\lambda\eps}.
	    \end{split}
	\end{equation}
	The last equality used the assumption that $0 < \lambda\eps < 1$. Therefore, 
	\[
	\phi(x)
	\le
	\eps\left(
	\frac{\lambda\sqrt{1+\lambda^2\eps^2}}{1-\lambda\eps}
	\sqrt{1+\frac{1}{(1+\eps)^2}}
	-
	\frac{0.995}{2(1+\eps)}
	\right).
	\]
	This upper bound is independent of $x$. Now assume $\eps<0.01$ and note that $\lambda=\frac35<\frac23$. Then, $\lambda\eps<\frac{1}{150}$, which implies $\frac{1}{1-\lambda\eps} < \frac{150}{149}$ and $\sqrt{1+\lambda^2\eps^2} < \sqrt{1+\frac{1}{22500}}$. Consequently,
    \begin{equation*}
        \begin{split}
            \phi(x) &\le \eps\left(\frac{\lambda\sqrt{1+\lambda^2\eps^2}}{1-\lambda\eps}\sqrt{1+\frac{1}{(1+\eps)^2}}-\frac{0.995}{2(1+\eps)}\right)\\
            &\leq \eps\left(\frac{150}{149}\sqrt{1+\frac{1}{22500}}\underbrace{\sqrt{1 + \frac{1}{(1 +\eps)^2}}}_{\leq \sqrt{2}}\,\lambda-\underbrace{\frac{0.995}{2(1+\eps)}}_{\leq\frac{0.995}{2.02}}\right)\\
		&<\eps\left(1.424\lambda-0.4925\right).
        \end{split}
    \end{equation*}
	Then, \Cref{eq:svstar} gives 
    \begin{equation*}
        \wts(s, v^\star)\leq \brac{1 + \epsilon (1.424\lambda - 0.4925) }\rho ,  
    \end{equation*}
    as claimed.
\end{proof}

For any vertex $u$ in $T_{v^\star}$, by the construction of $G(T_{v^\star})$, let $z\in N_{v^\star}$ be the closest net vertex to $u$, then by \Cref{lem:colinear} and \Cref{lem:dist-bnd} we have 
$$\begin{aligned}
	\dist_{G(T_{v^\star})}(s, u) &\leq \wts(s, v^\star) + \wts(v^\star, z) + \dist_{T_{v^\star}}(z, u)\\
	&\leq \brac{1 + \epsilon (1.424\lambda - 0.4925) }\rho + \wts(v^\star, z)+ 10^{-3}\epsilon\rho\\
	&\leq \brac{1 + \epsilon (1.424\lambda - 0.4925) }\rho + \underbrace{\frac{1 + \lambda\epsilon}{1 - \lambda\epsilon}}_{\leq 1 + 3\lambda\eps}\cdot \lambda\eps\rho + 10^{-3}\epsilon\rho \qquad \text{(by \Cref{cor:v*-point-dist})}\\
    &\leq \brac{1 + \epsilon (1.424\lambda - 0.4925) }\rho + (1 + 3\lambda\eps)\lambda\eps\rho + 10^{-3}\epsilon\rho \qquad \\
	&\leq \brac{1 + \epsilon(2.424\lambda + 3\lambda^2\eps - 0.4915)}\rho \leq (1 + \epsilon)\rho.
\end{aligned}$$
The last inequality holds when $2.424\lambda + 3\lambda^2\eps - 0.4915 \leq 1$ as we take $\lambda = 0.6$ and $\eps < 0.01$. Since $\wts(s, u) \geq \rho$ for every $u \in T_{v^\star}$, $\dist_{G(T_{v^\star})}(s, u) \leq (1  +\eps)\wts(s, u)$.

\paragraph{Weight analysis in Case 1.}
Since $N_{v^\star}$ is a $(10^{-3}\eps\cdot \wts(s,\anc(T_{v^\star})))$-net of the subtree $T_{v^\star}$, we can bound its size as
\[
    |N_{v^\star}|
    \leq \ceil{\frac{\widehat{\wts}(T_{v^\star})}{10^{-3}\epsilon \rho}} .
\]

On the other hand, by \Cref{cor:v*-point-dist} and \Cref{eq:dist-to-anc}, we know that the weight of any edge between $\anc(T_{v^\star})$ and $u\in N_{v^\star}$ is at most
$$\wts(\anc(T_{v^\star}), u)\leq \wts(v^\star, u) + \wts(v^\star, \anc(T_{v^\star}))\leq \frac{1 + \lambda\eps}{1 - \lambda\epsilon}\cdot \lambda\epsilon\rho + \lambda\eps\rho < 4\lambda\eps\rho .$$

Overall, the total weight of $G(T_{v^\star})$ is 
\begin{align*}
    \wts(G(T_{v^\star}))
    & \leq \wts(T_{v^\star})+ \wts(s,\anc(T_{v^\star}))+|N_{v^\star}|\cdot 4\lambda\eps\rho\\
    &\leq  \wts(T_{v^\star})+ \rho + 4000\cdot\widehat{\wts}(T_{v^\star})\\
    &\leq  \brac{\frac{1}{\lambda\epsilon} + 4001}\cdot\widehat{\wts}(T_{v^\star}) = \brac{\frac{1}{\lambda\epsilon} + O(1)}\cdot \widehat{\wts}(T_{v^\star})
\end{align*}

\paragraph{Weight analysis in Case 2.}
By \Cref{cor:v*-point-dist}, the weight of every edge between $v^\star$ and $u\in N_{v^\star}$ is at most
$$\wts(v^\star, u) \leq \frac{1 + \lambda\epsilon}{1 - \lambda\epsilon}\cdot \lambda\epsilon\rho .$$

Overall, the total weight of $G(T_{v^\star})$ is 
\begin{align*}
	\wts(G(T_{v^\star}))
	& \leq \wts(T_{v^\star})+ \wts(s,v^\star)+|N_{v^\star}|\cdot \frac{1 + \lambda\epsilon}{1 - \lambda\epsilon}\cdot\lambda\epsilon\rho\\
	&\leq  \wts(T_{v^\star})+ \underbrace{\wts(s, a)}_{=\rho} + \underbrace{\wts(a, v^\star)}_{\leq \lambda\eps\rho} + 2000\cdot\widehat{\wts}(T_{v^\star})\\
    &\leq \wts(T_{v^\star})+ (1 + \lambda\eps)\rho + 2000\cdot\widehat{\wts}(T_{v^\star})\\
	&\leq  \brac{\frac{1}{\lambda\epsilon} + 2001}\cdot\widehat{\wts}(T_{v^\star}) = \brac{\frac{1}{\lambda\epsilon} + O(1)}\cdot \widehat{\wts}(T_{v^\star}) .
\end{align*}

In both cases, we have shown that each subgraph $G(T_{v^\star})$ has root stretch at most $1+\eps$ and weight at most $\brac{\frac{1}{\lambda\eps} + O(1)}\cdot \widehat{\wts}(T_{v^\star})$.
As $\lambda = 0.6$, this completes the proof of \Cref{thm:high-dim-slt}.

\section{Improved SLT in Euclidean Plane}
In this section, we prove \Cref{thm:plane-slt}; namely, we are given a set $S\subset \mathbb{R}^2$ of $n$ terminal points in the \EMPH{Euclidean plane}, a root $s \in S$, and a parameter $\eps > 0$, and the goal is to build an SLT on  $S$ with the required bounds on the root-stretch and lightness.

We first need some notations. For any pair of points $u, v$ on $\mst$ (not necessarily from the point set $S$), let $\mst[u, v]$ be the tree path between $u, v$ on $\mst$. For an edge $e=xy$ and a region $R\subseteq\mathbb{R}^2$, $e\cap R$ denotes the portion of the segment $xy$ lying in $R$, and $\wts(e\cap R)$ denotes its Euclidean length. For the upper-bound construction in this section, after scaling we assume that the minimum distance between two distinct points in $S$ is $2$; this normalization does not affect stretch or lightness.
We use the following $\epsilon$-dependent auxiliary parameters:
\[
    \alpha=\epsilon^{1/3},
    \qquad
    m_\epsilon=\left\lceil \epsilon^{-1/12}\right\rceil,
    \qquad
    \beta=\frac1{m_\epsilon},
    \qquad
    \gamma=\epsilon^{1/4}.
\]
Besides, define
\[
    M_\epsilon=\left\lceil \frac{2\pi}{(1-4\gamma)\epsilon}\right\rceil,
    \qquad
    \theta=\frac{2\pi}{M_\epsilon}.
\]

\subsection{Algorithm Description}
Our algorithm consists of several phases.
\paragraph{Annular decomposition.}
Partition the exterior of the unit disk centered at $s$ into thin annuli.
For every $i\ge 0$ and $0\le j<1/\beta$, let
\[
A_i^{(j)}=
\left\{x\in\mathbb{R}^2:
(1+\alpha)^i(1+j\beta\alpha)
\le \wts(s,x)<
(1+\alpha)^i(1+(j+1)\beta\alpha)
\right\} .
\]

\begin{lemma}\label{lem:sparse-buffer}
	There is an index $j^\star$, $0\le j^\star<1/\beta$, such that
	\[
	\sum_{i\ge 0}\sum_{e\in E(\mst)} \wts\bigl(e\cap A_i^{(j^\star)}\bigr)
	\le \beta\,\wts(\mst).
	\]
\end{lemma}

\begin{proof}
	For fixed $i$, the annuli $A_i^{(0)},\ldots,A_i^{(1/\beta-1)}$ partition the shell
	$\{x:(1+\alpha)^i\le \wts(s,x)<(1+\alpha)^{i+1}\}$.
	Summing over $j$ therefore counts each portion of $\mst$ outside the unit disk exactly once, and portions inside the unit disk are not counted.
	Averaging over the $1/\beta$ possible values of $j$ gives the claim.
\end{proof}

Fix the index $j^\star$ from \Cref{lem:sparse-buffer}.  For $i\ge 0$, define
\[
L_i=(1+\alpha)^i(1+j^\star\beta\alpha),
\qquad
R_i=(1+\alpha)^{i+1}(1+(j^\star+1)\beta\alpha),
\]
and let
\[
A_i=\{x:L_i\le \wts(s,x)<R_i\}.
\]
Thus consecutive regions overlap exactly on the selected sparse annuli:
$A_i\cap A_{i+1}=A_{i+1}^{(j^\star)}$.
We call
\[
B_i=A_i^{(j^\star)}\cup A_{i+1}^{(j^\star)}
\]
the buffer of $A_i$.
By \Cref{lem:sparse-buffer}, the weight of the $\mst$ clipped in the regions $A_i$ is bounded by
\begin{equation}\label{eq:annulus-overlap}
	\sum_{i\ge0}\wts(\mst\cap A_i)
	\le (1+\beta)\wts(\mst).
\end{equation}

\paragraph{Net points.}
Fix a region $A_h$, where $h\ge 0$. By definition, $L_h$ is a lower bound on the distance from $s$ to every point of $A_h$.
Subdivide every edge of $\mst$ at each intersection with $\partial A_h$.
The subdivision points are Steiner points used only in the analysis.
Let $U_h$ be the set of Steiner points on $\partial A_h$, let $S_h=S\cap A_h$, and let $F_h$ be the forest obtained by restricting the subdivided $\mst$ to $A_h$. If a tree in $F_h$ does not contain any points in the non-buffer region $A_h\setminus B_h$, then remove it from $F_h$.

\begin{lemma}\label{lem:tree-contain}
	Every point $t\in S\setminus\{s\}$ is contained in $T\cap S$ for at least one tree $T\in F_h$ for some level $h$.
\end{lemma}

\begin{proof}
	If $t$ lies in the non-buffer part of some region $A_h$, then by definition the connected component of $\mst$ restricted to $A_h$ belongs to $F_h$.
	
	It remains to consider the case in which $t$ lies in a selected buffer annulus, say $A_i^{(j^\star)}$.
	This buffer is the overlap between the two neighboring regions $A_{i-1}$ and $A_i$ (with the evident one-sided interpretation when $i=0$).
	Let $P$ be the path in $\mst$ from $t$ to the root $s$.
	Since $s$ is not in the radial interval of $A_i^{(j^\star)}$, the path $P$ must leave this buffer interval.
	If the first exit from the buffer is through its outer boundary, then the component of $F_i$ containing $t$ contains a point of $A_i\setminus B_i$.
	If the first exit is through its inner boundary, then the component of $F_{i-1}$ containing $t$ contains a point of $A_{i-1}\setminus B_{i-1}$.
	In either case, the relevant component intersects the non-buffer part of the corresponding region, hence it is retained in that forest and contains $t$.
	Boundary coincidences can be handled by the same half-open convention used for the annuli, or by an infinitesimal perturbation.
\end{proof}

For any component $T\in F_h$, let $S(T):=T\cap S$, and let $\net(T)\subseteq S(T)$ be a $\gamma\epsilon L_h$-net in the tree metric of $T$. 
Then $\net(T)$ is also a $\gamma\epsilon L_h$-net of $S(T)$, and the standard packing bound for points separated in a tree gives
\[
|\net(T)|
\le 1+\frac{\wts(T)}{\gamma\epsilon L_h} .
\]

\begin{corollary}
	\label{lem:sum-components-revised}
	All the retained tree components satisfy
	\[
	\sum_{h\ge0}\sum_{T\in F_h}\wts(T)
	\le
	\sum_{h\ge0}\wts(\mst\cap A_h)
	\le
	(1+\beta)\wts(\mst).
	\]
\end{corollary}

\begin{proof}
	For each fixed level $h$, the components of $F_h$ are edge-disjoint subtrees
	of the subdivided MST restricted to $A_h$.  Removing components can only
	decrease the total weight, and hence
	\[
	\sum_{T\in F_h}\wts(T)
	\le \wts(\mst\cap A_h).
	\]
	Summing over $h$ and then using \eqref{eq:annulus-overlap} gives the claim.
\end{proof}

\paragraph{Shortcuts within cones.}
We now construct a graph $G=(S,E)$.
Initialize $G$ with the edges of $\mst$.
Partition the plane into cones of angle $(1-4\gamma)\epsilon$ with apex $s$. For each cone $C$ and each annulus $A_h$, if $C\cap A_h\cap S\neq \emptyset$, let $\anc(C, A_h)\in S$ be the closest point to $s$, and call $h$ the level of $\anc(C, A_h)$.

\begin{framed}
	\noindent\textbf{Shortcut scheme for $G$.}
	For each cone $C$, perform the following steps.
	\begin{enumerate}[(1)]
		\item\label{alg:shortcut1}
		Let
		\(
		h_1<h_2<\cdots<h_\ell
		\)
		be all levels $h$ for which $C$ contains at least one anchor.
		If this list is empty, do nothing for $C$.
		Otherwise choose the non-adjacent subsequence greedily from left to right: set $i_1=h_1$, and after $i_r$ has been chosen, let $i_{r+1}$ be the smallest $h_j$ with $h_j\ge i_r+2$, if such an index exists.
		Then $i_1<i_2<\cdots<i_k$ is a maximal subsequence of non-adjacent levels, and every anchor-bearing level $h_j$ is either $i_r$ or $i_r+1$ for a unique $r$.
		
		For each selected level $i_r$, let $a_{i_r}$ be an anchor occurrence in $C$ from level $i_r$. Add the edges
		\[
		(s,a_{i_1}),
		(a_{i_1},a_{i_2}),
		\ldots,
		(a_{i_{k-1}},a_{i_k})
		\]
		to $G$. The last anchor point $a_{i_k}$ will be called the \EMPH{terminal anchor} of cone $C$.
		Their concatenation is called the \emph{radial path} of the cone $C$. See \Cref{fig:cone-path} for an illustration.
		
		\item\label{alg:shortcut2}
			For each selected level $i_r$, for each level $h\in \{i_r, i_r+1\}$, and every component $T\in F_h$, add the edge $(v,a_{i_r})$ for every $v\in\net(T)\cap C$.
	\end{enumerate}
\end{framed}

\begin{figure}[t]
	\centering
		\begin{tikzpicture}[
		scale=1.0,
		every node/.style={font=\small},
		cone/.style={fill=gray!4, draw=black!30, line width=.5pt},
		annulus boundary/.style={draw=black!30, line width=.45pt},
		selected annulus/.style={fill=blue!7, draw=none},
		skipped annulus/.style={fill=gray!10, draw=none},
		shortcut/.style={line width=1.15pt, red!75!black},
		ray/.style={draw=black!40, dashed, line width=.55pt},
		selanchor/.style={
			circle,
			draw=blue!70!black,
			fill=white,
			line width=.7pt,
			minimum size=3.4mm,
			inner sep=0pt
		},
		skippedanchor/.style={
			circle,
			draw=black!45,
			fill=gray!12,
			minimum size=2.7mm,
			inner sep=0pt
		},
		source/.style={
			circle,
			draw=black,
			fill=white,
			line width=.8pt,
			minimum size=4.2mm,
			inner sep=0pt
		}
		]
		
		\def\angA{5}
		\def\angB{42}
		\def\R{7.3}
		
		\def\rzero{1.0}
		\def\rone{2.2}
		\def\rtwo{3.4}
		\def\rthree{4.6}
		\def\rfour{5.8}
		\def\rfive{7.0}
		
		\fill[cone]
		(0,0) -- (\angA:\R)
		arc[start angle=\angA, end angle=\angB, radius=\R]
		-- cycle;
		
		\draw[ray] (0,0) -- (\angA:\R+.35);
		\draw[ray] (0,0) -- (\angB:\R+.35);
		
		\node[black!65] at (46:5.45) {$C$};
		
		\path[selected annulus]
		(\angA:\rzero)
		arc[start angle=\angA, end angle=\angB, radius=\rzero]
		-- (\angB:\rone)
		arc[start angle=\angB, end angle=\angA, radius=\rone]
		-- cycle;
		
		\path[skipped annulus]
		(\angA:\rone)
		arc[start angle=\angA, end angle=\angB, radius=\rone]
		-- (\angB:\rtwo)
		arc[start angle=\angB, end angle=\angA, radius=\rtwo]
		-- cycle;
		
		\path[selected annulus]
		(\angA:\rtwo)
		arc[start angle=\angA, end angle=\angB, radius=\rtwo]
		-- (\angB:\rthree)
		arc[start angle=\angB, end angle=\angA, radius=\rthree]
		-- cycle;
		
		\path[skipped annulus]
		(\angA:\rthree)
		arc[start angle=\angA, end angle=\angB, radius=\rthree]
		-- (\angB:\rfour)
		arc[start angle=\angB, end angle=\angA, radius=\rfour]
		-- cycle;
		
		\path[selected annulus]
		(\angA:\rfour)
		arc[start angle=\angA, end angle=\angB, radius=\rfour]
		-- (\angB:\rfive)
		arc[start angle=\angB, end angle=\angA, radius=\rfive]
		-- cycle;
		
		\foreach \r in {\rzero,\rone,\rtwo,\rthree,\rfour,\rfive} {
			\draw[annulus boundary]
			(\angA:\r) arc[start angle=\angA, end angle=\angB, radius=\r];
		}
		
		\coordinate (s)  at (0,0);
		
		\coordinate (a1) at (18:1.65); 
		\coordinate (b1) at (34:2.85); 
		\coordinate (a2) at (19:4.05); 
		\coordinate (b2) at (34:5.20); 
		\coordinate (a3) at (24:6.45); 
		
		\draw[shortcut] (s) -- (a1);
		\draw[shortcut] (a1) -- (a2);
		\draw[shortcut] (a2) -- (a3);
		
		\node[red!75!black, align=center] at (-10:3.75)
		{radial path added in Step~(1)};
		
		
		\begin{pgfonlayer}{foreground}
			\node[source, label=below left:{$s$}] at (s) {};
			
			\node[selanchor] at (a1) {};
			\node[selanchor] at (a2) {};
			\node[selanchor] at (a3) {};
			
			\node[skippedanchor] at (b1) {};
			\node[skippedanchor] at (b2) {};
			
			\node[below right=1pt and 1pt] at (a1) {$a_{i_1}$};
			\node[below right=1pt and 1pt] at (a2) {$a_{i_2}$};
			\node[above right=1pt and 1pt] at (a3) {$a_{i_3}$};
			
			\node[above=5pt] at (b1) {\scriptsize anchor in $i_1+1$};
			\node[above=5pt] at (b2) {\scriptsize anchor in $i_2+1$};
		\end{pgfonlayer}
		
	\end{tikzpicture}
	\caption{Shortcut scheme Step~(1). Within a fixed cone, we consider the anchor-bearing annular levels and select a maximal subsequence of non-adjacent levels \(i_1,i_2,\ldots\). For each selected level \(i_g\), let \(a_{i_g}\) be the corresponding anchor. Step~(1) adds the radial shortcut path through these anchors, namely the edges from \(s\) to \(a_{i_1}\) and from \(a_{i_g}\) to \(a_{i_{g+1}}\) for consecutive selected levels. Adjacent anchor-bearing levels are skipped and will be handled separately.}
	\label{fig:cone-path}
\end{figure}
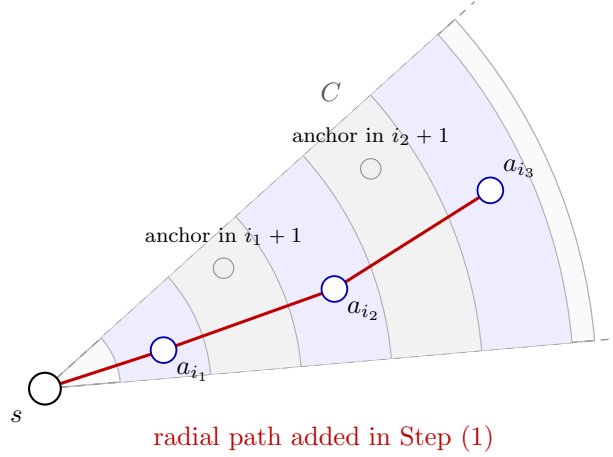

\subsection{Stretch Analysis}
\label{subsec:stretch-analysis}

We prove that the graph constructed above is a \((1+\epsilon)\)-shallow graph from the root. We use the following convention in the shortcut scheme.  In
Step~\ref{alg:shortcut1}, for a selected level $i_r$ in a cone $C$, the point
$a_{i_r}$ is the anchor $\anc(C,A_{i_r})$, namely the closest terminal to $s$ in
$C\cap A_{i_r}$.  This convention is needed because the graph $G$ has vertex
set $S$.

We use the parameter setting above.  In particular, for all sufficiently small
$\epsilon$,
\[
\alpha=\epsilon^{1/3},
\qquad 0<\beta\le \frac12,
\qquad 0<\gamma<\frac{1}{10},
\qquad
\alpha\le \frac{\gamma}{80},
\]
and the cone angle satisfies
\[
\theta\le (1-4\gamma)\epsilon\le \epsilon .
\]

\begin{lemma}
	\label{lem:radial-prefix-bound}
	Fix a cone $C$, and let
	\[
	a_{i_1},a_{i_2},\ldots,a_{i_k}
	\]
	be the selected anchors on the radial path of $C$.  For every $1\le r\le k$,
	\[
	\wts(s, a_{i_1})+
	\sum_{g=1}^{r-1}\wts(a_{i_g},a_{i_{g+1}})
	\le
	\wts(s, a_{i_r})+\frac{20\theta^2}{\alpha^2}\wts(s, a_{i_r}).
	\]
\end{lemma}

\begin{proof}
	Consider two consecutive selected anchors $x=a_{i_g}$ and $y=a_{i_{g+1}}$.
	Their levels differ by at least two, so $i_{g+1}\ge i_g+2$.  Since
	$x\in A_{i_g}$ and $y\in A_{i_{g+1}}$, we have
	\[
	\wts(s, x)<R_{i_g},
	\qquad
	\wts(s, y)\ge L_{i_g+2}.
	\]
	Therefore
	\[
	\frac{\wts(s, y)}{\wts(s, x)}
	>
	\frac{L_{i_g+2}}{R_{i_g}}
	=
	\frac{(1+\alpha)(1+j^\star\beta\alpha)}
	{1+(j^\star+1)\beta\alpha}.
	\]
	Moreover,
	\[
	\frac{(1+\alpha)(1+j^\star\beta\alpha)}
	{1+(j^\star+1)\beta\alpha}
	=1+
	\frac{\alpha(1-\beta)+j^\star\beta\alpha^2}
	{1+(j^\star+1)\beta\alpha}
	\ge 1+\frac{\alpha}{4},
	\]
	where we used $0<\beta\le 1/2$, $j^\star+1\le 1/\beta$, and $\alpha\le 1$.
	Thus
	\begin{equation}
		\label{eq:two-level-radial-gap}
		\wts(s, y)-\wts(s, x)
		\ge \frac{\alpha}{4}\wts(s, x).
	\end{equation}
	
	Let $\varphi$ be the angle between the rays $\overrightarrow{sx}$ and $\overrightarrow{sy}$.  Since $x$ and $y$
	lie in the same cone, $\varphi\le \theta$, by the law of cosines and the inequality
	$1-\cos\varphi\le \varphi^2/2$,
	\[
	\wts(x,y)^2
	= (\wts(s, y)-\wts(s, x))^2+2\wts(s, x)\wts(s, y)(1-\cos\varphi)
	\le (\wts(s, y)-\wts(s, x))^2+\wts(s, x)\wts(s, y)\theta^2.
	\]
	Hence
	\[
	\wts(x,y)-(\wts(s, y)-\wts(s, x))
	\le
	\frac{\wts(s, x)\wts(s, y)\theta^2}{\wts(s, y)-\wts(s, x)}
	\le
	\frac{4\theta^2}{\alpha}\wts(s, y),
	\]
	where the last inequality follows from \eqref{eq:two-level-radial-gap}. This gives
	\begin{equation}
		\label{eq:one-radial-step}
		\wts(s, x)+\wts(x,y)
		\le
		\wts(s, y)+\frac{4\theta^2}{\alpha}\wts(s, y).
	\end{equation}
	
	Summing \eqref{eq:one-radial-step} over $g=1,\ldots,r-1$ yields
	\[
	\wts(s, a_{i_1})+
	\sum_{g=1}^{r-1}\wts(a_{i_g},a_{i_{g+1}})
	\le
	\wts(s, a_{i_r})+\frac{4\theta^2}{\alpha}
	\sum_{g=2}^{r}\wts(s, a_{i_g}).
	\]
	The same radial-gap bound \eqref{eq:two-level-radial-gap} implies that the
	selected-anchor radii increase by a factor at least $1+\alpha/4$ at every step.
	Therefore
	\[
	\sum_{g=2}^{r}\wts(s, a_{i_g})
	\le
	\left(1+\frac{4}{\alpha}\right)\wts(s, a_{i_r})
	\le
	\frac{5}{\alpha}\wts(s, a_{i_r}),
	\]
	because $\alpha\le 1$.  Substituting this into the previous inequality proves
	the lemma.
\end{proof}

\begin{lemma}
	\label{lem:stretch}
	For every terminal $t\in S$, we have
	\[
	\dist_G(s,t)
	\le
	(1+\epsilon)\wts(s,t).
	\]
\end{lemma}

\begin{proof}
	The statement is trivial for $t=s$, so assume $t\ne s$.
	By \Cref{lem:tree-contain}, choose a level $h$ and a component
	$T\in F_h$ such that $t\in S(T)$.
		
	Let $u\in\net(T)$ be a net point minimizing the tree distance to $t$.  By the
	net property,
	\begin{equation}
		\label{eq:net-close}
		\dist_T(u,t)
		\le
		\gamma\epsilon L_h
		\le
		\gamma\epsilon\wts(s, t).
	\end{equation}
	Let $C$ be the cone containing $u$.  Since $u\in S(T)\subseteq A_h$, level $h$
	is anchor-bearing in the cone $C$.  By the greedy choice of the non-adjacent
	subsequence in Step~\ref{alg:shortcut1}, there is a selected level $i_r$ such
	that
	\[
	h\in\{i_r,i_r+1\}.
	\]
	Let
	\(
	b=a_{i_r}=\anc(C,A_{i_r}).
	\)
	By Step~\ref{alg:shortcut2}, the edge $(u,b)$ belongs to $G$. We first argue that
	\begin{equation}
		\label{eq:b-not-farther-than-u}
		\wts(s, b)\le \wts(s, u).
	\end{equation}
	If $h=i_r$, then both $b$ and $u$ lie in $C\cap A_{i_r}$, and $b$ is the
	closest terminal to $s$ in this set.  If $h=i_r+1$ and $u\in A_{i_r}$, the same
	argument applies.  Finally, if $h=i_r+1$ and $u\notin A_{i_r}$, then
	$\wts(s, u)\ge R_{i_r}$ whereas $b\in A_{i_r}$ implies $\wts(s, b)<R_{i_r}$.
	This proves \eqref{eq:b-not-farther-than-u}.
	
	Let $P$ be the path obtained by following the radial path of $C$ from $s$ to
	$b$, then the shortcut edge $(b,u)$, and finally the tree path in $T$ from $u$
	to $t$.  By \Cref{lem:radial-prefix-bound}, the radial prefix from $s$ to $b$
	has length at most
	\[
	\wts(s, b)+\frac{20\theta^2}{\alpha^2}\wts(s, b).
	\]
	Using \eqref{eq:net-close}, we have
	\[
	\wts(s, u)
	\le
	\wts(s, t)+\wts(u,t)
	\le
	\wts(s, t)+\dist_T(u,t)
	\le
	(1+\gamma\epsilon)\wts(s, t)
	\le 2\wts(s, t).
	\]
	Together with \eqref{eq:b-not-farther-than-u}, $\theta\le\epsilon$, and
	$\alpha=\epsilon^{1/3}$, this gives
	\[
	\frac{20\theta^2}{\alpha^2}\wts(s, b)
	\le
	20\epsilon^{4/3}\wts(s, u)
	\le
	40\alpha\epsilon\wts(s, t)
	\le
	\frac{\gamma}{2}\epsilon\wts(s, t),
	\]
	where the last inequality uses $\alpha\le\gamma/80$.
	Therefore the radial prefix has length at most
	\begin{equation}
		\label{eq:radial-prefix-final}
		\wts(s, b)+\frac{\gamma}{2}\epsilon\wts(s, t).
	\end{equation}
	
	Since $b$ and $u$ lie in the same cone and $\wts(s, b)\le\wts(s, u)$, we can connect
	$b$ to $u$ by first moving along the circle centered at $s$ of radius $\wts(s, b)$
	and then moving radially outward.  The straight segment is no longer than this
	broken path, so
	\begin{equation}
		\label{eq:bu-bound}
		\wts(b,u)
		\le
		\wts(s, u)-\wts(s, b)+\theta\wts(s, b).
	\end{equation}
	Combining \eqref{eq:net-close}, \eqref{eq:radial-prefix-final}, and
	\eqref{eq:bu-bound}, we obtain
	\[
	\begin{aligned}
		\wts(P)
		&\le
		\wts(s, b)+\frac{\gamma}{2}\epsilon\wts(s, t)
		+\wts(s, u)-\wts(s, b)+\theta\wts(s, b)
		+\dist_T(u,t) \\
		&\le
		\wts(s, u)+\theta\wts(s, b)
		+\frac{\gamma}{2}\epsilon\wts(s, t)
		+\gamma\epsilon\wts(s, t).
	\end{aligned}
	\]
	Using again $\wts(s, u)\le(1+\gamma\epsilon)\wts(s, t)$ and
	$\wts(s, b)\le\wts(s, u)$, we get
	\[
	\begin{aligned}
		\wts(P)
		&\le
		\wts(s, t)+\gamma\epsilon\wts(s, t)
		+(1-4\gamma)\epsilon(1+\gamma\epsilon)\wts(s, t)
		+\frac{\gamma}{2}\epsilon\wts(s, t)
		+\gamma\epsilon\wts(s, t) \\
		&\le
		\wts(s, t)+
		\left(2\gamma+(1-4\gamma)(1+\gamma)+\frac{\gamma}{2}\right)
		\epsilon\wts(s, t) \\
		&=
		\wts(s, t)+
		\left(1-\frac{\gamma}{2}-4\gamma^2\right)
		\epsilon\wts(s, t) \\
		&\le
		(1+\epsilon)\wts(s, t).
	\end{aligned}
	\]
	Thus $\dist_G(s,t)\le \wts(P)\le(1+\epsilon)\wts(s,t)$, as claimed.
\end{proof}

\subsection{Weight Analysis}
\label{subsec:weight-analysis}

We prove that the shortcut edges added by the revised construction have total
weight at most
\[
    \left(\frac{2\pi}{\sqrt{4\pi^2+1}}+o_{\epsilon}(1)\right)
    \frac{1}{\epsilon}\,\wts(\mst).
\]
We count duplicate shortcut edges only once, since $G$ is a graph with an edge
set rather than a multiset of edges.

Recall the parameter choice from the beginning of the section:
\[
    \alpha=\epsilon^{1/3},
    \qquad
    \beta=\frac{1}{\lceil \epsilon^{-1/12}\rceil},
    \qquad
    \gamma=\epsilon^{1/4},
\]
and the cone angle $\theta$ satisfies
\[
    (1-5\gamma)\epsilon\le \theta\le (1-4\gamma)\epsilon.
\]
In particular,
\[
    \beta=o_\epsilon(1),\qquad \gamma=o_\epsilon(1),\qquad
    \frac{\theta}{\epsilon}=1+o_{\epsilon}(1),
\]
and the error terms that will appear below satisfy
\[
    \frac{\epsilon^{1/6}}{\beta}=o_{\epsilon}(1),
    \qquad
    \frac{\epsilon^{1/3}}{\gamma}=o_{\epsilon}(1).
\]
For brevity put
\[
    c_\star=\frac{2\pi}{\sqrt{4\pi^2+1}}.
\]

\paragraph{Active and inactive paths.}
We next classify all shortcut edges as \EMPH{active} or \EMPH{inactive}.  
The only active edges will be radial-path edges from Step~\ref{alg:shortcut1}; 
every edge added in Step~\ref{alg:shortcut2} will be inactive.

For a cone $C$ whose radial path is nonempty, let $P(C)$ denote the whole radial
path added in Step~\ref{alg:shortcut1}.  Let $\anc(C)$ be the last selected anchor
on this path; we call $\anc(C)$ the \EMPH{terminal anchor} of $P(C)$.  Thus there
is at most one terminal anchor for each cone.

By \Cref{lem:tree-contain}, every terminal anchor belongs to $S(T)=T\cap S$
for at least one retained component $T\in F_h$ at some level $h$.  If there is
more than one possible choice, for instance because the anchor lies in an
overlap annulus, choose one such component once and for all.  We then say that
both the terminal anchor $\anc(C)$ and the radial path $P(C)$ are \emph{owned} by
$T$.  With this convention, if $\anc(C)$ is owned by a component of $F_h$, then
$\anc(C)\in A_h$.

We use a lifted angular order inside each retained component.  The topological
realization of a component $T\in F_h$ is a tree and is disjoint from the root
$s$.  Hence the polar-angle map $T\to \mathbb R/2\pi\mathbb Z$ admits a
continuous lift $\varphi_T:T\to\mathbb R$, unique up to adding an
integer multiple of $2\pi$.  We fix one such lift for each $T$.  To avoid
irrelevant ties, assume the cone partition has been infinitesimally rotated so
that no terminal lies on a cone boundary; equivalently, use the corresponding
half-open convention.  Since only finitely many terminal directions are involved,
this perturbation does not affect any estimate.

At the beginning, declare every radial path active.  We then process each
retained component $T\in F_h$ and possibly deactivate some of the radial paths
owned by $T$.  Let
\[
    A(T)=\{\anc(C) : P(C)\text{ is owned by }T\}
\]
be the set of terminal anchors owned by $T$, and list these anchors in increasing
lifted angular order:
\[
    a_1,a_2,\ldots,a_m,
    \qquad
    \varphi_T(a_1)<\varphi_T(a_2)<\cdots<
    \varphi_T(a_m).
\]
Let $C_r$ be the geometric cone whose radial path has terminal anchor $a_r$.

We now distinguish two cases.
\begin{enumerate}[(a),leftmargin=*]
    \item \textbf{Few terminal anchors.}
    If $m\le \epsilon^{-1/2}$, then we deactivate every radial path
    $P(C_r)$, $1\le r\le m$.  These paths are called
    \EMPH{type-(a) inactive}.

    \item \textbf{Many terminal anchors.}
    If $m>\epsilon^{-1/2}$, then only the two endpoint anchors in the lifted
    order are discarded: we deactivate $P(C_1)$ and $P(C_m)$.  These two paths
    are called \EMPH{type-(b) inactive}.  Every remaining path
    $P(C_r)$, $1<r<m$, stays active.
\end{enumerate}
This is the only source of active paths.  Thus an active radial path is exactly
a path $P(C_r)$ belonging to a component $T$ with many terminal anchors, where
its terminal anchor is not one of the two endpoints in the lifted angular order
for that component.

For each active path $P(C_r)$, with $1<r<m$, we will charge it to a crossing
subpath of $T$ inside the same geometric cone $C_r$.  The next lemma is the
formal version of this crossing statement.

\begin{lemma}
\label{lem:cone-crossing}
Let $T\in F_h$ be a retained component with many terminal anchors, listed in the
lifted angular order above, and let $P(C_r)$ be active with $1<r<m$.  Then there
are points $u,v\in T$ lying on the two boundary rays of the geometric cone
$C_r$ such that the tree path $T[u,v]$ is contained in $T\cap C_r$.  Moreover,
$u,v\in A_h$, and hence $\wts(s,u),\wts(s,v)\ge L_h$.
\end{lemma}

\begin{proof}
Write $C_r$ as the geometric cone with lower angular boundary $\lambda$ and
upper angular boundary $\lambda+\theta$, modulo $2\pi$.  Let
\[
    I=[\lambda+2\pi q,\lambda+\theta+2\pi q]
\]
be the lifted copy of this cone that contains
$\varphi_T(a_r)$.  Because there is only one terminal anchor for each
geometric cone, no other owned anchor lies in the interior of this same lifted
copy.  Since $a_r$ is not an endpoint in the lifted order, the neighboring owned
anchors $a_{r-1}$ and $a_{r+1}$ satisfy
\[
    \varphi_T(a_{r-1})<\lambda+2\pi q
    \quad\text{and}\quad
    \varphi_T(a_{r+1})>\lambda+\theta+2\pi q .
\]

Parametrize the unique tree path in $T$ from $a_{r-1}$ to $a_{r+1}$ by a
continuous curve $\Gamma:[0,1]\to T$.  The lifted angle
$\varphi_T(\Gamma(t))$ is continuous, starts below the lower endpoint
of $I$, and ends above the upper endpoint of $I$.  Let $t_+$ be the first time
at which this lifted angle reaches $\lambda+\theta+2\pi q$, and let $t_-$ be
the last time before $t_+$ at which it equals $\lambda+2\pi q$.  By the choices
of $t_-$ and $t_+$, the subpath $\Gamma([t_-,t_+])$ has lifted angle contained
in $I$; geometrically, it is contained in the cone $C_r$ and has endpoints on
the two boundary rays of $C_r$.  These endpoints are the desired points $u$ and
$v$.

Finally, $T$ is a component of the forest $F_h$ obtained by restricting the
subdivided MST to $A_h$.  Hence the whole subpath, including $u$ and $v$, lies
in $A_h$.
\end{proof}

All remaining shortcut edges are inactive by definition.  More precisely,
every edge added in Step~\ref{alg:shortcut2} is classified as follows.
\begin{enumerate}[(a),start=3,leftmargin=*]
    \item \textbf{Net point shortcuts.}
    Suppose the selected level $i_r$ in cone $C$ covers level
    $h\in\{i_r,i_r+1\}$, and let $T\in F_h$.  Then every shortcut
    $(v,a_{i_r})$ with $v\in\net(T)\cap C$ is called
    \EMPH{type-(c) inactive}.

\end{enumerate}
\begin{lemma}
\label{lem:active}
The total weight of all active radial paths is at most
\[
    \left(c_\star+o_{\epsilon}(1)\right)
    \frac{1}{\epsilon}\,\wts(\mst).
\]
\end{lemma}

\begin{proof}
Fix an active radial path $P=P(C)$, and let $a=\anc(C)$ be its terminal anchor.
Suppose $a$ is owned by a component $T\subseteq F_h$.  By
\Cref{lem:cone-crossing}, choose a crossing path
$\mst[u_P,v_P]\subseteq T\cap C$ with endpoints on the two boundary rays of
$C$.  Let $x_P$ and $y_P$ be points of $\mst[u_P,v_P]$ minimizing and maximizing
$\wts(s,z), z\in \mst[u_P,v_P]$, respectively.  Let $p_P,q_P,r_P$ be the orthogonal
projections of $v_P,x_P,y_P$, respectively, onto the boundary ray of $C$
containing $u_P$.  Let $e_P$ and $f_P$ be the projections of $x_P$ and $y_P$
onto the line through $v_P$ perpendicular to that boundary ray.
We first prove the Euclidean excess estimate used in the charging argument.
Let $\ell_P$ be the boundary ray of $C$ containing $u_P$.  Let $\mathbf t$ be
the unit vector along $\ell_P$, directed away from $s$, and let $\mathbf n$ be
the unit vector perpendicular to $\ell_P$ pointing into the cone $C$. See \Cref{fig:cone-project} for an illustration.

\begin{figure}[t]
	\centering
	\begin{tikzpicture}[
	scale=1.0,
	>=Latex,
	point/.style={circle,fill=black,inner sep=1.3pt},
	proj/.style={densely dashed,thin},
	treepath/.style={very thick},
	measure/.style={<->,thin},
	every node/.style={font=\small}
	]
	
	\coordinate (s)        at (-4,0);
	\coordinate (ellend)   at (9.4,0);
	\coordinate (upperend) at (9.4,4.38); 
	
	\coordinate (u) at (2.45,0);       
	\coordinate (v) at (6.30,3.37);    
	\coordinate (x) at (1.80,0.55);    
	\coordinate (y) at (8.10,1.50);    
	
	\coordinate (p) at (6.30,0);       
	\coordinate (q) at (1.80,0);       
	\coordinate (r) at (8.10,0);       
	
	\coordinate (e) at (6.30,0.55);    
	\coordinate (f) at (6.30,1.50);    
	
	\fill[gray!10] (s) -- (ellend) -- (upperend) -- cycle;
	\draw[thick] (s) -- (ellend) node[below right] {$\ell_P$};
	\draw[thick] (s) -- (upperend);
	\node[point,label=below left:$s$] at (s) {};
	\draw (-3.1,0) arc[start angle=0,end angle=19,radius=0.9];
	\node at (-2.58,0.22) {$\theta$};
	
	\draw[->] (-1,2) -- +(0.70,0) node[right] {$\mathbf t$};
	\draw[->] (-1,2) -- +(0,0.70) node[above] {$\mathbf n$};
	
	\draw[thin] (6.30,-0.18) -- (6.30,3.55)
	node[above] {$m_P\perp\ell_P$};
	
	\draw[treepath]
	(u) .. controls (2.15,0.05) and (1.80,0.25) ..
	(x) .. controls (3.35,1.35) and (6.40,0.78) ..
	(y) .. controls (8,1.72) and (7.40,3.05) ..
	(v);
	\node at (4.15,1.5) {$T[u_P,v_P]\subseteq C$};
	
	
	\draw[proj] (v) -- (p);
	\draw[proj] (x) -- (q);
	\draw[proj] (y) -- (r);
	\draw[proj] (x) -- (e);
	\draw[proj] (y) -- (f);
	
	\draw[measure] (1.80,-0.6) -- node[below] {$A_P$} (8.10,-0.6);
	\draw[measure] (6.03,0.55) -- node[above left] {$B_P$} (6.03,1.50);
	
	\node[point,label=below right:$u_P$] at (u) {};
	\node[point,label=left:$v_P$] at (v) {};
	\node[point,label=above left:$x_P$] at (x) {};
	\node[point,label=above right:$y_P$] at (y) {};
	
	\node[point,label=below:$p_P$] at (p) {};
	\node[point,label=below left:$q_P$] at (q) {};
	\node[point,label=below:$r_P$] at (r) {};
	\node[point,label=below right:$e_P$] at (e) {};
	\node[point,label=above right:$f_P$] at (f) {};
	
	\draw (1.80,0.12) -- (1.92,0.12) -- (1.92,0);
	\draw (8.10,0.12) -- (8.22,0.12) -- (8.22,0);
	\draw (6.30,0.12) -- (6.42,0.12) -- (6.42,0);
	\draw (6.30,0.67) -- (6.18,0.67) -- (6.18,0.55);
	\draw (6.30,1.62) -- (6.18,1.62) -- (6.18,1.50);
	
\end{tikzpicture}
	\caption{The points $p_P,q_P,r_P$ are projected onto the cone boundary ray $\ell_P$, while $e_P,f_P$ are projected onto the line through $v_P$ perpendicular to $\ell_P$.}
	\label{fig:cone-project}
		
\end{figure}
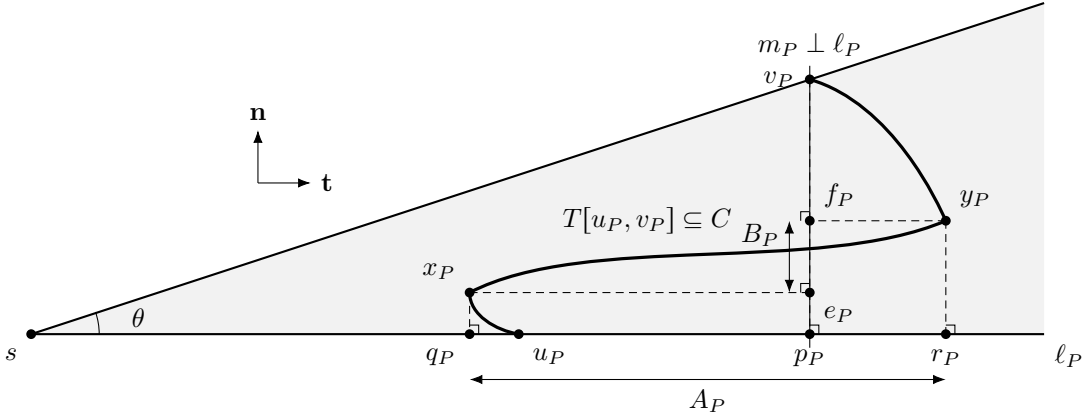

For
$z\in C$, define
\[
z=s+\xi(z)\mathbf t+\zeta(z)\mathbf n .
\]
Since $C$ has angle less than $\pi/2$ for sufficiently small $\epsilon>0$, the
orthogonal projection of every point of $C$ onto the line spanned by $\ell_P$
lies on the ray $\ell_P$.

Set parameters:
\[
H_P=\zeta(v_P)=\wts(v_P,p_P),
\qquad
A_P=|\xi(y_P)-\xi(x_P)|=\wts(q_P,r_P),
\qquad
B_P=|\zeta(y_P)-\zeta(x_P)|=\wts(e_P,f_P).
\]
Traverse the tree path $\mst[u_P,v_P]$ from $u_P$ to $v_P$, and let
$z_1,z_2$ be the two points $x_P,y_P$ in the order in which they are first
encountered along this traversal.  The middle subpath satisfies
\[
\wts(\mst[z_1,z_2])
\ge
\wts(x_P,y_P)
=
\sqrt{A_P^2+B_P^2}.
\]
The two outside subpaths have length at least their total change in the
$\zeta$-coordinate, so
\[
\begin{aligned}
	\wts(\mst[u_P,z_1])+\wts(\mst[z_2,v_P])
	&\ge
	|\zeta(z_1)-\zeta(u_P)|
	+
	|\zeta(v_P)-\zeta(z_2)| .
\end{aligned}
\]
Because $u_P\in \ell_P$, we have $\zeta(u_P)=0$, while
$\zeta(v_P)=H_P$.  By the triangle inequality on the real line,
\[
\begin{aligned}
	H_P
	&=
	|\zeta(v_P)-\zeta(u_P)|  \\
	&\le
	|\zeta(z_1)-\zeta(u_P)|
	+
	|\zeta(z_2)-\zeta(z_1)|
	+
	|\zeta(v_P)-\zeta(z_2)|  \\
	&=
	|\zeta(z_1)-\zeta(u_P)|
	+
	B_P
	+
	|\zeta(v_P)-\zeta(z_2)| .
\end{aligned}
\]
Therefore
\[
\wts(\mst[u_P,z_1])+\wts(\mst[z_2,v_P])
\ge
H_P-B_P .
\]
Combining the middle and outside subpaths gives
\begin{equation}
	\label{eq:single-active-charge-revised}
	\wts(\mst[u_P,v_P])
	\ge
	\wts(v_P,p_P)
	+\bigl(\sqrt{\wts(q_P,r_P)^2+\wts(e_P,f_P)^2}
	-\wts(e_P,f_P)\bigr).
\end{equation}

For each active path, define
\[
A_P=\wts(q_P,r_P),
\qquad
B_P=\wts(e_P,f_P).
\]
Since $A_P,B_P\ge0$, the triangle inequality in $\mathbb R^2$ gives
\[
\begin{aligned}
\sqrt{\brac{\sum_{P\text{ active}}\wts(q_P,r_P)}^2+\brac{\sum_{P\text{ active}}\wts(e_P,f_P)}^2}
&=
\left\|
\sum_{P\text{ active}} (A_P,B_P)
\right\|_2
\le
\sum_{P\text{ active}} \|(A_P,B_P)\|_2\\
&=
\sum_{P\text{ active}} \sqrt{A_P^2+B_P^2}.
\end{aligned}
\]
Subtracting $\sum_{P\text{ active}}\wts(e_P,f_P)$ from both sides yields
\[
\begin{aligned}
	\label{eq:exact-active-excess-revised}
	&\sum_{P\text{ active}}\bigl(\sqrt{\wts(q_P,r_P)^2+\wts(e_P,f_P)^2}-\wts(e_P,f_P)\bigr)\\
	&\ge
	\sqrt{\brac{\sum_{P\text{ active}}\wts(q_P,r_P)}^2+\brac{\sum_{P\text{ active}}\wts(e_P,f_P)}^2}-\sum_{P\text{ active}}\wts(e_P,f_P) .
\end{aligned}
\]

Let $L=\max_{t\in S}\wts(s,t)$, and let
\[
    X=\mst\setminus\bigcup_{P\text{ active}} \mst[u_P,v_P].
\]
For each active path,
\[
    \wts(q_P,r_P)
    \ge
    \wts(s,y_P)\cos\theta-\wts(s,x_P)
    \ge
    \wts(s,y_P)-\wts(s,x_P)-\frac{\theta^2}{2}L .
\]
There are exactly $2\pi/\theta$ cones under the divisibility convention, so
\[
\begin{aligned}
\label{eq:radial-progress-revised}
    \sum_{P\text{ active}}\wts(q_P,r_P)
    &\ge
    \sum_{P\text{ active}}(\wts(s,y_P)-\wts(s,x_P))-
    \pi\theta L\\
    &\ge
    \sum_{P\text{ active}}(\wts(s,y_P)-\wts(s,x_P))-\pi\epsilon L .
\end{aligned}
\]

Let $Q^{\mathrm{root}}$ be the path in $\mst$ from $s$ to a farthest terminal.
As one traverses $Q^{\mathrm{root}}$, the radial-distance function starts at
$0$ and reaches $L$.  Whenever a point of $Q^{\mathrm{root}}$ lies outside the
active crossing subpaths, its radial variation is paid for by the length of
$X$; whenever it lies inside an active crossing subpath $\mst[u_P,v_P]$, all
radii attained there lie in the interval $[\wts(s,x_P),\wts(s,y_P)]$ by the
definitions of $x_P$ and $y_P$.  Thus these intervals, together with the
portions of $Q^{\mathrm{root}}$ contained in $X$, cover every radius from $0$ to
$L$.  Since the radial-distance function is continuous along $\mst$, we have
\begin{equation}
\label{eq:cover-radius-revised}
    \sum_{P\text{ active}}(\wts(s,y_P)-\wts(s,x_P))+\wts(X)
    \ge L .
\end{equation}
Combining \eqref{eq:radial-progress-revised} and
\eqref{eq:cover-radius-revised},
\begin{equation}
\label{eq:q-plus-x-revised}
    \sum_{P\text{ active}}\wts(q_P,r_P)+\wts(X)
    \ge
    (1-\pi\epsilon)L .
\end{equation}
Also, since each active crossing lies in a cone of aperture $\theta$,
\[
    \wts(e_P,f_P)\le L\sin\theta
\]
for every active path.  Therefore
\begin{equation}
\label{eq:ef-sum-revised}
    \sum_{P\text{ active}}\wts(e_P,f_P)\le \frac{2\pi}{\theta}L\sin\theta\le 2\pi L .
\end{equation}

Using \eqref{eq:exact-active-excess-revised}, \eqref{eq:q-plus-x-revised}, and
\eqref{eq:ef-sum-revised}, we next argue that
\begin{equation}
	\label{eq:excess-plus-x-revised}
	\sum_{P\text{ active}}
	\left(\sqrt{\wts(q_P,r_P)^2+\wts(e_P,f_P)^2}-\wts(e_P,f_P)\right)
	+\wts(X)
	\ge
	\left(\sqrt{4\pi^2+(1-\pi\epsilon)^2}-2\pi\right)L .
\end{equation}
By \eqref{eq:q-plus-x-revised} and \eqref{eq:ef-sum-revised}, it says that
\[
\sum_{P\text{ active}}\wts(q_P,r_P)+\wts(X)\ge (1-\pi\epsilon)L.
\]
Since the function
\[
x\longmapsto \sqrt{\brac{\sum_{P\text{ active}}\wts(q_P,r_P)}^2+x^2}-x
\]
is nonincreasing for $x\ge0$, we have
\[
\sqrt{\brac{\sum_{P\text{ active}}\wts(q_P,r_P)}^2+x^2}-x
\ge
\sqrt{\brac{\sum_{P\text{ active}}\wts(q_P,r_P)}^2+4\pi^2L^2}-2\pi L.
\]
for any $x\leq 2\pi L$. Therefore, we have
\[
\begin{aligned}    
\sum_{P\text{ active}}
	&\left(\sqrt{\wts(q_P,r_P)^2+\wts(e_P,f_P)^2}-\wts(e_P,f_P)\right)
	+\wts(X)\\
    &\geq \sqrt{\brac{\sum_{P\text{ active}}\wts(q_P,r_P)}^2+4\pi^2L^2}-2\pi L + \wts(X) .
\end{aligned}
\]
It remains to lower-bound the last expression under the constraint
$\sum_{P\text{ active}}\wts(q_P,r_P)+\wts(X)\ge (1-\pi\eps)L$.  If $\sum_{P\text{ active}}\wts(q_P,r_P)\ge (1-\pi\eps)L$, then
\[
\sqrt{\brac{\sum_{P\text{ active}}\wts(q_P,r_P)}^2+4\pi^2L^2}-2\pi L+\wts(X)
\ge
\sqrt{(1-\pi\eps)^2L^2+4\pi^2L^2}-2\pi L.
\]
If $\sum_{P\text{ active}}\wts(q_P,r_P)<(1-\pi\eps)L$, then $\wts(X)\ge (1-\pi\eps)L-\sum_{P\text{ active}}\wts(q_P,r_P)$, and therefore
\[
\begin{aligned}
	&\sqrt{\brac{\sum_{P\text{ active}}\wts(q_P,r_P)}^2+4\pi^2L^2}-2\pi L+\wts(X)\\
	&\ge
	\sqrt{\brac{\sum_{P\text{ active}}\wts(q_P,r_P)}^2+4\pi^2L^2}-2\pi L+(1-\pi\eps)L-\sum_{P\text{ active}}\wts(q_P,r_P) \\
	&=
	(1-\pi\eps)L-2\pi L+\brac{\sqrt{\brac{\sum_{P\text{ active}}\wts(q_P,r_P)}^2+4\pi^2L^2}-\sum_{P\text{ active}}\wts(q_P,r_P)}.
\end{aligned}
\]
The function
\[
x\longmapsto \sqrt{x^2+4\pi^2L^2}-x
\]
is nonincreasing for $x\ge0$, so for $x<(1-\pi\eps)L$ we have
\[
\sqrt{x^2+4\pi^2L^2}-x
\ge
\sqrt{(1-\pi\eps)^2L^2+4\pi^2L^2}-(1-\pi\eps)L.
\]
Thus again
\[
\sqrt{\brac{\sum_{P\text{ active}}\wts(q_P,r_P)}^2+4\pi^2L^2}-2\pi L+\wts(X)
\ge
\sqrt{(1-\pi\eps)^2L^2+4\pi^2L^2}-2\pi L,
\]
which proves
\eqref{eq:excess-plus-x-revised}.

The active crossing paths overlap only inside the sparse buffer annuli.  Indeed,
inside a fixed retained component they lie in distinct cone interiors; different
components of one forest $F_h$ are edge-disjoint; and two different annular
regions can share edges only in their selected buffer overlap.  Consequently the
quantity $\sum_P\wts(\mst[u_P,v_P])+\wts(X)$ is dominated by the total MST
weight plus the selected-buffer double count, and therefore, by
\eqref{eq:annulus-overlap},
\begin{align}
    (1+\beta)\wts(\mst)
    &\ge
    \sum_{P\text{ active}}\wts(\mst[u_P,v_P])+\wts(X) \notag\\
    &\ge
    W+
    \left(\sqrt{4\pi^2+(1-\pi\epsilon)^2}-2\pi\right)L,
\label{eq:mst-lower-active-revised}
\end{align}
where
\[
    W=\sum_{P\text{ active}}\wts(v_P,p_P).
\]
There is at most one active radial path in each cone.  Moreover every point of
$\mst$ lies in the disk of radius $L$ centered at $s$, because each MST edge is
a segment between terminals in that disk.  Thus
$\wts(v_P,p_P)=\wts(s,v_P)\sin\theta\le L\sin\theta$ for each active path, and
hence $W\le (2\pi/\theta)L\sin\theta\le 2\pi L$.  Therefore
\eqref{eq:mst-lower-active-revised} implies
\begin{equation}
\label{eq:vp-bound-revised}
    W
    \le
    \frac{1+\beta}
    {1+\frac{\sqrt{4\pi^2+(1-\pi\epsilon)^2}-2\pi}{2\pi}}
    \wts(\mst).
\end{equation}

It remains to compare an active radial path with $W$. Let $P=P(C)$ be active, let $a=\anc(C)$ be its terminal anchor, and suppose $a$ is
owned by a component of $F_h$.  Then $a\in A_h$.  The endpoint $v_P$ of the
charged crossing path also lies in $A_h$, by \Cref{lem:cone-crossing}.  Hence
\[
    \wts(s,a)
    \le R_h
    \le (1+3\alpha)L_h
    \le (1+3\alpha)\wts(s,v_P).
\]
Furthermore,
\[
    \wts(v_P,p_P)=\wts(s,v_P)\sin\theta .
\]
By the radial-prefix bound from \Cref{lem:radial-prefix-bound},
\[
    \wts(P)
    \le
    \left(1+\frac{20\theta^2}{\alpha^2}\right)\wts(s,a).
\]
Consequently,
\begin{equation}
\label{eq:active-path-to-width-revised}
    \wts(P)
    \le
    \frac{(1+20\theta^2/\alpha^2)(1+3\alpha)}{\sin\theta}
    \wts(v_P,p_P).
\end{equation}
Combining \eqref{eq:vp-bound-revised} and
\eqref{eq:active-path-to-width-revised},
\[
\sum_{P\text{ active}}\wts(P)
\le
\frac{(1+\beta)(1+20\theta^2/\alpha^2)(1+3\alpha)}
{\sin\theta\left(1+\frac{\sqrt{4\pi^2+(1-\pi\epsilon)^2}-2\pi}{2\pi}\right)}
\wts(\mst).
\]
Multiplying the coefficient by $\epsilon$, and using
$\alpha=\epsilon^{1/3}$, $\beta=o_\epsilon(1)$, $\theta/\epsilon=1+o_{\epsilon}(1)$,
and $\gamma=o_\epsilon(1)$, gives
\[
\begin{aligned}
\epsilon\cdot
\frac{(1+\beta)(1+20\theta^2/\alpha^2)(1+3\alpha)}
{\sin\theta\left(1+\frac{\sqrt{4\pi^2+(1-\pi\epsilon)^2}-2\pi}{2\pi}\right)} = 
\left(1+o_{\epsilon}(1)\right)
\frac{2\pi}{\sqrt{4\pi^2+1}}
= c_\star+o_{\epsilon}(1).
\end{aligned}
\]
Therefore the active radial paths have total weight at most
\[
    \left(c_\star+o_{\epsilon}(1)\right)
    \frac{1}{\epsilon}\,\wts(\mst).
    \qedhere
\]
\end{proof}

\paragraph{Inactive shortcuts.}
We now bound all inactive edges.  We use the following elementary lower bound on
the weight of a connected component.

\begin{lemma}
\label{lem:tree-weight-lb}
Let $T\subseteq F_h$ be a connected component.  Then
\[
    \wts(T)\ge \frac{\beta\alpha}{2}L_h .
\]
\end{lemma}

\begin{proof}
The component $T$ contains a boundary point of $A_h$: since $s\notin A_h$ and
$\mst$ connects every point of $T$ to $s$, the original MST path from any point
of $T$ to $s$ must leave $A_h$, and after subdivision it leaves through a
boundary vertex of the restricted component.  The retained component also
intersects $A_h\setminus B_h$.  Therefore a path in $T$ from the boundary point
to $A_h\setminus B_h$ crosses either the inner selected buffer annulus
$A_h^{(j^\star)}$ or the outer selected buffer annulus $A_{h+1}^{(j^\star)}$.
The radial width of the inner selected buffer is
\[
    (1+\alpha)^h\beta\alpha
    \ge \frac{\beta\alpha}{2}L_h,
\]
using $L_h=(1+\alpha)^h(1+j^\star\beta\alpha)\le 2(1+\alpha)^h$ for
$\alpha\le1$.  The outer selected buffer is no thinner.  Since Euclidean length
is at least radial progress, $\wts(T)\ge (\beta\alpha/2)L_h$.
\end{proof}

\begin{lemma}
\label{lem:inactive}
The total weight of all inactive shortcut edges is
\[
    O\left(\frac{\wts(\mst)}{\epsilon^{11/12}}\right)
    =o_\epsilon\left(\frac{\wts(\mst)}{\epsilon}\right).
\]
\end{lemma}

\begin{proof}
We bound the three inactive types separately.

\smallskip\noindent\textbf{Type (a).}
Consider a retained component $T\subseteq F_h$ with at most
$\epsilon^{-1/2}$ owned terminal anchors.  For each inactive radial path owned
by $T$, its terminal anchor lies in $A_h$.  By
\Cref{lem:radial-prefix-bound}, the whole radial path has length at most a
constant times $L_h$, for sufficiently small $\epsilon$.  Hence the total cost
charged to $T$ is at most $O(\epsilon^{-1/2}L_h)$.  By
\Cref{lem:tree-weight-lb}, this is
\[
    O\left(\frac{1}{\beta\alpha\epsilon^{1/2}}\right)\wts(T)
    =O(\epsilon^{-11/12})\wts(T),
\]
because $\alpha=\epsilon^{1/3}$ and $\beta=\Theta(\epsilon^{1/12})$.  Summing
over all retained components and using \eqref{eq:annulus-overlap} gives
$O(\epsilon^{-11/12}\wts(\mst))$.

\smallskip\noindent\textbf{Type (b).}
If $T\subseteq F_h$ owns more than $\epsilon^{-1/2}$ terminal anchors, only
the two extreme radial paths are type-(b) inactive.  Their total length is
$O(L_h)$, and by \Cref{lem:tree-weight-lb} this is
\[
    O\left(\frac{1}{\beta\alpha}\right)\wts(T)
    =O(\epsilon^{-5/12})\wts(T).
\]
After summing over all $T$, this is dominated by the type-(a) bound.

\smallskip\noindent\textbf{Type (c).}
Let $T\in F_h$, and let $v\in\net(T)\cap C$ be a net point for which the
algorithm adds an edge from $v$ to the selected radial anchor $a_{i_r}$ of cone
$C$.  Then $h\in\{i_r,i_r+1\}$.  Thus $v\in A_h$ and
$a_{i_r}\in A_{i_r}$, so the two endpoints lie in the same cone and in the same
or two neighboring annular regions.  Their radial coordinates differ by at most
$O(\alpha L_h)$, and the angular contribution is at most
$\theta R_h=O(\epsilon L_h)$.  Since $\epsilon\le\alpha$, every such edge has
length $O(\alpha L_h)$.

The net-size bound and \Cref{lem:tree-weight-lb} give
\[
    |\net(T)|
    \le
    1+\frac{2\wts(T)}{\gamma\epsilon L_h}
    =O\left(\frac{\wts(T)}{\gamma\epsilon L_h}\right),
\]
where the additive $1$ is absorbed because
$\wts(T)/(\gamma\epsilon L_h)\ge \beta\alpha/(2\gamma\epsilon)\to\infty$.
Each net point belongs to exactly one cone, so this bound also bounds the number
of type-(c) edges caused by this occurrence of $T$.  The type-(c) cost caused by
this component is at most
\[
    O(\alpha L_h)\cdot
    O\left(\frac{\wts(T)}{\gamma\epsilon L_h}\right)
    =O\left(\frac{\alpha}{\gamma\epsilon}\right)\wts(T)
    =O(\epsilon^{-11/12})\wts(T),
\]
using $\gamma=\epsilon^{1/4}$.  Summation over all components in $F_h$, and then
applying \Cref{lem:sum-components-revised} gives
$O(\epsilon^{-11/12}\wts(\mst))$.
\end{proof}

\begin{proof}[Proof of Theorem~\ref{thm:plane-slt}]
The stretch guarantee is exactly \Cref{lem:stretch}.  The MST edges contribute
$\wts(\mst)$.  By \Cref{lem:active,lem:inactive}, all shortcut edges together
have weight at most
\[
    \left(c_\star+o_{\epsilon}(1)\right)
    \frac{1}{\epsilon}\wts(\mst)
    +o_\epsilon\left(\frac{\wts(\mst)}{\epsilon}\right).
\]
Thus
\[
    \wts(G)
    \le
    \wts(\mst)+
    \left(c_\star+o_{\epsilon}(1)\right)
    \frac{1}{\epsilon}\wts(\mst)
    +o_\epsilon\left(\frac{\wts(\mst)}{\epsilon}\right) \le
    \left(c_\star+o_{\epsilon}(1)\right)
    \frac{1}{\epsilon}\wts(\mst).
\]
A shortest-path tree of $G$ rooted at $s$ has weight at most $\wts(G)$ and
preserves all root-distance upper bounds proved in \Cref{lem:stretch}.  This
completes the construction of the desired $(1+\epsilon)$-shallow tree.
\end{proof}
\section{Lower Bounds in Higher Dimensions}

\subsection{Three Dimensional Euclidean Space}
\LowerBoundThree*

Our construction will resemble a section of a square lattice mapped to the surface of the unit sphere, where the geodesic distance of each lattice edge is $(1+o_{\eps}(1))\eps$. The point $s$ is the center of the sphere, and $P$ contains the vertices of the lattice. The number of points is $\Theta(1/\eps^2)$. The weight of the minimum spanning tree is then $\Theta(1/\eps)$. Thus, the radial edge contributes only an $\eps$ fraction of $\wts(\mst)$. The distance between any two lattice points is more than $\eps$, forcing a direct connection from $s$ to every lattice point. We then obtain that the total weight of any $(1+\eps,\beta)$-SLT is $1/\eps^2$. However, there is no square lattice on a sphere. We know that within a small distortion, a square lattice on the sphere is similar to a square lattice in the plane. Thus, we only construct a lattice within distance $\eps^{1/4}$ from an equator. We note that any exponent less than $1/2$ should work.

Formally, consider a 3-dimensional Euclidean space $\mathbb{R}^3$, where the coordinate axes are the $x$-, $y$- and $z$-axes. 
Let $s$ be the origin. We first define a section of a square lattice in the $xy$-plane, and then map the points to the unit sphere. Specifically, we set the distance between any two consecutive lattice points to be $b = \eps + \eps^2$ in $xy$-plane. Let $N = \lfloor\frac{\eps^{1/4}}{b}\rfloor$ ($N = \Theta(\eps^{-3/4})$). For $0 \leq i, j \leq N$, we define
\begin{equation*}
    p_{i, j} = \left(ib, jb, \sqrt{1 - (ib)^2 - (jb)^2}\right).
\end{equation*}

The point $p_{i, j}$ is well-defined, since for $i, j \leq N$, $(ib)^2 + (jb)^2 \leq \eps^{1/4} < 1$. Set $P = \{s\} \cup \{p_{i, j}\}_{0 \leq i, j \leq N}$. It is immediate that $p_{i, j}$ lies on the unit sphere $\mathbb{S}^2$.

\begin{observation}
    \label{obs:no-grid-replace}
    For any two distinct lattice points, $(i, j)$ and $(i', j')$, we have $\wts(p_{i, j}, p_{i', j'}) \geq b$.
\end{observation}
\begin{proof}
    We have
    \begin{equation*}
        \wts(p_{i, j}, p_{i', j'}) \geq \sqrt{(i - i')^2b^2 + (j - j')^2b^2} \geq b ,
    \end{equation*}
    since at least one of $(i - i')^2$ and $(j - j')^2$ is at least one. 
\end{proof}

We now show that the distance between any two consecutive lattice points is more than $\eps$.

\begin{claim}
    For any two pairs $(i, j)$ and $(i', j')$ such that $(i', j') \in \{(i + 1, j), (i, j + 1)\}$, 
    \begin{equation*}
        \wts(p_{i, j}, p_{i', j'}) = \eps + O(\eps^{3/2}) .
    \end{equation*}
\end{claim}
\begin{proof}
    We may assume w.l.o.g. that $(i', j') = (i + 1, j)$; the remaining case can be shown similarly. First, observe that
    \begin{equation*}
    \begin{split}
        \sqrt{1 - (i + 1)^2b^2 - (jb)^2} - \sqrt{1 - (ib)^2 - (jb)^2}  &= \frac{2ib^2 + b^2}{\sqrt{1 - (i + 1)^2b^2 - (jb)^2} + \sqrt{1 - (ib)^2 - (jb)^2}}\\
        &= \frac{2ib^2 + b^2}{2 - o(\sqrt{\eps})} .
    \end{split}
    \end{equation*}
    Then, the squared distance between $p_{i, j}$ and $p_{i + 1, j}$ is at most
    \begin{equation*}
        \wts(p_{i, j}, p_{i + 1, j})^2 \leq b^2 + \left(\frac{2ib^2 + b^2}{2 - o(\sqrt{\eps})}\right)^2 = b^2 + O(\eps^{5/2}) 
    \end{equation*}
    since $i \leq \eps^{-3/4}$. Thus,
    \begin{equation*}
        \wts(p_{i, j}, p_{i + 1, j}) \leq \sqrt{b^2 + O(\eps^{5/2})} = \eps + O(\eps^{3/2}) ,
    \end{equation*}
    since $b = \eps + \eps^2$. 
\end{proof}

We can now prove \Cref{thm:LB-3D}.

\begin{proof}[Proof of \Cref{thm:LB-3D}]
    Consider the following spanning graph $G$ with vertex set $P$. The graph $G$ contains an edge from $s$ to $p_{0, 0}$. For each $j$ from $0$ to $N - 1$, add the edge $(p_{0, j}, p_{0, j + 1})$ to $G$. Then, for each $j$ from $0$ to $N$ and $i$ from $0$ to $N - 1$, add the edge $(p_{i, j}, p_{i + 1, j})$ to $G$. It is immediate that $G$ is connected, hence $\wts(\mst(P)) \leq \wts(G)$. We have
    \begin{equation*}
        \begin{split}
            \wts(G) &= \wts(s, p_{0, 0}) + \sum_{0 \leq j \leq N - 1} \wts(p_{0, j}, p_{0, j + 1}) + \sum_{0 \leq i \leq N - 1}\sum_{0 \leq j \leq N} \wts(p_{i, j}, p_{i + 1, j})\\
            &\leq 1 + N(\eps + O(\eps^{3/2})) + (N + 1) \cdot N \cdot (\eps + O(\eps^{3/2})) = 1 + N(N + 2)(\eps + O(\eps^{3/2}))\\
        \end{split}
    \end{equation*}

    Let $T$ be a $(1 + \eps, \beta)$-SLT of $P$ for some $\beta>0$. By \Cref{obs:no-grid-replace}, $T$ must contain the edge from $s$ to $p_{i, j}$ for all $0 \leq i, j \leq N$. Thus, $\wts(T) \geq (N + 1)^2$. Therefore,
    \begin{equation*}
        \beta = \frac{\wts(T)}{\wts(\mst(P))} \geq \frac{(N + 1)^2}{N(N + 2)(\eps + O(\eps^{3/2})) + 1} .
    \end{equation*}
    Since $N = \Theta(\eps^{-3/4})$, we have $\beta \geq \frac{(N + 1)^2}{N(N+2)\eps + O(1)}$. Assume that $\beta \geq \frac{(N + 1)^2}{N(N+2)\eps + c}$ for some constant $c > 1$.  
    Consider $1/\eps - \beta$. We have
    \begin{equation*}
        \begin{split}
            \frac{1}{\eps} - \beta &\leq \frac{1}{\eps} \cdot\frac{N(N+ 2)\eps + c - \eps(N + 1)^2}{N(N+2)\eps + c}\\
            &= \frac{1}{\eps} \cdot \frac{c - \eps}{N(N+2)\eps + c} = \frac{1}{\eps} \cdot \Theta(\sqrt{\eps}) .
        \end{split}
    \end{equation*}
    The last equation holds because $N = \Theta(\eps^{-3/4})$. Therefore, $\beta \geq \left(1 - \Theta(\sqrt{\eps})\right) \cdot \frac{1}{\eps}$. 
\end{proof}

\subsection{Four Dimensional Euclidean Space}

\subsubsection{Four Dimensional Euclidean Space: Weaker Warm-up Lower Bound}
\label{ssec:four-dim-lower-bound}

In this section, we show that there exist point sets in $\mathbb{R}^4$ for which any $(1 + \eps, \beta)$-SLT must have lightness $\beta$ strictly larger than $1/\eps$.

\begin{theorem}
    \label{thm:easy-4d-bound}
    For every sufficiently small $\eps > 0$, there exists a finite point set $P \subset \mathbb{R}^4$ and a root $s \in P$ such that every $(1 + \eps, \beta)$-SLT rooted at $s$ has lightness $\beta > \frac{1.0113}{\eps}$.
\end{theorem}

We briefly outline the idea of our construction. Recall that the lower bound for SLTs in Euclidean plane is given by points on a unit circle and a root $s$ at the center of the circle. Every $(1+\eps,\beta)$-SLT is forced to connect $s$ to an $\eps$-net on the circle. The idea of our construction is to exploit the higher dimension to save the distance between two consecutive net points; that is, to force $k$ connections to a set of points whose contribution to the MST is strictly less than $k\eps$, and hence save enough to compensate for the radial edge in the MST.

We denote the coordinate axes of $\mathbb{R}^4$ 
by $x$-, $y$-, $z$- and $w$-axis, respectively;
and denote the standard basis vectors by $e_1,\ldots ,e_4$.
The point set consists of a root $s$ and a set of points constructed as follows. Consider a circle $C$ centered at $s$, and a $(c \eps)$-net on $C$, where $c < 1$. To each net point, we attach a tetrahedron: each tetrahedron has one vertex at the net point and the remaining three vertices form a base triangle. The goal of adding tetrahedra is to force connections from $s$ to all vertices in the bases, thereby saving $(1 - c)\eps$ along the circle net. Formally, let $s$ be the origin in $\mathbb{R}^4$, and let $C$ be the circle of radius $R = 1 + a\eps$ centered at $s$ in the $xy$-plane, for some constant $a>0$ specified later. Here, $a\eps$ is the height of a tetrahedron. Let $p_1, p_2, \ldots ,p_N$ be a $(c\eps)$-net of $C$, where $c\in (0,1)$ is a constant specified later. Hence, $N = \lfloor\frac{2\pi R}{c\eps}\rfloor$.
We assume that $N$ is even; otherwise, we drop the last point $p_N$.

For all $i\in\{1,\ldots, N\}$, let $u_i$ be a point on the segment $sp_i$ such that $\wts(s, u_i) = 1$. At each point $p_i$, we attach a tetrahedron whose base triangle lies in an affine 2-flat parallel to the $zw$-plane and contains the point $u_i$. Each base triangle is an equilateral triangle of side length strictly larger than $\eps$, but sufficiently close to $\eps$. Specifically, we choose the side length to be $\eps + \eps^3$. For any two consecutive base triangles, attached to $p_i$ and $p_{i + 1}$, we ``rotate'' the triangles by an angle large enough so that the connection from $s$ to a point in one base cannot be used to connect to a point in the other base while achieving a good approximate path. The rotation, together with the net distance, is sufficient to force connections to all base points. We set the rotation angle to be $\pi/3$. 
We define three base points, denoted $q_{i, j}$ for $j = 0, 1, 2$. If $i$ is even, then let
\begin{equation*}
    q_{i,j} = u_i + \frac{\eps + \eps^3}{\sqrt{3}}(\cos{(2\pi j /3)}e_3 + \sin{(2\pi j/3)}e_4).
\end{equation*}
If $i$ is odd, let
\begin{equation*}
    q_{i,j} = u_i + \frac{\eps + \eps^3}{\sqrt{3}}(\cos{(\pi/3 + 2\pi j /3)}e_3 + \sin{(\pi/3 + 2\pi j/3)}e_4).
\end{equation*}

Our point set is $P = \{s\} \cup \{p_i~:~i\in\{1,\ldots, N\}\} \cup \{q_{i, j}: i \in \{1, \ldots N\}, j \in \{0, 1, 2\}\}$. See \Cref{fig:warm-up} for an illustration.

\begin{figure}[h!]
	\centering
%

\begin{tikzpicture}[
    scale=0.92,
    every node/.style={font=\small},
    rootpoint/.style={circle, fill=red!70!black, inner sep=2pt},
    netpoint/.style={circle, fill=blue!70!black, inner sep=1.8pt},
    centerpoint/.style={circle, fill=black, inner sep=1.3pt},
    basepoint/.style={circle, fill=green!60!black, inner sep=1.5pt},
    >=Stealth
]

\begin{scope}
    \def\rIn{1.85}
    \def\rOut{2.35}
    \def\angA{26}
    \def\angB{58}
    \def\triR{0.22}

    \draw[->, gray] (-2.75,0) -- (2.75,0) node[right] {$x$};
    \draw[->, gray] (0,-2.55) -- (0,2.75) node[above] {$y$};

    \node[rootpoint,label=below left:{$s$}] (s) at (0,0) {};

    \draw[thick, blue!60!black] (0,0) circle (\rOut);

    \coordinate (pi) at ({\rOut*cos(\angA)},{\rOut*sin(\angA)});
    \coordinate (pj) at ({\rOut*cos(\angB)},{\rOut*sin(\angB)});
    \coordinate (ui) at ({\rIn*cos(\angA)},{\rIn*sin(\angA)});
    \coordinate (uj) at ({\rIn*cos(\angB)},{\rIn*sin(\angB)});

    \node[netpoint,label=above right:{$p_i$}] at (pi) {};
    \node[netpoint,label=above right:{$p_{i+1}$}] at (pj) {};

    \draw[gray] (s) -- (pi);
    \draw[gray] (s) -- (pj);
    \draw[thick, orange!80!black] (uj) -- (pj);

    \draw[thick, blue!60!black] (pi) arc[start angle=\angA, end angle=\angB, radius=\rOut];
    \draw[orange!80!black, ->]
        ({(\rOut+0.24)*cos(\angA)},{(\rOut+0.24)*sin(\angA)})
        arc[start angle=\angA, end angle=\angB, radius=\rOut+0.24];
    \node[orange!80!black] at ({(\rOut+0.50)*cos(42)},{(\rOut+0.50)*sin(42)}) {$c\eps$};

    \coordinate (qi0) at ($(ui)+(90:\triR)$);
    \coordinate (qi1) at ($(ui)+(210:\triR)$);
    \coordinate (qi2) at ($(ui)+(330:\triR)$);

    \coordinate (qj0) at ($(uj)+(150:\triR)$);
    \coordinate (qj1) at ($(uj)+(270:\triR)$);
    \coordinate (qj2) at ($(uj)+(30:\triR)$);

    \draw[green!50!black, thick] (qi0)--(qi1)--(qi2)--cycle;
    \draw[green!50!black, thick] (qj0)--(qj1)--(qj2)--cycle;

    \foreach \q in {qi0,qi1,qi2} {
        \node[basepoint] at (\q) {};
        \draw[dashed, green!50!black] (pi) -- (\q);
    }
    \foreach \q in {qj0,qj1,qj2} {
        \node[basepoint] at (\q) {};
        \draw[dashed, green!50!black] (pj) -- (\q);
    }

    \node[green!50!black] at ($(ui)+(0,-0.48)$){};
    \node[green!50!black] at ($(uj)+(0,0.52)$){};
\end{scope}

\end{tikzpicture} 
	\caption{The points set $P$ in the $xy$-plane. }
	\label{fig:warm-up}
\end{figure}
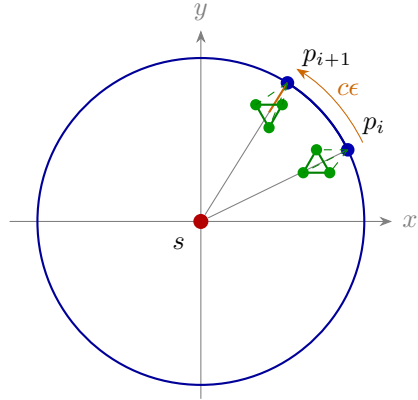

\begin{observation}
    \label{obs:net-base}
    For any $i \in \{1, \ldots, N\}$ and $j \in \{0, 1, 2\}$,  $\wts(p_i, q_{i, j}) = \eps\sqrt{a^2 + (1 + \eps^2)^2/3}$.
\end{observation}

\begin{proof}
    By the Pythagorean theorem,
    \begin{equation*}
        \sqrt{\wts(p_i, u_i)^2 + \wts(u_i, q_{i, j})^2} = \sqrt{(a\eps)^2 + (\eps + \eps^3)^2/3} = \eps\sqrt{a^2 + (1 + \eps^2)^2/3},
    \end{equation*}
    as desired.
\end{proof}

\begin{claim}
    \label{clm:mst-wt-1}
    $\wts(\mst(P)) \leq R + 2\pi R + 3N\eps\cdot \sqrt{a^2 + (1 + \eps^2)^2/3}$.
\end{claim}
\begin{proof}
    Consider a spanning graph $H$ of $P$ consisting of an edge from $s$ to an arbitrary net point, which has weight $R$; edges from each net point to the next net point, whose total weight is at most the perimeter of the circle, namely $2\pi R$; and edges from each net point to the three vertices in the base of the tetrahedron attached to it. By the Pythagorean theorem, the length of each edge from a net point $p_i$ to a base vertex $q_{i, j}$ is $\eps\sqrt{a^2 + (1 + \eps^2)^2/3}$ by \Cref{obs:net-base}.
    
    Hence, the weight of $H$ is at most
    \begin{equation*}
        R + 2\pi R + \sum_{1 \leq i \leq N}\sum_{0 \leq j \leq 2}\wts(p_i, q_{i, j}) = R + 2\pi R + 3N\eps\sqrt{a^2 + (1 + \eps^2)^2/3}.
    \end{equation*}
    Since the weight of the MST of $P$ is at most the weight of any spanning subgraph, the claim follows.
\end{proof}

We then show that no net point can be used in an approximate path from $s$ to any base point. Observe that:

\begin{observation}
    \label{obs:root-base}
    $\wts(s, q_{i, j}) = \sqrt{1 + \eps^2(1 + \eps^2)^2/3}$. 
\end{observation}

We have the following claim.

\begin{claim}
    \label{clm:net-base-forbidden}
    Assume that $a \geq 17/50$. For every point $p_i$ and every base point $q$,
    \begin{equation*}
        \wts(s, p_i) + \wts(p_i, q) > (1 + \eps)\wts(s, q). 
    \end{equation*}
\end{claim}

\begin{proof}
    Note that $\wts(s, p_i) = R$. We consider two cases: either $q$ belongs to the tetrahedron attached to $p_i$, or it does not. In the former case, we have
    \begin{equation*}
        \wts(s, p_i) + \wts(p_i, q)  = R + \eps\sqrt{a^2 + (1 + \eps^2)^2/3} = 1 + a\eps + \eps\sqrt{a^2 + (1 + \eps^2)^2/3} > 1 + \left(a + \sqrt{a^2 + \frac{1}{3}}\right)\eps. 
    \end{equation*}

    On the other hand,
    \begin{equation*}
        (1 + \eps)\wts(s, q) = (1 + \eps)\sqrt{1 + \frac{\eps^2(1 + \eps^2)^2}{3}} = 1 + \eps + O(\eps^2). 
    \end{equation*}
    Since $a \geq 17/50$, we have $a + \sqrt{a^2 + \frac{1}{3}} > 1.01$, and the claim follows for sufficiently small $\eps>0$.

    Now suppose that $q$ belongs to a tetrahedron attached to a net point other than $p_i$, say $p_{i'}$. It is enough to show that $\wts(p_i, q) > \wts(p_{i'}, q)$, since $\wts(s,p_i)=\wts(s,p_{i'})=R$ and the previous case then applies to $p_{i'}$. By the Pythagorean theorem,
    \begin{equation*}
        \wts(p_i, q)^2 = \wts(p_i, u_{i'})^2 + \wts(u_{i'}, q)^2.
    \end{equation*}
    We claim that $\wts(p_i, u_{i'}) > \wts(p_{i'}, u_{i'}) = a\eps$. By the triangle inequality,
    \begin{equation*}
        \wts(p_i, u_{i'}) > \wts(s, p_i) - \wts(s, u_{i'}) = a\eps,
    \end{equation*}
    implying $\wts(p_i, q)^2 > \wts(p_{i'}, u_{i'})^2 + \wts(u_{i'}, q)^2 = \wts(p_{i'}, q)^2$, or $\wts(p_i, q) > \wts(p_{i'}, q)$.
\end{proof}

We then show that a base point cannot be used in the approximate path to another base point.

\begin{claim}
    \label{clm:base-base-forbidden}
    Assume that $c > 0.82$ and $a \geq 17/50$. For every two distinct base points $q_1$ and $q_2$,
    \begin{equation*}
        \wts(s, q_1) + \wts(q_1, q_2) > (1 + \eps)\wts(s, q_2).
    \end{equation*}
\end{claim}

\begin{proof}
    By \Cref{obs:root-base}, $\wts(s, q_1) = \wts(s, q_2) = \sqrt{1 + \eps^2(1 + \eps^2)^2/3}$. Thus, we only need to show that $\wts(q_1, q_2) > \eps \cdot \wts(s, q_2) = \eps\sqrt{1 + \eps^2(1 + \eps^2)^2/3} < \eps + \eps^3/3 + O(\eps^5)$.
    
    We consider three cases. Case $1$ is when $q_1$ and $q_2$ are in the base of the same tetrahedron. Then $\wts(q_1, q_2) = \eps + \eps^3 > \eps\sqrt{1 + \eps^2(1 + \eps^2)^2/3}$ for sufficiently small $\eps$.

    Case $2$ is when $q_1$ and $q_2$ are in two consecutive tetrahedra attached to $p_i$ and $p_{i + 1}$. Assume that $q_1 = u_i + \frac{\eps + \eps^3}{\sqrt{3}}(\cos{\theta}e_3 + \sin{\theta}e_4)$ and $q_2 = u_{i + 1} + \frac{\eps + \eps^3}{\sqrt{3}}(\cos{\theta'}e_3 + \sin{\theta'}e_4)$. By construction, $\theta-\theta'$ is an odd multiple of $\pi/3$. By the Pythagorean theorem,
    \begin{equation*}
    \begin{split}
        \wts(q_1, q_2)^2 &= \wts(u_i, u_{i + 1})^2 + \frac{(\eps + \eps^3)^2}{3}\left((\sin{\theta} - \sin{\theta'})^2 + (\cos{\theta} - \cos{\theta'})^2\right) \\
        &= \frac{\wts(p_i, p_{i + 1})^2}{R^2} + \frac{(\eps + \eps^3)^2}{3}(2 - 2\cos{\theta - \theta'}).
    \end{split}
    \end{equation*}
    Since $\theta-\theta'$ is an odd multiple of $\pi/3$, we have $\cos{\theta - \theta'} \leq 1/2$, and hence
    \begin{equation*}
        \wts(q_1, q_2)^2 \geq \frac{\wts(p_i, p_{i + 1})^2}{R^2} + \frac{(\eps + \eps^3)^2}{3} = c^2\eps^2/(1 + a\eps)^2 + \frac{(\eps + \eps^3)^2}{3}.    
    \end{equation*}
    Hence, $\wts(q_1, q_2) \geq \eps\sqrt{\frac{c^2}{(1 + a\eps)^2} + \frac{(1 + \eps^2)^2}{3}}$. Choosing $c$ slightly larger than $\sqrt{2/3}$, say $c \geq 0.82$, and taking $a \geq 17/50$, we obtain $\wts(q_1, q_2) > \eps \cdot \wts(s, q_2)$ for sufficiently small $\eps$.

    Case $3$ is when $q_1$ and $q_2$ are in two non-consecutive tetrahedra attached to $p_i$ and $p_{i'}$. In this case, their distance is at least $\wts(u_i, u_{i'})$. By the Pythagorean theorem, we have
    \begin{equation}
        \wts(q_1, q_2) \geq \wts(u_i, u_{i'}) = \frac{\wts(p_i, p_{i'})}{R}. 
    \end{equation}
    Since $p_i$ and $p_{i'}$ are non-consecutive, $\wts(p_i, p_{i'}) \geq \wts(p_i, p_{i + 2}) = 2c\eps - O(\eps^3) > 1.1\eps$ for sufficiently small $\eps>0$, provided $c > 0.6$. Therefore $\wts(q_1,q_2) > \eps \cdot \wts(s,q_2)$ for sufficiently small $\eps>0$.
\end{proof}

\begin{lemma}
    \label{lm:root-base-connection}
    Let $T$ be any $(1 + \eps, \beta)$-SLT of $P$ rooted at $s$. Then $T$ contains the edge from $s$ to $q_{i, j}$ for every $i \in \{1, \ldots, N\}$ and every $j \in \{0, 1, 2\}$.
\end{lemma}

\begin{proof}
    Fix a base point $q$, and let $S$ be the path in $T$ from $s$ to $q$. By \Cref{clm:net-base-forbidden}, the path $S$ cannot contain any net point $p_{i'}$. Indeed, if $S$ contained such a point, then the length of $S$ would be at least $\wts(s,p_{i'})+\wts(p_{i'},q)$, which is larger than $(1+\eps)\wts(s,q)$.

    Similarly, by \Cref{clm:base-base-forbidden}, the path $S$ cannot contain any base point $q'\neq q$. Therefore, since the only points in $P$ are the root, the net points, and the base points, the path $S$ contains no intermediate vertex. Hence $S$ is the single edge from $s$ to $q$.
\end{proof}

\begin{proof}[Proof of \Cref{thm:easy-4d-bound}]
    Consider the point set $P$ constructed above. By \Cref{lm:root-base-connection}, the total weight of $T$ is at least
    \begin{equation*}
        3N\wts(s, q_{1, 0}) = 3N\sqrt{1 + \eps^2(1 + \eps^2)^2/3} = 3N + O(\eps).
    \end{equation*}
    By \Cref{clm:mst-wt-1},
    \begin{equation*}
        \wts(\mst(P)) \leq R + 2\pi R + 3N\eps\sqrt{a^2 + (1 + \eps^2)^2/3}.
    \end{equation*}
    Hence,
    \begin{equation*}
        \begin{split}
            \beta &= \frac{\wts(T)}{\wts(\mst(P))} \geq \frac{3N + O(\eps)}{R + 2\pi R + 3N\eps\sqrt{a^2 + (1 + \eps^2)^2/3}}\\
            &= \frac{6\pi/(c\eps)}{1 + 2\pi + 6\pi/c \cdot \sqrt{a^2 + (1 + \eps^2)^2/3}} - o(1/\eps) \qquad \text{since $N = \frac{2\pi R}{c\eps} - O(1)$}\\
            &= \frac{1}{\eps} \frac{6\pi}{(1 + 2\pi)c + 6\pi\sqrt{a^2 + 1/3}} - o(1/\eps).
        \end{split}
    \end{equation*}
    Choosing $c = 33/40$ and $a = 17/50$, we obtain $\beta > \frac{1.0113345}{\eps}$ for sufficiently small $\eps>0$.
\end{proof}

\subsubsection{Four Dimensional Euclidean Space: A Stronger Lower Bound}
\label{ssec:four-dim-strong-bound}
In this section, we give a stronger lower bound than 
\Cref{thm:easy-4d-bound}
for the lightness of any $(1+\eps,\beta)$-SLT in $\mathbb{R}^4$, showing that the lightness $\beta$ must be at least nearly $12/11$.

\LowerBoundFour*

We now explain the idea of the construction. Instead of using a circle, we pack regular tetrahedra on the unit sphere $\mathbb{S}^3$. We construct a curve $C\subset \mathbb{S}^3$ of arclength polynomial in $1/\epsilon$. Note that this can only be done on a higher-dimensional sphere (i.e., $\mathbb{S}^d$ for $d\geq 2$), since any circle has constant length. We then place a $c\eps$-net, with $c = 1 + o(\eps)$, on $C$ and denote these points by $p_1, p_2, \ldots, p_N$. Each tetrahedron contains two consecutive net points $p_i$ and $p_{i + 1}$, together with two additional points $u_i$ and $v_i$, chosen so that $\{p_i, p_{i + 1}, u_i, v_i\}$ is a regular tetrahedron in $\mathbb{S}^3$.

The purpose of the curve $C$ is to replicate the role of the circle, while ensuring that the Euclidean distance between nearby points is approximately their distance along the curve. A useful analogy is a spherical helix in three dimensions. There are many curves with this behavior; for ease of computation, we use a curve on the Clifford torus. Formally, consider the curve
\begin{equation*}
    \Gamma(t) = \frac{1}{\sqrt{2}}\left(\cos{t}, \sin{t}, \cos{\frac{t}{\sqrt{\eps}}}, \sin{\frac{t}{\sqrt{\eps}}}\right), \qquad t\in [0, \pi].
\end{equation*}

The arclength of the curve is $\int_{0}^\pi\|\Gamma'(t)\|dt = \frac{\pi}{\sqrt{2}}\sqrt{1 + \frac{1}{\eps}} = \Theta(\eps^{-1/2})$. Placing points $\{p_1, p_2, \ldots p_N\}$ on $C$ such that two consecutive points have distance $b$, 
with $b = \eps(1 + \eps^{1/3})$, we obtain $N = \Theta(\eps^{-3/2})$. We now construct the tetrahedron containing $p_i$ and $p_{i + 1}$. Consider the intersection of the sphere centered at $p_i$ of radius $b$, the sphere centered at $p_{i + 1}$ of radius $b$, and the unit sphere $\mathbb{S}^3$. This intersection is a circle of radius $b\sqrt{\frac{3 - b^2}{4 - b^2}}$ (see \Cref{clm:intersection-circle}), which is greater than $b/2$ for sufficiently small $\eps$. Hence, we can choose two points on this circle whose mutual distance is exactly $b$. However, we need to choose these two points for each tetrahedron in a pattern so that the distance between $\{u_i,v_i\}$ and $\{u_{i+1},v_{i+1}\}$, in two consecutive tetrahedra,  
is also strictly greater than $\eps$. Meaning that $(u_i, v_i)$ and $(u_{i + 1}, v_{i + 1})$ are ``almost'' parallel.

The choices of the two points, $u_i$ and $v_i$, are made consistently along the curve. 
For an analogy, imagine a different scenario when $p_i$, $p_{i + 1}$ and $p_{i + 2}$ are on a straight line at distance $b$ apart in $\mathbb{R}^3$ equipped with the standard basis $\{e_1, e_2, e_3\}$ such that $p_i$ is the origin, $p_{i + 1} = be_1$ and $p_{i + 2} = 2be_1$. In this case, one can choose $u_i$ and $v_i$ arbitrarily so that $\{p_i,p_{i+1},u_i, v_i\}$ form a regular tetrahedron, and then translate them by vector $b\overrightarrow{e_1}$ to get $(u_{i + 1}, v_{i + 1})$. As $u_i$ and $v_i$ lie in the orthogonal bisector of $p_ip_{i+1}$, then $\{u_i,v_i\}$ and $\{u_{i+1},v_{i+1}\}$ lie in two parallel planes at distance $b$ apart. 
However, if $p_i,p_{i+1}$ are on a spherical curve, we need to shift $(u_i, v_i)$ along the curve $\Gamma\subset \mathbb{S}^3$. Since the triple $p_i,p_{i+1},p_{i+2}\in \Gamma$ deviates from a straight line by only $O(b^2) = O(\eps^2)$, we will still get our separation property (see \Cref{clm:base-base-forbidden2} below). Formally, consider the a tetrahedron $\Delta_i$ containing $p_i, p_{i + 1}$. Let $f_{i, 1}$ and $f_{i, 2}$ be unit vectors in directions $p_{i + 1} - p_i$ and $\frac{p_i + p_{i + 1}}{2}$. We find an orthonormal basis $\{f_{i, 3}, f_{i, 4}\}$ of the plane containing all possible point pairs $\{u_i, v_i\}$ as follows. (We prove in \Cref{clm:intersection-circle} below that the locus of potential points for $u_i$ and $v_i$ is a two-dimensional circle.) We first choose $f_{i, 4}$ such that $f_{i, 4}$ is almost orthogonal to the change of $p_i$, i.e., the dot product between $f_{i, 4}$ with $f_{i + 2, 2} - f_{i, 2}$ is insignificant $O(\eps^2)$. Then, we find two points $u_i, v_i$ such that $u_i = t_2 \cdot f_{i, 2} + t_3 \cdot f_{i, 3} + t_4 \cdot f_{i, 4}$ and $v_i = t_2 \cdot f_{i, 2} + t_3 \cdot f_{i, 3} - t_4 \cdot f_{i, 4}$, where the coefficients $\{t_j\}_{j \in \{2, 3, 4\}}$ are constant and chosen consistently for all $i=1,\ldots ,N-1$. We note that there are different methods to compute $u_i$ and $v_i$. Our goal is to choose $u_i$ and $v_i$ so that the calculations in our proof are as simple as possible. We show that suitable coefficients $t_2, t_3$ and $t_4$ exist in \Cref{clm:constant-exist}.

The circumcenter of the regular tetrahedron $\Delta_i = \{p_i, p_{i + 1}, u_i, v_i\}$ is $o_i = \frac{(p_i + p_{i + 1} + u_i + v_i)}{4}$. Let $x_i = o_i + h \cdot \frac{o_i}{\|o_i\|}$ for some constant $h\in (0,1)$ to be specified later. 
Our point set is 
\[P = \{s\} \cup \{p_i\}_{i = 1}^N \cup \{u_i, v_i, x_i\}_{i = 1}^{N - 1}.\] 
We call the vertices of $\Delta_i$, for $i=1,\ldots , N-1$,
\emph{tetrahedon points}.
Note that $\wts(o_i,x_i)=h$ and the distance from $x_i$ to all four vertices of $\Delta_i$ is the same, let us denote it by $L$. We shall choose $h$ so that $L<\eps$ and $x_i$ cannot be used in an approximate path from $s$ to any vertex of $\Delta_i$.  Since the circumradius of $\Delta_i$ is $\frac{\sqrt{6}}{4}b$, using the Pythagorean theorem for the right triangle $o_i p_i x_i$, we have
\begin{equation}
    \label{eq:con1}
    h^2 + \frac{3b^2}{8} = L^2.
\end{equation}

We first prove the following claim.
\begin{claim}
    \label{clm:constant-exist} There exist real numbers $t_2$, $t_3$, and $t_4$ such that for every $i=1,\ldots , N-1$, the set $\{u_i, v_i, p_i, p_{i + 1}\}$ forms a regular tetrahedron $\Delta_i$ in $\mathbb{S}^3$. Moreover, $t_2 = 1 - o(\eps^2)$ and $t_3, t_4 = O(\eps)$.
\end{claim}

Before proving \Cref{clm:constant-exist}, we show that the intersection of the sphere centered at $p_i$ and $p_{i + 1}$ of radius $b$ and the unit sphere $\mathbb{S}^3$ is a circle of radius more than $b/2$. For any point $x$ and any radius $r > 0$, let $S(x, r)$ be the sphere of center $x$ with radius $r$. For any two vectors $v_1$ and $v_2$, let $\langle v_1, v_2\rangle$ be the dot product between them. 

\begin{claim}
    \label{clm:intersection-circle} The intersection of $S(p_i, b)$, $S(p_{i + 1}, b)$ and $\mathbb{S}^3$ is a circle centered at $\frac{1-b^2/2}{1-b^2/4} \cdot \frac{p_i + p_{i + 1}}{2}$ of radius $\sqrt{\frac{3-b^2}{4-b^2}}\cdot b$.
\end{claim}

\begin{proof}
    Let $D = S(p_i, b) \cap S(p_{i + 1}, b) \cap\mathbb{S}^3$. For any point $x \in D$, we have $\langle x, p_i\rangle = \langle x, p_{i + 1}\rangle = 1 - \frac{b^2}{2}$.
    Thus $D$ is the intersection of $\mathbb{S}^3$ with two affine hyperplanes. Since the intersection of these two hyperplanes is an affine 2-flat, then $D$ is a circle. Let $m=(p_i+p_{i + 1})/2$. By symmetry, the center of $D$ lies on the line spanned by $m$, so it has the form $tm$. 
    Moreover, since the vector $tm$ is orthogonal to the plane containing $D$, we have
    \begin{equation}
        \langle tm, m\rangle = \langle x, m\rangle = \frac{1}{2} \cdot \left( \langle x, p_i\rangle + \langle x, p_{i + 1}\rangle \right) = 1 - \frac{b^2}{2},
    \end{equation}
    implying that $t = \frac{1 - b^2/2}{\|m\|^2}$. The Pythagorean theorem in the triangle $\Delta(s,m,p_i)$ yields $\|m\|^2 = 1 - \frac{b^2}{4}$. Thus, $t = \frac{1 - b^2/2}{1-b^2/4}$ and the center of $D$ is $\frac{1 - b^2/2}{1-b^2/4} \cdot \frac{p_i + p_{i + 1}}{2}$.
    Finally, the radius of $D$ is also obtained from the Pythagorean theorem: Since $x-tm$ is orthogonal to $m$, then
    \begin{equation*}
    \|x-tm\|^2=1-\|tm\|^2 = 1-\frac{(1-b^2/2)^2}{1-b^2/4} = \frac{b^2(3-b^2)}{4-b^2}.
    \end{equation*}
    Thus the radius of $D$ is $\sqrt{\frac{3-b^2}{4-b^2}}\cdot b$.
\end{proof}

\begin{proof}[Proof of \Cref{clm:constant-exist}]
    We find suitable $t_2$, $t_3$ and $t_4$ for the tetrahedron $\Delta_1$. The claim follows by applying the same parameters to all $i=1,\ldots , N-1$. By \Cref{clm:intersection-circle}, the intersection of $S(p_1, b)$, $S(p_2, b)$ and $\mathbb{S}^3$ is a circle centering at $\frac{1-b^2/2}{1-b^2/4} \cdot \frac{p_1 + p_2}{2}$ of radius $r := \sqrt{\frac{3-b^2}{4-b^2}} \cdot b > b/2$. Since $\sqrt{\frac{3-b^2}{4-b^2}} > \frac12$, we can choose $u_1$ and $v_1$ on that circle $D$ such that $\wts(u_1, v_1) = b$. This ensures that $\{p_1, p_2, u_1, v_1\}$ forms a regular tetrahedron. Since the circle $D$ is in a plane orthogonal to both $f_{1, 1}$ and $f_{1, 2}$, we can choose
    \begin{equation*}
        \begin{split}
            &u_1 = \frac{1-b^2/2}{1-b^2/4} \cdot \frac{p_1 + p_2}{2} + \sqrt{r^2 - \frac{b^2}{4}}f_{1, 3} + \frac{b}{2} f_{1, 4} ,\\
            &v_1 = \frac{1-b^2/2}{1-b^2/4} \cdot \frac{p_1 + p_2}{2} + \sqrt{r^2 - \frac{b^2}{4}}f_{1, 3} - \frac{b}{2} f_{1, 4} .
        \end{split}
    \end{equation*}
    Therefore, we can take $t_2 = \frac{1-b^2/2}{1-b^2/4} \cdot \frac{\|p_1 + p_2\|}{2}$, $t_3 = \sqrt{r^2 - \frac{b^2}{4}}$, and $t_4 = \frac{b}{2}$. 
\end{proof}

We can now choose $f_{i, 4}$ for every $i$ such that $f_{i, 4}$ is almost orthogonal to the change of $f_{i, 2}$. 
\begin{claim}
    \label{clm:almost-ortho} For every $i = 1,\ldots , N - 1$, there is a choice of $f_{i, 4}$, such that for every $i' \in \{i + 1, i + 2\}$, we have $|\langle f_{i, 4}, f_{i, 2} - f_{i', 2} \rangle| = O(\eps^2)$.
\end{claim}

\begin{proof}
    The intuition is that $f_{i, 2} - f_{i', 2}$ is close to a tangent of the curve $\Gamma$, then, we find $f_{i, 4}$ within the radius of that tangent. 
    Observe that each point in the Clifford-torus can be decomposed into a sum of two vectors: one corresponding to the slower cycle, containing the first two coordinates, and one corresponding to the faster cycle, containing the last two coordinates. 
    Let $t_i \in [0, \pi]$ such that $p_i = \Gamma(t_i)$ for every $i$. Let $\delta = t_{i + 1} - t_i = \eps^{3/2} + o(\eps^2)$. Let $m_i = \frac{t_i + t_{i + 1}}{2}$. Consider two vectors corresponding to the slow cycle and the fast cycle of the Clifford-torus curve:
    \begin{equation*}
        \begin{split}
            R_i &= (\cos m_i, \sin m_i, 0, 0),\\
            Q_i &= \left(0, 0, \cos \frac{m_i}{\sqrt{\eps}}, \sin \frac{m_i}{\sqrt{\eps}}\right).
        \end{split}
    \end{equation*}
    We choose $f_{i, 4}$ to be orthogonal to $p_i + p_{i + 1}$ and $p_{i + 1} - p_i$, and still within the span of $R_i$ and $Q_i$. Observe that $p_{i + 1} - p_i$ is orthogonal to both $R_i$ and $Q_i$. On the other hand,
    \begin{equation*}
        p_i + p_{i + 1} = \sqrt{2}\left(\cos{\delta} \cdot R_i + \cos{\frac{\delta}{\sqrt{\eps}}} \cdot Q_i\right).
    \end{equation*}

    Thus, 
    $$f_{i, 2} = \frac{\cos\delta \cdot R_i + \cos\frac{\delta}{\sqrt{\eps}} \cdot Q_i}{\sqrt{\cos^2{\delta} + \cos^2\frac{\delta}{\sqrt{\eps}}}}.$$

    Then, since $f_{i, 4}$ is in $\mathsf{span}(R_i, Q_i)$ and orthogonal to $f_{i, 2}$, we have
    \begin{equation*}
        f_{i, 4} = \frac{\cos\frac{\delta}{\sqrt{\eps}} \cdot R_i -  \cos\delta \cdot Q_i}{\sqrt{\cos^2{\delta} + \cos^2\frac{\delta}{\sqrt{\eps}}}}.
    \end{equation*}
    We show that our choice of $f_{i, 4}$ satisfies \Cref{clm:almost-ortho}. Since $\langle f_{i, 4}, f_{i, 2} \rangle = 0$, we only need to show that $|\langle f_{i, 4}, f_{i', 2} \rangle| = O(\eps^2)$. Similar to $f_{i, 2}$, we have
    \begin{equation*}
        f_{i', 2} = \frac{\cos\delta \cdot R_{i'} + \cos\frac{\delta}{\sqrt{\eps}} \cdot Q_{i'}}{\sqrt{\cos^2{\delta} + \cos^2\frac{\delta}{\sqrt{\eps}}}}.
    \end{equation*}
    Observe that $\langle R_i, R_{i'} \rangle = \cos((i' - i)\delta)$, $\langle Q_i, Q_{i'} \rangle = \cos((i' - i)\delta/\sqrt{\eps})$ and $\langle R_i, Q_{i'} \rangle = \langle R_{i'}, Q_{i} \rangle = 0$. Set $k = i' - i$. Expanding the product $\langle f_{i', 2}, f_{i, 4} \rangle$, we get
    \begin{equation*}
        |\langle f_{i', 2}, f_{i, 4} \rangle| = \frac{\cos \delta\cos \frac{\delta}{\sqrt{\eps}}}{\cos^2\delta + \cos^2\frac{\delta}{\sqrt{\eps}}} \left(\cos(k\delta) - \cos\frac{k\delta}{\sqrt{\eps}}\right) = O(\eps^2) ,
    \end{equation*}
 since $0 < k < 3$ and $\delta = \eps^{3/2} + o(\eps^2)$.  
\end{proof}

\begin{claim}
    \label{clm:base-base-forbidden2}
    For any two tetrahedon points, $q_1$ and $q_2$, $\wts(s, q_1) + \wts(q_1, q_2) > (1 + \eps)\wts(s, q_2)$. 
\end{claim}

\begin{proof}
    Since $\wts(s, q_1) = \wts(s, q_2) = 1$, it remains to show that $\wts(q_1, q_2) > \eps$. If $q_1$ and $q_2$ are in the same tetrahedron, then $\wts(q_1, q_2) = b = \eps(1 + \eps^{1/3}) > \eps$. Assume now that $q_1$ and $q_2$ are in different tetrahedra, corresponding to $p_i$ and $p_{i'}$. If $i' > i + 3$, then by the triangle inequality, $\wts(q_1, q_2) \geq 3b - \wts(p_i, q_1) - \wts(p_{i + 1}, q_2) = b > \eps$. We focus on the case $i' \leq i + 3$.
    
    There are two cases to consider: If $q_2$ is the image of $q_1$ of a shift from $p_i$ to $p_{i'}$, we may assume that $q_1 = u_i = t_2f_{i, 2} + t_3f_{i, 3} + t_4f_{i, 4}$ and $q_2 = t_2f_{i', 2} + t_3f_{i', 3} + t_4f_{i', 4}$. By the choice of $\{f_{i, j}\}_{j \in \{2, 3, 4\}}$ and $\{f_{i', j}\}_{j \in \{2, 3, 4\}}$, for any $j \in \{2, 3, 4\}$, the angle between $f_{i, j}$ and $f_{i', j}$ is only $O(\eps)$. Thus, $\langle f_{i, j}, f_{i', j'}\rangle$ is $1 - o(\eps^2)$ if $j = j'$ and $O(\eps)$ otherwise. We have
    \begin{equation*}
        \begin{split}
            \|q_1 - q_2\|^2 &= \|(t_2f_{i, 2} + t_3f_{i, 3} + t_4f_{i, 4}) - (t_2f_{i', 2} + t_3f_{i', 3} + t_4f_{i', 4})\|^2. 
        \end{split}
    \end{equation*}
    Let $d_j = f_{i, j} - f_{i', j}$ for $j \in \{2, 3, 4\}$. With this notation,   \begin{equation*}
        \begin{split}
            \|q_1 - q_2\|^2 &= \|t_2d_2 + t_3d_3 + t_4d_4\|^2 = \|t_2d_2\|^2 + \|t_3d_3\|^2 + \|t_4d_4\|^2 + 2\sum_{2\leq j < j' \leq 4}\langle t_jd_j, t_{j'}d_{j'}\rangle\\
            &\geq \|t_2d_2\|^2  + 2\sum_{2\leq j < j' \leq 4}t_jt_{j'}\langle d_j, d_{j'}\rangle \geq \|t_2d_2\|^2 - \sum_{2\leq j < j' \leq 4}t_jt_{j'}(\underbrace{d_j^2 + d_{j'}^2}_{= O(\eps^2)})\\
            &= \|t_2d_2\|^2 - O(\eps^3).
        \end{split}
    \end{equation*}
    Since $\|t_2d_2\|^2 = \Theta(\eps^2)$, we have $\wts(q_1, q_2) = \|t_2d_2\| - O(\eps^2) = t_2 \cdot \|f_{i, 2} - f_{i', 2}\| - O(\eps^2)$. Substituting the value of $t_2$ yields
    \begin{equation*}
        \wts(q_1, q_2) \geq \frac{1}{2}\|p_{i'} + p_{i' + 1} - p_{i} - p_{i + 1}\| - O(\eps^2) \geq b - O(\eps^2) > \eps, 
    \end{equation*}
    since $b = \eps(1 + \eps^{1/3})$ and $i' > i$. 

    Otherwise, we may assume w.o.l.g.\ that $q_1 = u_i = t_2f_{i, 2} + t_3f_{i, 3} + t_4f_{i, 4}$ and $q_2 = v_{i'} = t_2f_{i', 2} + t_3f_{i', 3} - t_4f_{i', 4}$. Recall that we choose $f_{i, 4}$ and $f_{i', 4}$ such that $|\langle f_{i, 4}, f_{i, 2} - f_{i', 2} \rangle| = O(\eps^2)$ in \Cref{clm:almost-ortho}. By symmetry, $|\langle f_{i', 4}, f_{i, 2} - f_{i', 2} \rangle| = O(\eps^2)$. Let $e_4 = f_{i, 4} + f_{i', 4}$. Then, we have
    \begin{equation*}
        \begin{split}
            \|q_1 - q_2\|^2 &= \|t_2d_2 + t_3d_3 + t_4e_4\|^2 \\
            &\geq \|t_2d_2\|^2  + \|t_4e_4\|^2 + \underbrace{2t_2t_3\langle d_2, d_3\rangle}_{\geq -O(\eps^3)} + \underbrace{2t_3t_4\langle d_3, e_4\rangle}_{\geq -O(\eps^3)} + \underbrace{2t_2t_4\langle d_2, e_4\rangle}_{\geq -O(\eps^3)}\\
            &\geq \|t_2d_2\|^2 + \|t_4e_4\|^2  - O(\eps^3) \geq b^2 - O(\eps^3), 
        \end{split}
    \end{equation*}
    implying that $\wts(q_1, q_2) > b - O(\eps^2) > \eps$.
\end{proof}
    
\begin{claim}
    \label{clm:x-base-forbid}
    If $h + L > \eps + \Theta(\eps^2)$, then for every tetrahedron point $q$ and every $x_i$, we have $\wts(s, x_i) + \wts(x_i, q) > (1 + \eps)\wts(s, q)$.
\end{claim}

\begin{proof}
    Since $\wts(s, q) = 1$, it suffices to show that $\wts(s, x_i) + \wts(x_i, q) > 1 + \eps$. We focus on the case $q \in \Delta_i$, as otherwise $\wts(x_i, q) \geq \wts(x_{i'}, q)$, with $i'$ being the index of the tetrahedron that $q$ belong to. Since $\Delta_i$ is a regular tetrahedron with side length $b$, its circumradius is $\frac{\sqrt{6}}{4}b$. Thus, we need to find $h$ and $L$ such that $\wts(s, x_i) + L > 1 + \eps$. Consider the circumcenter $o_i=\frac{(p_i + p_{i + 1} + u_i + v_i)}{4}$ of $\Delta_i$. Then
    \begin{equation*}
        \wts(s, x_i) = h + \wts(s, o_i) = h + \sqrt{1 - 3b^2/8} = h + 1 - \Theta(\eps^2) ,
    \end{equation*}
    since $b = \eps(1 + \eps^{1/3})$. Thus, if $h$ and $L$ satisfy $h + L > \eps + \Theta(\eps^2)$, we obtain our claim.
\end{proof}

We can now prove \Cref{thm:four-dim-lower-bound}. 
\begin{proof}[Proof of \Cref{thm:four-dim-lower-bound}]
    Consider the following spanning graph $G$ on vertex set $P$: it contains a single edge from $s$ to $p_1$, and the edges $(x_i, p_i)$, $(x_i, p_{i + 1})$, $(x_i, u_i)$ and $(x_i, v_i)$ for all $i\in \{1,\ldots ,N - 1\}$. Clearly, $G$ is connected, and its weight is $1 + 4(N - 1)L$. Thus, $\wts(\mst(P)) \leq \wts(G) = 1 + 4(N - 1)L$. 

    Let $T$ be a $(1 + \eps, \beta)$-SLT of $P$. By Claims~\ref{clm:base-base-forbidden2}--\ref{clm:x-base-forbid}, and using a proof similar to that of \Cref{lm:root-base-connection}, we conclude that $T$ must contain edges from $s$ to all tetrahedron points. In other word, $T$ contains the edge $(s, q)$, for all $q \in \Delta_i$ and $i \in \{1, \ldots , N - 1\}$. Since all tetrahedron points are on the unit sphere $\mathbb{S}^4$, this gives a weight lower bound  $\wts(T) \geq 3N + 1$. 

    Therefore, the lightness of $T$ is $\beta = \frac{\wts(T)}{\wts(\mst(P))} \geq \frac{3N + 1}{1 + 4(N - 1)L}$. Recall that $N = \Theta(\eps^{-3/2})$. For sufficiently small $\eps$, we have $\beta \geq \frac{3}{4L}$. The combination of \Cref{eq:con1} and \Cref{clm:x-base-forbid} gives
    \begin{equation*}
        \begin{cases}
            h + L > \eps + \Theta(\eps^2), \\
            h^2 + \frac{3b^2}{8} = L^2 .
        \end{cases}
    \end{equation*}
    By minimizing $L$ subject to these constraints, we obtain $L = \frac{11\eps}{16} + o(\eps^2)$. Thus, by setting $L = \frac{11\eps}{16} + o(\eps^2)$, we get that $\beta > \left(\frac{12}{11} - o(1)\right)\frac{1}{\eps}$. 
\end{proof}

\subsection{High Dimensional Spaces}

In this section, we give lower bounds on the minimum lightness of
$(1+\eps)$-SLTs in high-dimensional $\mathbb{R}^d$, first under the
Euclidean ($\ell_2$) norm, and then under the $\ell_1$ norm.

\subsubsection{Lower Bound in High-Dimensional Euclidean Space}
We first consider the Euclidean ($\ell_2$) norm.
\LowerBoundHigh*

We now generalize the construction from \Cref{ssec:four-dim-strong-bound}. The only change is that in \(\mathbb{R}^d\), instead of using tetrahedra, we use regular simplices with \(d\) vertices on the unit sphere $\mathbb{S}^{d-1}$. The tetrahedral construction corresponds to the case \(d=4\). We construct a point set $P$ as follows. Let $s$ be the origin and consider the unit sphere $\mathbb{S}^{d - 1}$. Similar to the four-dimensional Euclidean case, we construct a curve $C$, which is a curve on the Clifford torus in four dimensions.

\begin{equation*}
    \Gamma(t) = \frac{1}{\sqrt{2}}\left(\cos{t}, \sin{t}, \cos{\frac{t}{\sqrt{\eps}}}, \sin{\frac{t}{\sqrt{\eps}}}, 0, 0, \ldots, 0\right), \qquad t\in [0, \pi].
\end{equation*}

We select a sequence of points $\{p_1, p_2, \ldots p_N\}$ on $C$ such that any two consecutive points are distance $b$ apart, where $b = \eps(1 + \eps^{1/3})$. Note that $N = \Theta(\eps^{-3/2})$. 
Then, we place regular $(d-1)$-dimensional simplices, each containing two consecutive points $p_i$ and $p_{i + 1}$, for $i=1,\ldots , N-1$. Let 
\begin{equation*}
    \Delta_i = \{u_{i, 1} = p_i, u_{i, 2} = p_{i + 1}, u_{i, 3}, u_{i, 4}, \ldots u_{i, d}\}    .
\end{equation*}
We call the vertices of $\Delta_i$, $i=1,\ldots , N-1$, \emph{simplex points}.

We place the simplices consistently along the curve so that the distance between any two simplex points is greater than $\eps$. 
This can be done similarly to the four-dimensional case. Then, for each simplex $\Delta_i$, we find a point $x_i$, slightly outside of the unit sphere $\mathbb{S}^{d - 1}$, such that the distance from $x_i$ to each point in $\Delta_i$ is $L$, to be specified later. 
Our point set is
\begin{equation*}
    P = \{s\} \cup \{p_i\}_{i = 1}^N \cup \{u_{i, j}\}_{1 \leq i \leq N-1, 1 \leq j \leq d}. 
\end{equation*}

Let $o_i$ be the circumcenter of $\Delta_i$. Thus, $x_i$ must lie on the ray through $o_i$. Let $x_i = h\cdot \frac{o_i}{\|o_i\|} + o_i$. The circumradius of $\Delta_i$ is $\sqrt{\frac{d - 1}{2d}} \cdot b$. By the Pythagorean theorem for the right triangle ${o_i p_i x_i}$, we have
\begin{equation}
    \label{eq:highdim-Pythagorean}
    h^2 + \frac{(d - 1)b^2}{2d} = L^2. 
\end{equation}
Using an argument similar to \Cref{clm:base-base-forbidden} and \Cref{clm:x-base-forbid}, we conclude that if
\begin{equation}
    \label{eq:highdim-stretch}
    h + L > \eps + \Theta(\eps^2),
\end{equation}
then any $(1 + \eps, \beta)$-SLT of $P$ with root $s$ must contain edges from $s$ to every simplex point. To estimate the weight of $\mst(P)$, we construct a graph $G$ containing an edge $(s, p_1)$ and all edges $(x_i, u_{i, j})$ for $i = 1,\ldots , N - 1$ and $j = 1,\ldots , d$. Since $G$ is connected, we have
\begin{equation}
    \label{eq:highdim-mst}
    \wts(\mst(P)) \leq \wts(G) = 1 + \sum_{i=1}^{N-1}\sum_{j=1}^{d}\wts(x_i, u_{i, j}) = 1 + (N - 1)d \cdot L.
\end{equation}

Consider any $(1 + \eps, \beta)$-SLT $T$ of $P$ rooted at $s$. Since $T$ contains all edges from $s$ to $u_{i, j}$ for $i = 1,\ldots , N - 1$ and $j = 1,\ldots , d$, we have the following bound on $\wts(T)$:
\begin{equation}
    \label{eq:highdim-slt}
    \wts(T) \geq (d - 1)(N - 1) + 1. 
\end{equation}

From \Cref{eq:highdim-mst} and \Cref{eq:highdim-slt}, we can bound the lightness $\beta$ as
\begin{equation*}
    \beta = \frac{\wts(T)}{\wts(\mst(P))} \geq \frac{(d - 1)(N - 1) + 1}{1 + (N - 1)d \cdot L}.
\end{equation*}
By choosing the minimum $L$ satisfying \Cref{eq:highdim-Pythagorean} and \Cref{eq:highdim-stretch}, we get $L = \frac{3\eps}{4} + o(\eps^2)$. Thus,
\begin{equation*}
    \beta = \frac{(d - 1)(N - 1) + 1}{1 + (N - 1)d \cdot \left(\frac{3\eps}{4} + o(\eps^2)\right)} = \frac{1}{\eps}\left(\frac{4}{3} - o_\eps(1)\right) ,
\end{equation*}
for sufficiently large $d$, as claimed. 

\subsubsection{Lower Bound in High-Dimensional \texorpdfstring{$\ell_1$}{ell1} Space}

In this section, we prove a lower bound for the lightness of shallow-light trees in high-dimensional $\ell_1$ spaces.

\LowerBoundHighellone*

We briefly describe the idea behind our lower bound. In \cite{KRY95}, the authors construct a series-parallel graph such that any graph with root-stretch $1+\eps$ must have lightness $1+2/\eps$. Their construction consists of a root point $s$, a star whose edges have weight slightly larger than $\eps/2$, edges of weight $1$ from $s$ to every leaf of the star, and an edge of weight $1-\eps/2$ from $s$ to the center of the star. Any graph with root-stretch $1+\eps$ must contain all edges from $s$ to the leaves.

Our idea is to use this star construction to obtain a lower bound in $\ell_1^\infty$. In $\ell_2$, such a star with more than two leaves cannot be embedded with the required distances, regardless of the dimension. In contrast, in $\ell_1^d$, one can embed a star with $d$ leaves. For example, $a=(\frac{\eps}{2},\frac{\eps}{2},\ldots,\frac{\eps}{2})$ be the center of the star, and define each leaf to have the same coordinates as $a$, except in one coordinate, where the value is set to $0$. We then place a large number of stars, each with $d$ leaves, around $s$.

We continue with the detailed construction in $\ell_1^d$, for sufficiently large $d\in \mathbb{N}$. Let $\eps = \frac{1}{k}$ for some sufficiently large positive integer $k$, let $\eps'$ be a real number slightly large than $\eps$, which we will define later. First, consider a star with center at $v = (0, 0, 0, \frac{\eps'}{2}, \frac{\eps'}{2}, \ldots, \frac{\eps'}{2})$ and leaves $u_i = v - \frac{\eps'}{2}e_{i + 3}$ for $i = 1, \ldots , d - 3$, where $e_j$ is the $j$th unit vector in an orthonormal basis. Let $S = \{v\} \cup \{u_i\}_{i = 1}^{d - 3}$.
Our point set $P$ consists of the origin $s$ and many translated copies of $S$. 
Specifically, let
$$B = \left\{ \left(x \cdot \frac{\eps'}{2}, y \cdot \frac{\eps'}{2}, z \cdot \frac{\eps'}{2}, 0, \ldots ,0\right) :  x, y, z\in \mathbb{N}_0\,\,\, \mathrm{and} \,\,\, x + y + z = \frac{2}{\eps'} - d + 4.\right\} .$$
Our point set is
\begin{equation*}
   P:= \{s\} \cup  \{S+b\}_{b \in B}.
\end{equation*}

Let $C = \{v+b\}_{b \in B}$ be the set of centers of stars (called \emph{center points}) and $L = \{b + u_i\}_{b \in B}$ the set of leaves of stars (called \emph{leaf points}). By our choice, the distance from $s$ to any leaf point is exactly $1$. We choose $\eps'$ such that $2/\eps'$ is a positive integer greater than $(d - 4)$, specifically:

\begin{equation*}
    \eps' = \frac{2}{\frac{2}{\eps} - 1} = \eps + O(\eps^2). 
\end{equation*}
Observe that the distance between $s$ and any center point is $1 + \frac{\eps'}{2}$.

We show that no other point can be used in an approximate path from $s$ to any leaf point. 

\begin{claim}
    \label{clm:no-leaf-sub}
    For every $q \in L$ and $p \in P \setminus \{s, q\}$, we have
    \begin{equation*}
        \|s - p\|_1 + \|p - q\|_1 > (1 + \eps)\|s - q\|_1.
    \end{equation*}
\end{claim}

\begin{proof}
    Since $\|s - q\|_1 = 1$, we only need to show that $\|s - p\|_1 + \|p - q\|_1 > 1 + \eps$. Assume that $q = b + u_i$ for some $b \in B$ and $i \in \{1,\ldots , d - 3\}$. If $p = b + v$, then by definition of $p$ and $q$, $\|p - q\|_1 = \frac{\eps'}{2}$. Thus,
    \begin{equation*}
        \begin{split}
            \|s - p\|_1 + \|p - q\|_1 = 1 + \frac{\eps'}{2} + \frac{\eps'}{2} = 1 + \eps' > 1 + \eps , 
        \end{split}
    \end{equation*}
    since $\eps' > \eps$. 

    If $p = b + u_{i'}$ for some $i' \neq i$, then
    \begin{equation*}
        \begin{split}
            \|p - q\|_1 &= \|b + u_{i'} - b - u_{i}\|_1 = \|u_{i'} - u_{i}\|_1\\
            &= \left\|v - \frac{\eps'}{2} \cdot e_{i' + 3} - \left(v - \frac{\eps'}{2} \cdot e_{i + 3}\right)\right\|_1 \\
            &= \frac{\eps'}{2}\|e_{i + 3} - e_{i' + 3}\|_1 = \eps'.
        \end{split}
    \end{equation*}
    Consequently, $\|s - p\|_1 + \|p - q\|_1 = 1 + \eps' > 1 + \eps$.

    The last case is when $p$ belongs to a different star, i.e., $p \in b' + S$, for some $b' \neq b$. Observe that $b$ and $b'$ are different in at least $2$ coordinates, since their $\ell_1$ norms are the same and none of them contains any negative coordinate. Therefore, $\|p - q\|_1 \geq \|b - b'\| \geq \eps' > \eps$. Since $\|s - p\|_1 \geq 1$ for all $p \neq s$, we have $\|s - p\|_1 + \|p - q\|_1 > 1 + \eps$. 
\end{proof}

We are ready to prove \Cref{thm:high-dim-lower-bound-ell1}.
\begin{proof}[Proof of \Cref{thm:high-dim-lower-bound-ell1}]
    Let $\eps>0$ be given, and assume w.l.o.g.\ that $\eps=\frac{1}{k}$ for some  $k\in \mathbb{N}$. Consider the point set $P$ constructed above with $\eps'>\eps$ and $d\in \mathbb{N}$ to be specified below. 
    Let $T$ be an $(1 + \eps, \beta)$-SLT of $P$ rooted at $s$. \Cref{clm:no-leaf-sub} implies that $T$ contains the edges $(s, q)$ for all $q \in L$, each of unit weight. Therefore, $\wts(T) \geq |L| \cdot 1 = (d - 3) \cdot |B|$. Let $N = \frac{2}{\eps'} - d + 4=2k-d+3$, and note that $N > 0$ if $\eps>0$ is sufficiently small. Since the number of triples $(x,y,z)\in \mathbb{N}_0^3$ with $x+y+z=N$ is $|B|=\binom{N + 2}{2}$, then
    \begin{equation}
        \label{eq:l1SLT}
        \wts(T) \geq (d - 3) \cdot \binom{N + 2}{2}. 
    \end{equation} 

    In order to give an upper bound for $\wts(\mst(P))$, we construct a spanning graph $G$ on $P$. Observe that the coordinates of any point in $C$ are of the form $\frac{\eps'}{2}(x, y, z, 1, 1, \ldots 1)$ with $x, y, z\in \mathbb{N}_0$ and $x + y + z = N$. Add the edge between $s$ and an arbitrary leaf in $L$ to $G$. For every point $c = \frac{\eps'}{2}(x, y, z, 1, 1, \ldots, 1)$, we define the associated point of $c$ to be $c' = \frac{\eps'}{2}(x, y - 1, z + 1, 1, 1, \ldots 1)$ if $y > 0$, and $c' = \frac{\eps'}{2}(x + 1, 0, z - 1, 1, 1, \ldots 1)$ if $y = 0$ and $z > 0$. It is immediate that $\|p - p'\| = 2\frac{\eps'}{2} = \eps'$. For every $c \in C$, we add the edge $(c, c')$ between $c$ and its associated point $c'$. Finally, add all edges between $c$ and its corresponding leaves. Thus, $G$ contains the following edges:
    \begin{itemize}
        \item An edge between $s$ and a point in $L$, which has weight $1$. 
        \item $|C|$ edges between centers $c\in C$ and their associated points, each of weight at most $\eps'$. The total weight of these edges is $|C| \cdot \eps' = |B| \cdot \eps' = \binom{N + 2}{2} \cdot \eps'$.
        \item $|L|$ edges between leaves in $L$ and their corresponding centers in $C$, each of weight $\frac{\eps'}{2}$. The total weight of these edges is $|L| \cdot \frac{\eps'}{2} = (d - 3) \cdot \binom{N + 2}{2}\cdot \frac{\eps'}{2}$. 
    \end{itemize}
    Overall, we have $\wts(G) 
    \leq 1 + \binom{N + 2}{2} \cdot \eps' + (d - 3) \cdot \binom{N + 2}{2}\cdot \frac{\eps'}{2}$.
     Since $G$ is connected, 
    \begin{equation}
    \label{eq:l1mst}
        \wts(\mst) \leq 1 + \binom{N + 2}{2} \cdot \eps' + (d - 3) \cdot \binom{N + 2}{2}\cdot \frac{\eps'}{2}.
    \end{equation}

    We now bound the value of $\beta$. From \Cref{eq:l1SLT,eq:l1mst}, we have
    \begin{align}
    \beta = \frac{\wts(T)}{\wts(\mst)} 
            &\geq \frac{(d - 3) \cdot \binom{N + 2}{2}}{1 + \binom{N + 2}{2} \cdot \eps' + (d - 3) \cdot \binom{N + 2}{2}\cdot \frac{\eps'}{2}}\nonumber\\
             &= \frac{(d - 3) \cdot \binom{N + 2}{2}}{\frac{2}{\eps'} + 2\cdot \binom{N + 2}{2}  + (d - 3) \cdot \binom{N + 2}{2}} \cdot \frac{2}{\eps'}.\label{eq:l1_beta1}
    \end{align}
    Since $N = \frac{2}{\eps'} - d + 4=\Theta(\frac{1}{\eps})$ if $\eps d<1$, and $\eps' = \eps + O(\eps^2)$, we have 
    $$\frac{2}{\eps'} + 2\cdot \binom{N + 2}{2} 
    = \frac{2}{\eps + O(\eps^2)} +  \Theta\left(\frac{1}{\eps^2}\right) 
    = \Theta\left(\frac{1}{\eps^2}\right) .$$
    Since $(d - 3)\cdot \binom{N + 2}{2} = (d - 3) \cdot\Theta\left(\frac{1}{\eps^2}\right)$, we get
    \begin{equation}
        \label{eq:exta_factor}
        \frac{1 + \binom{N + 2}{2}}{(d - 3)\cdot \binom{N + 2}{2}} = \Theta(1/d). 
    \end{equation}

    Combining \Cref{eq:l1_beta1,eq:exta_factor}, we obtain
    \begin{equation*}
        \beta 
        \geq \frac{1}{\Theta(1/d) + 1} \cdot \frac{2}{\eps'}
        = \left(1 - O\left(\frac{1}{d}\right) \right) \cdot \frac{2}{\eps'} = (1 - O_\eps(1))\cdot \frac{2}{\eps}, 
    \end{equation*}
    as claimed. The last equation holds because $\eps < \frac{1}{d}$ and $\eps' = \eps + O(\eps^2)$.
\end{proof}

\section*{Acknowledgment}
The authors used ChatGPT to assist with redrafting and polishing the exposition of the manuscript, including rewriting portions of the introduction, technical overview, and proof presentation. All definitions, theorem statements, algorithms, proofs, and final text were reviewed, edited, and approved by the authors, who take full responsibility for the correctness, originality, and integrity of the paper.

Numerical calculations in the proof of \Cref{lem:colinear}, \Cref{thm:LB-3D}, \Cref{thm:easy-4d-bound}, and \Cref{thm:four-dim-lower-bound} were assisted by ChatGPT. 

The lower-bound construction of \Cref{thm:high-dim-lower-bound-ell1} builds upon an initial idea developed by the authors. Through an iterative interaction with ChatGPT, this idea was ultimately refined into the final construction, which was subsequently verified and fully formalized by the authors.

Hung Le and Cuong Than are supported by the NSF CAREER award CCF-2237288, the NSF grants CCF-2517033 and CCF-2121952, and a Google Research Scholar Award. Cuong Than is also supported by a Google Ph.D. Fellowship.
Research by Csaba D.\ T\'oth was supported by the NSF award DMS-2154347.
Shay Solomon is funded by the European Union (ERC, DynOpt, 101043159). Views and opinions expressed are however those of the author(s) only and do not necessarily reflect those of the European Union or the European Research Council. Neither the European Union nor the granting authority can be held responsible for them. Shay Solomon is also funded by a grant from the United States-Israel Binational Science Foundation (BSF), Jerusalem, Israel, and the United States National Science Foundation (NSF). Tianyi Zhang is supported by the Fundamental and Interdisciplinary Disciplines Breakthrough Plan of the Ministry of Education of China (No. JYB2025XDXM118) and the ``111 Center'' (No. B26023).

\bibliographystyle{alphaurl}
\bibliography{ref}
\end{document}